\documentclass[11pt]{article}
\usepackage{amssymb,amsthm,subcaption}
\usepackage[font={footnotesize,it}]{caption}
\usepackage[top=0.6in, bottom=0.6in, left=0.6in, right=0.6in]{geometry}

\DeclareCaptionStyle{italic}[justification=centering]{labelfont={bf},textfont={it},labelsep=colon}
\usepackage{graphicx,psfrag,epsf}
\usepackage{enumerate}
\usepackage{natbib}
\usepackage{url,setspace,bbm}
\usepackage{authblk}
\usepackage{booktabs, bm, paralist,mathpazo,tikz,todonotes,longtable,dsfont,rotating} 
\usepackage[pdftex,colorlinks=true]{hyperref}
\usepackage{mathtools}	      % \coloneqq (:=)
\usepackage[title]{appendix}
\usepackage{mathrsfs}
\usepackage{multirow}
\usepackage{prettyref}
\newrefformat{hypo}{H\ref{#1}}
\newcounter{hypothesis}

\definecolor{darkblue}{rgb}{0,0,.6}
\hypersetup{citecolor=darkblue,linkcolor=darkblue,urlcolor=darkblue}
\newcommand{\commHS}[1]{{\leavevmode\color{darkblue}#1}}
\newcommand{\comm}[1]{{\leavevmode\color{darkblue}#1}}
\usepackage{orcidlink}

\newcommand{\xdownarrow}[1]{%
  {\left\downarrow\vbox to #1{}\right.\kern-\nulldelimiterspace}
}
\newcommand{\blind}{0}

\graphicspath{{plots/}}
\DeclareMathOperator*{\argmin}{\arg\!\min}
\newsavebox\CBox

\date{\today}

\theoremstyle{plain}
\newtheorem{theorem}{Theorem}[section]
\newtheorem{lemma}{Lemma}[section]
\newtheorem{proposition}[theorem]{Proposition}

\usepackage{braket}\usepackage{mathtools}\usepackage{tikz-cd}
\newtheorem{assumption}{Assumptions}[section]

\theoremstyle{definition}

\newtheorem{remark}{Remark}[section]

\usepackage{bibentry}

\newcommand{\hL}{\widehat{\boldsymbol{L}}}

\newcommand{\pL}{\boldsymbol{L}}
\newcommand{\pmu}{\boldsymbol{\mu}}
\newcommand{\hPsi}{\boldsymbol{\widehat{\Psi}}}
\newcommand{\tPsi}{\boldsymbol{\widetilde{\Psi}}}
\newcommand{\pPsi}{\boldsymbol{\Psi}}

\newcommand{\hGa}{\boldsymbol{\widehat{\Gamma}}}
\newcommand{\tGa}{\boldsymbol{\widetilde{\Gamma}}}
\newcommand{\pGa}{\boldsymbol{\Gamma}}

\newcommand{\hA}{\boldsymbol{\widehat{A}}}
\newcommand{\tA}{\boldsymbol{\widetilde{A}}}
\newcommand{\pA}{\boldsymbol{A}}

\newcommand{\hS}{\boldsymbol{\widehat{\Sigma}}}
\newcommand{\tS}{\boldsymbol{\widetilde{\Sigma}}}
\newcommand{\pS}{\boldsymbol{\Sigma}}
\newcommand{\hf}{\boldsymbol{\widehat{f}}}
\newcommand{\tf}{\boldsymbol{\widetilde{f}}}
\newcommand{\brf}{\boldsymbol{\breve{f}}}
\newcommand{\pf}{\boldsymbol{f}}
\newcommand{\pfo}{\boldsymbol{f}^o}
\newcommand{\bof}{\boldsymbol{f}^\ast}
\newcommand{\boof}{\boldsymbol{f}^b}
\newcommand{\hQ}{\boldsymbol{\widehat{Q}}}

\newcommand{\pQ}{\boldsymbol{Q}}
\newcommand{\pQo}{\boldsymbol{Q}^o}
\newcommand{\pq}{\boldsymbol{q}}
\newcommand{\pqo}{\boldsymbol{q}^o}

\newcommand{\py}{\boldsymbol{y}}
\newcommand{\by}{\boldsymbol{y}^\ast}

\newcommand{\pc}{\boldsymbol{c}}

\newcommand{\pu}{\boldsymbol{u}}
\newcommand{\hu}{\boldsymbol{\widehat{u}}}
\newcommand{\pv}{\boldsymbol{v}}

\newcommand{\pI}{\boldsymbol{I}}
\newcommand{\pJ}{\boldsymbol{J}}
\newcommand{\be}{\boldsymbol{e}^\ast}
\newcommand{\booe}{\boldsymbol{e}^b}
\newcommand{\pe}{\boldsymbol{e}}
\newcommand{\te}{\boldsymbol{\widetilde{e}}}
\newcommand{\he}{\boldsymbol{\widehat{e}}}
\newcommand{\bre}{\boldsymbol{\breve{e}}}

\newcommand{\pM}{\boldsymbol{M}}

\newcommand{\pPi}{\boldsymbol{\Pi}}
\newcommand{\tPi}{\boldsymbol{\widetilde{\Pi}}}
\newcommand{\hPi}{\boldsymbol{\widehat{\Pi}}}

\newcommand{\bdel}{\delta^\ast}
\newcommand{\pdel}{\delta}

\newcommand{\pone}{\boldsymbol{1}}

\newcommand\numberthis{\addtocounter{equation}{1}\tag{\theequation}}

\usepackage[nopatch=footnote]{microtype}

\begin{document}

\def\spacingset#1{\renewcommand{\baselinestretch}
{#1}\small\normalsize} \spacingset{1}

\if0\blind
{
  \title{\bf AR-sieve Bootstrap for High-dimensional Time Series}
\author[a]{Daning Bi \orcidlink{0000-0001-9825-7290}
\thanks{Corresponding to: College of Finance and Statistics, Hunan University, 1 Lushan South Road, Changsha, Hunan, 410082, China; Email: daningbi@hnu.edu.cn}}
\author[b]{Han Lin Shang \orcidlink{0000-0003-1769-6430}}
\author[c]{Yanrong Yang \orcidlink{0000-0002-3629-5803}} 
\author[d]{Huanjun Zhu \orcidlink{0000-0003-0575-4525}}
	
\affil[a]{Hunan University}
\affil[b]{Macquarie University}
\affil[c]{Australian National University}
\affil[d]{Xiamen University}

\maketitle
} \fi

\if1\blind
{
   \title{\bf AR-sieve Bootstrap for High-dimensional Time Series}
   \author{}
   \maketitle
} \fi

\begin{abstract}
This paper proposes a new AR-sieve bootstrap approach to high-dimensional time series. The major challenge of classical bootstrap methods on high-dimensional time series is two-fold: curse \commHS{of} dimensionality and temporal dependence. To address such a difficulty, we utilize factor modeling to reduce dimension and capture temporal dependence simultaneously. A factor-based bootstrap procedure is constructed, which performs an AR-sieve bootstrap on the extracted low-dimensional common factor time series and then recovers the bootstrap samples for the original data from the factor model. Asymptotic properties for bootstrap mean statistics and extreme eigenvalues are established. Various simulation studies further demonstrate the advantages of the new AR-sieve bootstrap in high-dimensional scenarios. An empirical application on particulate matter (PM) concentration data is studied, where bootstrap confidence intervals for mean vectors and autocovariance matrices are provided. 
\end{abstract}

\noindent%

\spacingset{1.0}

\section{Introduction}\label{1:intro}

% Intro to Bootstrap Progress in Literature 

The bootstrap is a computer-intensive resampling-based methodology that arises as an alternative to asymptotic theory. The bootstrap method, initially introduced by \cite{efron_bootstrap_1979} for independent sample observations, was later extended to more complicated dependent data in the literature. As an important extension to stationary time series, blockwise bootstrap \citep{kunsch_jackknife_1989}, autoregressive (AR) sieve bootstrap \citep{kreiss_asymptotic_1988, buhlmann_sieve_1997}, and frequency-domain bootstrap \citep{Franke_Hardle1992, Dahlhaus_Janas1996} have received the most discussions and developments in the past few years. Several variants of block bootstrap methods have appeared, such as the block bootstrap for time series with fixed regressors \citep{NordmanLahiri2012}, the double block bootstrap \citep{LeeLai2009}, and the stationary bootstrap \citep{PolitisRomano1994}, among others. An apparent disadvantage of the blockwise bootstrap is the neglected dependence between different blocks. The AR-sieve bootstrap method takes up the ``sieve'' strategy, which approximates a stationary time series by an AR model with a large number of time lags. Compared with the blockwise bootstrap, the AR-sieve bootstrap samples are conditionally stationary and keep the dependence structure well. The AR-sieve bootstrap was introduced by \citet{kreiss_asymptotic_1988} and has been well studied from stationary linear processes \citep{buhlmann_sieve_1997} to strictly stationary time series that fulfill a general moving average MA ($\infty$) representation \citep{kreiss_range_2011}. After this work, the theoretical requirement and validity of a general AR-sieve bootstrap method for certain types of statistics have been discussed for univariate \citep{kreiss_range_2011}, multivariate \citep{meyer_vector_2015}, and functional time series \citep{paparoditis_sieve_2018, PaparoditisShang2021}, respectively. The frequency-domain bootstrap to implement the resampling schemes is based on frequency-domain methods, which are motivated by the observation that periodogram ordinates at a finite number of frequencies are approximately independently distributed so that Efron's ideas may be employed. Compared to the AR-sieve bootstrap, this method could deal with more general dependence structures for time series \citep{MeyerPaparoditisKreiss2020, Hidalgo2021}. 

%%%%%%%%%%%%%%%%%%Bootstrap on High-dimensional Time Series%%%%%%%%%%%%%
The main goal of this paper is to extend the AR-sieve bootstrap to high-dimensional time series. Due to the curse of dimensionality, the traditional AR-sieve bootstrap fails in the high-dimensional case. This is because the AR model approximation for high-dimensional time series could result in a large approximation error, and the bootstrap procedure on high-dimensional independent and identically distributed (i.i.d.)\ residuals is also inaccurate. The curse of dimensionality on traditional bootstrap methods is demonstrated vividly in \citet{el_karoui_can_2018}. As a remedy, reducing the parameter space is essential for successfully modifying bootstrap methods. Fitting sparse models and low-rank models to high-dimensional data is one of the commonly used techniques to eliminate the curse of dimensionality. \citet{Chernozhukov_Chetverikov_Kato2017} provide a theoretical guarantee on the bootstrap approximation for the distribution of the sample mean vector for high-dimensional i.i.d.\ data. \citet{Chen2018} studies the bootstrap approximation for $U$ statistics constructed with high-dimensional i.i.d data. \citet{Ahn_Reinsel1988} propose a nested reduced-rank structure for coefficients in multivariate AR time series models. For high-dimensional time series, \citet{zhang_bootstrapping_2014} study bootstrap inference for weakly dependent time series based on a general Gaussian approximation theory, and \citet{krampe_bootstrap_2019} consider the AR-sieve bootstrap for vector AR time series with sparse coefficients.  In this article, we will contribute to proposing an appropriate low-rank model for the AR-sieve bootstrap on high-dimensional time series. 

%%%%%%%%%%%%%%%%%%Method: Factor Modelling%%%%%%%%%%%%%%%%%%%%%%%
Factor modeling or low-rank representation can project high-dimensional data into a low-dimensional subspace. Principal component analysis (PCA) is a common technique for pursuing projections or subspaces with the most variation in the original data \citep{bai_determining_2002, fan_high-dimensional_2011}. Identifying a low-dimensional representation for high-dimensional time series is more complicated because keeping the temporal dependence in dimension reduction is a crucial requirement. The earlier literature on multivariate time series in this field is vast and includes canonical correlation analysis \citep{Box_Tiao1977}, factor models \citep{Pena_Box1987}, and a scalar component model \citep{Tiao_Tsay1989}. Later, \citet{lam_estimation_2011} studied a factor model for high-dimensional time series based on an accumulation of autocovariance matrices, aiming to capture all temporal dependence by common factors.  

%%%%%%%%%%%%%%%%%%Main Contribution%%%%%%%%%%%%%%%%%%%%%%%%%%%%%%%
In this article, we reduce high-dimensional time series based on a factor model whose common factors possess all the temporal dependence of the original time series. Efficient estimation for such a factor model is borrowed from the idea of \cite{lam_estimation_2011}, which conducts eigen-decomposition for a set of autocovariance matrices with various time-lags. \comm{However, it is important to distinguish our contribution from this foundational work. While \cite{lam_estimation_2011} established the estimation theory (consistency and convergence rates), they did not address the problem of statistical inference, such as constructing confidence intervals. Our work fills this gap by utilizing their estimation procedure as a building block to develop the theoretical validity of the AR-sieve bootstrap for uncertainty quantification in high-dimensional time series.} With lower-dimensional common factor time series, the AR-sieve bootstrap is feasible and produces bootstrap samples for common factors. Finally, the AR-sieve bootstrap could recover the relationship between common factors and the original high-dimensional time series. 

%%%%%%%%%%%%%%%%%%%%%%%%%%%%%%%%%%%%%%%%%%%%%%%%%%%%%%%
We also study the theoretical properties of the proposed AR-sieve bootstrap on two commonly used statistics - the mean statistics and spectral statistics of autocovariance matrices. The common factors are in a ``representation and activation position'' in the entire bootstrap method. Under the scenario of comparable $N$ (dimension) and $T$ (time-serial length), we first provide convergence rates for the estimation of common factors, which could affect the statistical properties of the final AR-sieve bootstrap statistics. Furthermore, for the two high-dimensional statistics under consideration, the consistency of the bootstrap versions with the population versions is established. Finite-sample experiments demonstrate the influence of the dimension, the sample size, and the factors' strength on the bootstrap results. Moreover, we also performed an empirical application on $\text{PM}_{10}$ data. As a by-product of interest, we apply the proposed AR-sieve bootstrap for high-dimensional series on sparsely observed discrete functional time series and compare them with the results from the AR-sieve bootstrap for functional time series \citep{paparoditis_sieve_2018}. Due to the smoothing inaccuracy for sparsely observed discrete functional time series, the high-dimensional bootstrap method sometimes results in better statistical inferences than the functional approach. Various simulations in Section~\ref{simu2} and the Supplementary Material could reflect this point. 

%%%%%%%%%%%%%%%%%%Organization of this Paper%%%%%%%%%%%%%%%%%%%%%%%

\comm{The remainder of this paper is organized as follows. Section~\ref{1:sec:2} introduces factor models for high-dimensional time series and discusses the AR representation of factor time series, a building block of the general AR-sieve bootstrap. In Section~\ref{1:est}, the estimation procedure for factor models and the AR-sieve bootstrap procedure for factor time series are introduced with regularity conditions in factor models. The additional assumptions and asymptotic validity of our novel AR-sieve bootstrap method for the mean statistics of factor time series and spiked eigenvalues of symmetrized autocovariance matrices are discussed in Section~\ref{1:sec:4}. %In Section~\ref{1:sim}, via several simulation experiments, we verify the validity of our novel AR-sieve bootstrap methods on the mean statistics. 
Section~\ref{1:sec:6} provides an example of applying our novel AR-sieve bootstrap method to $\text{PM}_{10}$ data, and Section~\ref{1:sec:7} concludes the paper. The technical proofs of the theorems are located in Appendix~\ref{1:sec:Appendix_A}, whereas discussions for assumptions, auxiliary lemmas and their proofs, additional simulations, and further applications on sparsely observed functional time series are included in Appendices~\ref{assu}~to~\ref{appendix0} of the Supplementary Material, respectively.
}
%are left in Appendix~\ref{1:sec:Appendix_B} in the Supplementary Material. simulation studies on the validity of bootstrapping mean statistics and the spiked eigenvalues of symmetrized autocovariance matrices, and further applications are included in the Supplementary . %In the supplementary, Appendix \ref{appendix0} explores the impact of the density of discrete functional time-series observations on pre-smoothing results. Technical proofs and auxiliary lemmas are presented in Appendices~\ref{1:sec:Appendix_A} and~\ref{1:sec:Appendix_B} in additional supplementary documents.

%%%%%%%%%%%%%%%%%%%%%%%%%%%%%%%%%%%%%%%%%%%%%%%%%%%%%%%
%In this work, we propose a novel sieve bootstrap method for high-dimensional time series using a factor model approach.

\section{Factor-based AR-sieve representation}\label{1:sec:2}

We first propose a factor model to project the high-dimensional time series into a lower-dimensional subspace, where the common factor time series could represent the original data to capture the most temporal dependence. Secondly, an AR-sieve representation for common factors is provided, which plays a significant role in the AR-sieve bootstrap. 

Consider a stationary $N$-dimensional time series $\{\py_{t} \in \mathbb{R}^{N},\ t \in \mathbb{Z}\}$ following a general unobservable factor model, given as
\begin{equation} \label{1:e1}
\py_t = \pQ \pf_t + \pu_t, \quad t=1, 2, \ldots, T,
\end{equation}
where $\{\pf_t \in \mathbb{R}^{r},\ t \in \mathbb{Z}\}$ are unobserved $r$-dimensional factor time series, $\pQ$ is an $N\times r$ factor loading matrix, and $\{\pu_t \in \comm{\mathbb{R}^
N},\ t \in \mathbb{Z}\}$ are $N$-dimensional white noise with zero means and covariance matrix $\pS_u$. Factor models have received numerous discussions, and there are various identification conditions and assumptions on $\pQ$, $\pf_t$, and $\pu_t$ depending on various objectives. In our work, we adapt the identification condition in \citet{lam_estimation_2011} to consider a factor model where the temporal dependence of $\{\py_t\}$ can be fully captured by the factors $\{\pf_t\}$. 

Then, we introduce an AR-sieve representation for multivariate common factor time series. 
For the $r \times 1$ common factors $\{\pf_t\}$, we know via Wold's theorem \citep[see, e.g.,][]{buhlmann_sieve_1997} that $\{\pf_t\}$ can be written as a one-sided infinite-order moving-average (MA) process 
\begin{align}\label{1:e3}
\pf_t = \sum_{l=1}^\infty \pPsi_{l} \pe_{t-l} + \pe_{t},\ t \in \mathbb{Z},
\end{align}
where $\{\pe_{t} \in \mathbb{R}^r,\ t \in \mathbb{Z}\}$ are full rank uncorrelated white-noise innovation processes with $\mathbb{E} (\pe_t) = 0 $ and  $\mathbb{E} (\pe_t \pe_s^\top) = \mathbf{1}_{t=s} \pS_e$, with $\pS_e$ a full rank $r \times r$ covariance matrix. $\{\pPsi_{l} \in \mathbb{R}^{r \times r},\ l \in \mathbb{N}\}$ are the coefficients matrices. Under the requirement on invertibility of the process in~\eqref{1:e3}, which would narrow the class of stationary processes a little, we can represent $\{\pf_t\}$ as a one-sided infinite-order autoregressive (AR) process. That is, there exists an infinite sequence of $r \times r$ matrices $\{\pA_{l} \in \mathbb{R}^{r \times r},\ l \in \mathbb{N}\}$ such that factors $\{\pf_t\}$ can be expressed as 
\begin{align}\label{1:e2}
\pf_t = \sum_{l=1}^\infty \pA_{l} \pf_{t-l} + \pe_{t},\ t \in \mathbb{Z},
\end{align}
where the coefficient matrices of the expansion for the power series $\left( \boldsymbol{I}_r -  \sum_{l=1}^{\infty} \pA_{l} z^l \right)^{-1}$ are $\{\pPsi_{l} \in \mathbb{R}^{r \times r},\ l \in \mathbb{N}\}$. Here $|z| \le 1$ \citep{brockwell_time_1991}. Note that~\eqref{1:e3} is a representation instead of an imposed assumption or model. The AR-sieve bootstrap is based on an approximated AR representation for~\eqref{1:e2}, i.e., 
\begin{eqnarray}\label{new1}
\pf_t \approx \sum_{l=1}^p \pA_{l} \pf_{t-l} + \pe_{t},\ t \in \mathbb{Z},
\end{eqnarray}
where $p$ is a large integer that tends to infinity. The AR-sieve bootstrap is a nonparametric approach, although~\eqref{new1} looks like a ``fake'' parametric model.   

The (vector) AR representation in~\eqref{1:e2} is more attractive for statistical applications and has received more attention since it relates $\pf_t$ to its past values. The AR-sieve bootstrap, on the other hand, utilizes the AR representation in~\eqref{1:e2} to generate bootstrap common factors by resampling from the de-centered innovations. In practice, since neither factors $\{\pf_t\}$ nor their loadings $\pQ$ are observable, the AR-sieve bootstrap is performed on estimates of $\{\pf_t\}$ rather than true factors. Hence, we will introduce the estimation and bootstrap procedure in the following section. 

\section{Factor-based AR-sieve bootstrap} \label{1:est}

We first introduce the estimation approach for the factor model in~\eqref{1:e1} and then provide the AR-sieve bootstrap procedure for high-dimensional time series. %Further, a flow chart is shown to summarize the whole procedure of the AR-sieve bootstrap for high-dimensional time series.  

\subsection{Analysis on common-factors estimation}

Recall that common factors $\{\pf_t\}$ in model~\eqref{1:e2} are assumed to contain all the temporal dependence of $\{\py_t\}$ because the error components $\{\pu_t\}$ have no temporal dependence. As analyzed by \citet{bathia_identifying_2010} and \citet{lam_estimation_2011}, the factor loading space, that is, the $r$-dimensional linear space spanned by the columns of the factor loading matrix $\pQ$, denoted by $\mathcal{M}(\pQ)$, is uniquely defined. Furthermore, this subspace $\mathcal{M}(\pQ)$ is spanned by the eigenvectors of an accumulated symmetrized autocovariance matrix below, corresponding to its nonzero eigenvalues, 
\begin{equation*}
\pL=\sum_{k=1}^{k_0} \pGa_y(k)\pGa_y(k)^\top,
\end{equation*}
where $\pGa_y(k) = \text{Cov}(\py_t,\py_{t+k})$ is the autocovariance of $\{\py_t\}$ in lag $k$. $k_0$ is a prescribed small integer. Intuitively speaking, the matrix $\pL$ collects the temporal dependence of $\{\py_t\}$ by combining the information contained in the first $k_0$-lags of autocovariance with the squared (symmetrized) form, facilitating the spectral decomposition on~$\pL$.

\begin{remark}
The reason why we do not consider the covariance matrix $\pS_y$ in $\pL$ is straightforward. %As discussed by \citet{lam_estimation_2011}, f
For the factor model~\eqref{1:e1}, $\pS_y = \pGa_y(0) = \pQ \pGa_{\pf}(0) \pQ^\top + \pS_u$, where $\pGa_{\pf}(0)$ is the covariance matrix of \comm{$\{\pf_t\}$} and $\pS_u$ is the covariance matrix of $\{\pu_t\}$. Hence, excluding $\pS_y$ from $\pL$ can filter out the impact of covariance on $\{\pu_t\}$, especially for $N \to \infty$.
\end{remark}

Then it is straightforward to use the spectral (eigenvalue) decomposition on $\pL$ to estimate the factor loading matrix $\pQ$ and the factors $\{\pf_t\}$. Before we discuss the details of the estimation procedure, we first summarize the assumptions and identification conditions for the factor model defined in~\eqref{1:e1}. 
\begin{assumption} [Conditions on factor models] \label{1:c1}
For factor models~\eqref{1:e1}, we suppose that
\begin{enumerate}
\item[(i)] $\{\pf_t\}$ are strictly stationary time series with $\mathbb{E}\pf_t = \mathbf{0}$ and $\mathbb{E}\left\|\pf_t\right\|^2 < \infty$; $\{\pu_t\} \sim WN(0,\pS_u)$ are uncorrelated white noises with covariance matrix $\pS_u$, and all eigenvalues of $\pS_u$ are uniformly bounded as $N \to \infty$; $\{\pf_t\}$ are independent of $\{\pu_s\}$ for any $t, s \in \mathbb{Z}$.
\item[(ii)] \commHS{$\frac{1}{N}\pQ^\top\pQ = \pI_r$,  and for a prescribed small integer $k_0 >0$, the $r\times r$ accumulated autocovariance matrix of the factors, defined as $\pM_f = \sum_{k=1}^{k_0} \pGa_f(k)\pGa_f(k)^\top$ is a diagonal matrix with distinct positive entries $\infty>\lambda_1(\pf) > \lambda_2(\pf) > \cdots > \lambda_r(\pf) >0$ as $N \to \infty$.}
    %\pGa_f (k)=\text{Cov}(\pf_t,\pf_{t+k})$ is of full rank for all $k=0,1,...,k_0$, i.e. the eigenvalues $\{\plam_i(\pf),\ i=1,2,...,r\}$ of the matrix satisfy $\infty>\lambda_1(\pf) \ge \lambda_2(\pf) \ge \cdots \ge \lambda_r(\pf) >0$ as $N \to \infty$.
\item[(iii)] $\{\py_t\}$, therefore $\{\pf_t\}$, are $\psi$-mixing with the mixing coefficients $\psi(\cdot)$ satisfying the condition that $\sum_{t\ge 1} \psi(t)^{1/2} < \infty$, and $\mathbb{E}{|y_{j,t}|^4} < \infty$ elementwisely.
\end{enumerate}
\end{assumption}

\commHS{
\begin{remark}[Identification strategy]
Factor models are subject to rotational indeterminacy; that is, the model $\py_t = \pQ \pf_t + \pu_t$ is observationally equivalent to $\py_t = (\pQ\mathbf{S})(\mathbf{S}^{-1}\pf_t) + \pu_t$ for any invertible matrix $\mathbf{S}$. To ensure that the estimated factors ${\hf}_t$ converge to a unique target $\pf_t$ (up to sign and not merely a rotation thereof), we impose identification conditions consistent with our estimation method.
Since our estimator relies on the spectral decomposition of the sample accumulated autocovariance matrix (which produces orthogonal eigenvectors), we essentially define the true factors $\pf_t$ as the specific rotation of the latent space such that the factor loadings are orthogonal ($\frac{1}{N}\pQ^{\top}\pQ=\pI_{r}$); and the factors' temporal dependence structure (accumulated autocovariance $\mathbf{M}_f$) is diagonal. Combined with the distinct eigenvalue condition in Assumption~3.1(ii), these constraints uniquely identify $\pf_t$ and $\pQ$ up to a sign change (column permutation is fixed by the ordering of eigenvalues). This ensures the one-to-one correspondence required for the asymptotic validity of the bootstrap procedure on $\pf_t$, which is a prerequisite for establishing the consistency of the bootstrap procedure on the estimated factors.
\end{remark}

Additional discussions and justifications regarding to Assumption~\ref{1:c1} are left to Appendix~\ref{assu} in the supplementary. %To further explain the use of Assumption~\ref{1:c1} with the estimation procedure, first notice that $\{\pf_t\}$ are strong factors and no linear combinations of the components of $\{\pf_t\}$ are white noises (WN) as implied by Assumption~\ref{1:c1} (ii).  Recall that $\pL$ is non-negative definite and can be represented as in~\eqref{1:e13} with $\frac{1}{N} \pQ^\top \pQ = \pI_r $. Since $\sum_{k=1}^{k_0} \pGa_{\pf}(k) \pGa_{\pf}(k)^\top$, the middle part of~\eqref{1:e13}, is symmetric, we can apply spectral decomposition on it and recognize $\pL$ as $N \pQ \left( \pU  \pD \pU^\top \right) \pQ^\top$, where $\pU$ is an $r \times r$ orthonormal matrix and $\pD$ is an $r \times r$ diagonal matrix. Furthermore, since $\frac{1}{N} (\pQ\pU)^\top (\pQ\pU) = \pI_r$, the columns of $\pQ\pU$ are in fact the eigenvectors of $\pL$ corresponding to those non-zero eigenvalues scaled by $\sqrt{N}$. Therefore, we can treat $\pQ\pU$ as $\pQ$ for inferences' purpose and estimate $\pQ$ and $\pf_t$ based on the spectral decomposition of $\pL$.
Under regular conditions in Assumptions~\ref{1:c1}, we can estimate the factors and their loadings, and then generate a sample time series using the AR-sieve bootstrap. To facilitate the estimation process, we define $\pQo = \frac{1}{\sqrt{N}} \pQ$ as the (unscaled) orthonormal factor loading matrix such that $\pQ^{o\top} \pQo = \pI_r$ and $\pfo_t$ as the scaled factors such that $\py_t = \pQo \pfo_t + \pu_t$ is equivalent to model~\eqref{1:e1}, but with different scaling on $\pQ$ and $\{\pf_t\}$. Note that since standard eigenvectors have unit norm, $\pQ = \sqrt{N}\pQ^o$ satisfies the identification condition $\frac{1}{N}\pQ^\top \pQ = \pI_r$.} Details of the proposed method, including the estimation and the bootstrap procedure, are illustrated in the following subsection.

\subsection{The procedure of factor-based AR-sieve bootstrap}

We divide the whole procedure for the proposed factor-based AR-sieve bootstrap %. Then, a flow chart is provided to clarify the essential idea of this procedure further. Below, we divide the proposed method 
into four steps.\\ 
\textbf{Step 1:} Estimation of $\pQ$:

We first define the accumulation of symmetrized sample autocovariance up to lag $k_0$ as $\widetilde{\pL}=\sum_{k=1}^{k_0}  \tGa_y(k)  \tGa_y(k)^\top$, with $\tGa_y(k)$ the sample autocovariance at lag $k$ defined as $\tGa_y(k)=\frac{1}{T-k}\sum_{t=1}^{T-k} (\py_{t+k}-\overline{\py}) (\py_{t}-\overline{\py} )^\top$. Applying spectral (eigenvalue) decomposition on $\widetilde{\pL}$, we can obtain  $\widehat{\pQo}=(\widehat{\pq^o_1},\widehat{\pq^o_2},\dots,\widehat{\pq^o_r})$ with $\widehat{\pq^o_i}$ the eigenvector of $\widetilde{\pL}$ corresponding to the $i$\textsuperscript{th} largest eigenvalue of $\widetilde{\pL}$. $\widehat{\pQo}$ is then a natural estimator of the unscaled loading matrix $\pQo$.  And by scaling up $\widehat{\pQo}$ with $\sqrt{N}$, the square root of the dimension, we ended up with $\hQ = \sqrt{N}\widehat{\pQo}$ as the estimator of $\pQ$.

\textcolor{darkblue}{The existing t suggests that the estimation results are robust to the choice of $k_0$ \citep[see, e.g.,][]{lam_estimation_2011, zhang_factor_2024},} and the numerical results associated with $k_0=1$ to $k_0=5$ are similar. In general, when the dimension $N$ is large compared to $T$, a relatively larger $k_0$ may be considered to better calculate sample estimates, while $k_0 = 1$ is computationally more efficient when the sample size $T$ is large compared to the dimension $N$. Besides, for finite samples, some of the non-spiked eigenvalues of $\widetilde{\pL}$ may not be exactly zero; therefore, we can use the ratio-based estimator of \textcolor{darkblue}{\cite{lam_estimation_2011} and \cite{zhang_factor_2024}} %as discussed in \citet{lam_estimation_2011, zhang_factor_2024} to estimate the number of factors $r$. The ratio-based estimator 
defined as $\widehat{r} = \argmin_{1 \le j \le R} \widehat{\lambda}_{j+1}/\widehat{\lambda}_{j}$, with $\widehat{\lambda}_{1} \ge \widehat{\lambda}_{2} \ge \cdots \ge \widehat{\lambda}_{N}$ the eigenvalues of $\widetilde{\pL}$ and $R$ an integer that satisfies $r \le R < N$, to estimate the number of factors $r$. And practically, $R$ can be taken as $N/2$ or $N/3$ for the efficiency of the computation.\\
\textbf{Step 2:} Estimation of $\{\pf_t\}$:	

With $\widehat{\pQ}$ the estimator of $\pQ$, it is then straightforward to estimate $\{\pf_t\}$ by $\widehat{\pf_t}=\commHS{\frac{1}{N}} \widehat{\pQ}^\top \py_t.$\\
\textbf{Step 3:} AR-sieve bootstrap on $\{\hf_t\}$:

To apply the AR-sieve bootstrap on $\{\hf_t\}$, we can first fit a $p$\textsuperscript{th} order VAR model on the $r$-dimensional time series $\{\hf_t\}$ as
\begin{equation*}
\hf_t=\sum_{l=1}^{p} \hA_{l,p}(r) \hf_{t-l} + \he_{t,p},\quad  t=p+1,p+2,\dots,T,
\end{equation*}
where $\he_{t,p}$ denotes the residuals and \comm{the order $p$ of the VAR model can be selected based on an information criterion such as AIC \citep{akaike_new_1974} and SC \citep{schwarz_estimating_1978}.} 

Equivalently, we have $\he_{t,p}= \hf_t - \sum_{l=1}^{p} \hA_{l,p} (r) \hf_{t-l},\ t=p+1,p+2,...,T$, where $\{\hA_{l,p},\ l=1,2,...,p;\ t=p+1,p+2,...,T\}$ are Yule-Walker estimators of the AR coefficient matrices. We can then generate $\{\be_t\}$, the bootstrap sample of residuals, by resampling from the empirical distribution of the centered residual vectors. Consequently, based on the idea of an AR-sieve bootstrap \citep[see, e.g.][]{kreiss_bootstrap_1992, meyer_vector_2015, paparoditis_sieve_2018}, we can generate the $r$-dimensional pseudo-time series $\{\pf_t^\ast,\ t=1,2,...,T\}$ by simulating the VAR model with bootstrap residuals $\{\be_t\}$. Therefore, an AR-sieve bootstrap sample of $\{\pf_t^\ast\}$ is generated by
$\pf_{t}^{\ast}=\sum_{l=1}^{p} \hA_{l,p} (r) \pf_{t-l}^{\ast} + \pe_{t}^{\ast}$, where $\{\pe_{t}^\ast\}$ are i.i.d.\ random vectors following the empirical distribution of the centered residual vectors $\{\widetilde{\pe}_{t}\}$, where $\widetilde{\pe}_{t,p} = \he_{t,p} - \overline{\he}_{{T'},p}$ and $\overline{\he}_{{T'},p} = 1/(T-p)\sum_{t=p+1}^{T}\he_{t,p}$.\\
\textbf{Step 4:} Generating bootstrap data $\{\by_t\}$:

Lastly, the bootstrap time series $\{\py_t^\ast\}$ can be constructed as
\begin{equation}\label{new3}
\py_t^\ast=\sum_{j=1}^{r} \pf_{j,t}^\ast \widehat{\pq_j},
\end{equation}
where $\widehat{\pq_j}=\sqrt{N}\widehat{\pqo_j}$ is the scaled eigenvector of $\hL$ corresponding to the $j$\textsuperscript{th} largest eigenvalue. Following this AR-sieve bootstrap procedure, pseudo-time series $\{\by_{t}\}$ can mimic the temporal dependence of the original data $\{\py_t\}$ using a factor model.

\begin{remark}
It should be noted that the bootstrap version in~\eqref{new3} is constructed using the bootstrap version of the common factors. We could also modify it to involve an additional term related to the error components. For example, with the estimate $\hu_t = \py_t - \hQ \hf_t$, under some regular sparse conditions in the population covariance matrix $\boldsymbol{\Sigma}_u$, we can obtain an appropriate estimator $\widehat{\Sigma}_u$ \citep{fan_large_2013}. Then, a modified bootstrap version is 
\begin{equation}\label{new4}
\py_t^{\ast\ast}=\sum_{j=1}^{r} \pf_{j,t}^\ast \widehat{\pq_j}+\widehat{\boldsymbol{\Sigma}}_u^{1/2}\widetilde{\pu}_t,
\end{equation}
where $\widetilde{\pu}_t$ is $N$-dimensional random vector generated from the standard normal distribution $\mathcal{N}\ (\textbf{0}_N, \boldsymbol{I}_{N})$. In this way, the bootstrap version $\py_t^{\ast\ast}$ is not of low rank. For instance, conditional on the original sample observations, the covariance matrix of $\py_t^{\ast\ast}$ is of full rank. Due to the high-dimensionality of error components, $\{\hu_t\}$, non-parametric bootstrap on error would again incur the curse of dimensionality again \citep{el_karoui_can_2018}.

For simplicity, we study the mean statistics and the largest eigenvalue of sample autocovariance matrices based on the bootstrap version in~\eqref{new3}, because~\eqref{new3} and~\eqref{new4} produce bootstrap statistics with similar asymptotic properties. 
\end{remark}

% To clarify the four steps mentioned above, a flow chart is provided in Figure \ref{new2}, which summarises the basic logic and procedure for the proposed method.  
% \begin{figure}[!htb]
% \centering
% 	%\includegraphics[angle=90]{fig0.pdf}
% \includegraphics{fig0.pdf}
% \caption{Flow chart for AR-sieve bootstrap method}\label{new2}
% \end{figure} 

\section{Asymptotic theory}\label{1:sec:4}

Some regular assumptions are first introduced. Then, we establish the asymptotic properties for two commonly used statistics: the mean statistics and the largest eigenvalues of the accumulated autocovariance matrices.

\subsection{Regular assumptions}\label{1:reg}

Before introducing the additional regularity assumptions, we first fix some notation. We use $\|\cdot \|_2$ to denote the $L_2$ norm (also known as the spectral norm or operator norm) of a matrix or vector, and $\| \cdot \|_F$ to denote the Frobenius norm of a matrix. We use $a \asymp b$ to denote the case that $a = O_P(b)$ and $b = O_P(a)$.

In addition to Assumptions~\ref{1:c1} made on the factor model~\eqref{1:e1}, to apply the AR-sieve bootstrap on $\{\hf_t\}$, the estimates of factors $\{\pf_t\}$, we also need some regularity conditions on $\{\pf_t\}$ for the AR-sieve bootstrap to be consistent and valid. Denoted by $W(\cdot)$, the spectral density matrix of a vector process for all frequencies $\omega \in (0,2\pi]$, then the spectral density matrix of $\{\pf_t\}$ can be defined as
\begin{align*}
W_{\pf} (\omega) = \frac{1}{2\pi} \sum_{k=-\infty}^{\infty} \pGa_{\pf}(k) e^{-i \omega k}, \omega \in (0,2\pi].
\end{align*}
\begin{assumption}\label{1:a2}
In model~\eqref{1:e1}, we strengthen the Assumption~\ref{1:c1} such that $\{\pf_t\} $ are strictly stationary and purely nondeterministic stochastic processes of full rank with $\mathbb{E}\pf_t = \mathbf{0}$ and $\mathbb{E}\left\|\pf_t\right\|^2 < \infty$. $\pGa_f(k)$, the autocovariance matrix of $\pf_t$ at lag $k$ satisfies the matrix norm summability condition $\sum_{k=-\infty}^{\infty} (1+|k|)^\gamma \left\| \pGa_f(k) \right\|_F < \infty$ for some $\gamma \ge 0$ that will be specified later. 
\end{assumption}

\begin{lemma}\label{1:l5}
Let $\sigma_j(\omega)$ be the $j$\textsuperscript{th} largest eigenvalue of the spectral density matrix $W_{\pf} (\omega)$ for $\{\pf_t\}$, $j=1,2,...,r$, $\omega \in (0,2\pi] $. Under Assumption~\ref{1:c1} and \ref{1:a2} with $\gamma=0$, \comm{$\sigma_j(\omega)$} fulfills the following so-called boundedness condition \citep{wiener_prediction_1958}:
\begin{equation*}
c \le \sigma_j(\omega) \le d,\ \text{for all } \omega \in (0,2\pi], 0 < c \le d < \infty.
\end{equation*}
\end{lemma}

The continuity and boundedness properties in Lemma~\ref{1:l5} then entail the existence of a vector AR representation for any vector process satisfying Assumption~\ref{1:a2} \citep[see, e.g.][]{meyer_vector_2015, cheng_baxters_1993, wiener_prediction_1958}. That is, the AR representation~\eqref{1:e2} and the Wold representation~\eqref{1:e3} are valid under Assumption~\ref{1:a2}.

The validity of AR-sieve bootstrap on a class of strictly stationary multivariate time series that fulfill Assumption~\ref{1:a2} has been discussed in \citet{meyer_vector_2015}, where some additional conditions on the convergence of Yule-Walker estimators of the finite predictor coefficients on $\{\pf_t\}$ are also introduced. We summarize these conditions in Assumption~\ref{1:a3} and leave the results of \citet{meyer_vector_2015} to Lemma~\ref{1:l6} in Appendix~\ref{1:sec:Appendix_B}, as they are preliminary to show the bootstrap consistency.

\begin{assumption}\label{1:a3}
The Yule-Walker estimators $\{\tA_{l,p}, l = 1,2,...,p\}$ of $\{\pA_{l,p}, l = 1,2,...,p\}$ in~\eqref{1:e2}, the finite predictor coefficients matrices on the VAR representation of $\{\pf_t\}$, fulfills that $p^2 \sum_{l=1}^{p} \|\tA_{l,p} - \pA_{l,p} \|_F = O_P\left(1\right),$ as $T \to \infty$ and $p \to \infty$.
\end{assumption}

% \textbf{Justification for Assumption~\ref{1:a3}}:
% Assumption~\ref{1:a3} requires $p \to \infty$ at a relatively slower rate of sample size $T$, which is required for the convergence of the Yule-Walker estimator of $\pA_{p} = (\pA_{1,p},...,\pA_{p,p})$. In other words, the order $p$ of the AR terms in AR-sieve bootstrap depends on the sample size $T$ and has to be chosen properly. For $\{\pf_t\}$ fulfilling Assumption~\ref{1:a2}, Assumption~\ref{1:a3} is also satisfied if we choose $p = O\left((T/\ln T)^{1/6}\right)$ \citep[e.g.,][]{meyer_vector_2015}. Assumptions \ref{1:a2} and \ref{1:a3} are widely discussed in the literature of AR-sieve bootstrap, for example, in \citet{kreiss_range_2011} and \citet{meyer_vector_2015}. In summary, Assumption~\ref{1:a2} ensures the existence of VAR representation in~\eqref{1:e2} and specifies the rate of decaying for the coefficient matrices and Assumption~\ref{1:a3} relates to the convergence of Yule-walker estimators $\{\tA_{l,p}\}$ to the finite predictor coefficient matrices $\{\pA_{l,p}\}$.
\begin{assumption} \label{1:a4}
The dimension $N$ and AR(p) satisfy $N \to \infty$, $p \to \infty$ when $T \to \infty$ such that $p^{11/2}(N^{-1/2} + T^{-1/2}) \to 0$.
\end{assumption}

% \textbf{Justification for Assumption~\ref{1:a4}}:
% In addition to Assumption~\ref{1:a3}, Assumption~\ref{1:a4} is introduced as the bootstrap procedure is performed on the estimated factors $\{\hf_{t}\}$ rather than true unobservable factors $\{\pf_t\}$, where the error comes from both the estimation of factors and finite order approximation of AR-sieve representations. In other words, we need to control the error imposed by the bootstrap procedure by restricting the speed at which the AR order $p$ goes to infinity. On the other hand, the order on dimension $N$ in Assumption $\ref{1:a4}$ also indicates `blessing of dimensionality', since the increase of the dimension $N$ will enhance the strength of common factors $\{\pf_t\}$ \citep{lam_estimation_2011}.

Discussions and justifications for Assumptions~\ref{1:a2}~to~\ref{1:a4} are left to Appendix~\ref{assu} of the Supplementary Material.

\subsection{Bootstrap validity for mean statistics}

%One of the most fundamental problems in functional data analysis is to estimate the mean function from observations with noises. Some methods and applications can be found, for example, in \citet{ramsay_applied_2002} and \citet{cai_optimal_2011}. Under the setting where observations are generally sparse, statistical inferences for (general) mean statistics of high-dimensional data are also fundamental.

The validity of the general AR and VAR sieve bootstrap has been discussed by \citet{kreiss_range_2011} and \citet{meyer_vector_2015}. It is worth noting that the general AR and VAR sieve bootstrap does not mimic the behavior of the underlying processes in~\eqref{1:e3} or~\eqref{1:e2}, but the behavior of so-called companion processes $\{\brf_t\}$. The companion processes $\{\brf_t\}$ are defined in the same form as $\{\pf_t\}$ but with i.i.d.\ white noises $\{\bre_{t}\}$ rather than the uncorrelated white noises $\{\pe_t\}$ in~\eqref{1:e3} or~\eqref{1:e2}, although $\{\pe_t\}$ and $\{\bre_{t}\}$ share the same distribution. That means, without additional assumptions on the distribution of $\{\pe_t\}$, the higher-order properties of $\{\brf_t\}$ and $\{\pf_t\}$ are not necessarily the same. In other words, except for the Gaussian case, the general AR and VAR sieve bootstrap work for statistics that only depend on up-to-second-order quantities of $\{\pf_t\}$. %With this result, the bootstrap consistency for mean statistics has been investigated by \citet{kreiss_range_2011} and \citet{meyer_vector_2015}.
%Based on this result, we can study statistical inferences for such class of statistics based on the unobservable factor terms $\pQ$ and $\{\pf_t\}$ by bootstrapping $\{\hf_t\}$.

To summarize our first result on the bootstrap consistency of $\pQ \overline{\pf_T}$, the mean statistics of the unobservable factor component $\{\pQ\pf_t\}$, we use $\mathbb{E}^\ast$ to denote the expectation with respect to the measure assigning probability $1/(T-p)$ to each observation.
\begin{theorem}\label{1:t1}
Suppose that Assumptions~\ref{1:c1},~\ref{1:a2} ($\gamma=1$),~\ref{1:a3} and~\ref{1:a4} are satisfied for a fixed and known number of factors $r$. In addition, if we further assume that 
\begin{enumerate}
		%\item[(a)] $\mathbb{E} \left( e_{j,t}^{4} \right)<\infty$ for each element $e_{j,t}$ in $\{\pe_t\}$
\item[(a)] The empirical distribution of $\{\pe_t\}$ converges weakly to the distribution function of $\mathcal{L}(\pe_t)$.
\item[(b)] $\lim_{T \to \infty} \mathbb{V} \comm{(\sqrt{{T'}}\overline{\pf_{T'}} )} = \sum_{k\in\mathbf{Z}} \pGa_{\pf} (k) > 0$.
\end{enumerate}
Then, for any vector $\pc \in \mathbb{R}^N$ such that $\| \pc^\top \pQ\|_{\ell_1} < \infty$ and $ 0 < \sum_{k\in\mathbf{Z}} \pc^\top \pQ \pGa_{\pf} (k) \pQ^\top \pc < \infty$ as $N \to \infty$, we can conclude that when $N \to \infty$ and $T \to \infty$,
\comm{\begin{align*}
d_{K}\left(\mathcal{L} \left( \left. \sqrt{T'} \pc^{\top} \hQ \left( \overline{\bof_{T'}} - \mathbb{E}^\ast \overline{\bof_{T'}} \right) \right| \py_1,\py_2,...,\py_T \right), \mathcal{L} \left(\sqrt{T'} \pc^\top \pQ \left( \overline{\pf_{T'}} - \mathbb{E} \overline{\pf_{T'}} \right) \right)\right)  \overset{p}{\to} 0,
\end{align*}
where $T' = T-p$ is the effective sample size, $\overline{\bof_{T'}} = \frac{1}{T'}\sum_{t=p+1}^{T}\bof_{t}$, $\overline{\pf_{T'}} = \frac{1}{T'}\sum_{t=p+1}^{T}\pf_{t}$,} $\mathcal{L}$ and $d_K$ denote the probability distribution and the Kolmogorov distance, respectively.
\end{theorem}
\color{darkblue}
\begin{remark} \label{rem:c_vector_example}
The condition $\|\pc^{\top}\pQ\|_{1} < \infty$ in Theorem 4.1 ensures that the asymptotic variance of the bootstrap mean statistic remains well-defined as $N \to \infty$. A practical example satisfying this condition is the inference on the global cross-sectional mean. Consider $c = \frac{1}{N}(1, 1, \dots, 1)^{\top} \in \mathbb{R}^{N}$. In this case, the statistic of interest is the grand mean of the observed process. The projection term becomes
\begin{equation*}
\pc^{\top}\pQ = \frac{1}{N}\sum_{i=1}^{N}\pq_{i}.
\end{equation*}
Under Assumption~3.1 (ii) (where $\frac{1}{N}\pQ^{\top}\pQ = \pI_{r}$), the factor loadings are pervasive. Consequently, the average loading $\pc^{\top}\pQ$ converges to the population mean of the factor loadings (a constant vector in $\mathbb{R}^{r}$) as $N \to \infty$. Thus, the term $\|\pc^{\top}\pQ\|_{1}$ remains bounded, and the limiting distribution in Theorem 4.1 stabilizes.
\end{remark}
\color{black}
Theorem~\ref{1:t1} states the validity of the proposed AR-sieve bootstrap methods on estimated factors $\{\hf_t\}$. In general, the bootstrap inferences can be considered as an alternative statistical tool for practical use compared with the asymptotic results, which can be rather difficult to derive, especially for high-dimensional time series. The factor model in~($\ref{1:e1}$) filters out time-invariant noise $\{\pu_t\}$ and projects the original time series onto a low-dimensional subspace where the AR-sieve bootstrap procedure can be developed.
% \begin{remark}
% 	As discussed in \citet{kreiss_range_2011} and \citet{meyer_vector_2015}, AR-sieve bootstrap mimics the behavior of a companion process $\brf_t$ which shares the same first and second-order quantities as $\{\pf_t\}$. Hence for the mean statistics, AR-sieve bootstrap works without any additional assumptions made on the higher-order moments of $\{\pf_t\}$. Also, for AR-sieve bootstrap to be asymptotically valid on $\{\pf_t\}$, the dimension $r$ needs not to go to infinity. To study the impact of the factor strength on the validity of the AR-sieve bootstrap, we consider various factor strengths in simulation studies in Section~\ref{1:sim}.
% \end{remark}

\subsection{Bootstrap consistency for autocovariance matrices}

For high-dimensional i.i.d.\ data, the covariance matrix plays an important role in dimension-reduction techniques, such as factor models and principal component analysis. However, for high-dimensional dependent data, the autocovariance matrices are vital or even more crucial than the covariance matrix. %\cite{lam_estimation_2011} provides a discussion on the use of autocovariance in dimension reduction. 
Therefore, it is critical to establish the bootstrap consistency for the autocovariance matrices under the proposed AR-sieve bootstrap method. In the next theorem, we show that the proposed AR-sieve bootstrap method can guarantee the asymptotic consistency on the autocovariance matrices, which in turn implies the validity of using bootstrap data $\{\py_t^\ast\}$ to approximate the original data $\{\py_t\}$.

Recall that $\pGa_f(k) = \text{Cov}(\pf_t,\pf_{t+k})$ is the autocovariance of unobservable factors $\{\pf_t\}$ at lag $k$, for $k>0$. Without loss of generality, we again assume that the means of the factors are~0 to simplify the notation and define $\text{Cov}^\ast$ as the covariance with respect to the measure that assigns probability $1/(T-p)$ to each observation. Denoted by $\pGa^\ast_f(k) = \text{Cov}^\ast(\bof_t,\bof_{t+k})$ the autocovariance of bootstrap factors $\{\bof_t\}$ at lag $k$, we have the following theorem on the consistency of $\pGa^\ast_f(k)$.
\begin{theorem}\label{1:t2}
Suppose that Assumptions~\ref{1:c1}, \ref{1:a2} ($\gamma = 1$) and \ref{1:a3} are satisfied for fixed and known number of factors $r$. In addition, if we further assume that the empirical distribution of $\{\pe_t\}$ converges weakly to the distribution function of $\mathcal{L}(\pe_t)$.
	%\begin{enumerate}
	%\item[(a)] $\mathbb{E} \left( e_{j,t}^{4} \right)<\infty$ for each element $e_{j,t}$ in $\{\pe_t\}$, (see (\ref{1:gmean}) for the definition of $h$).
	%\item[(a)] 
	%\end{enumerate}
Then for $k \in \mathbb{N}$, we have as $N, T\rightarrow\infty$, 
\begin{equation*}
\left\| \pGa^\ast_f(k) - \pGa_f(k) \right\|_{2} \overset{p}{\to} 0.
\end{equation*}
\end{theorem}
Let $\{\pdel_i(k)\}_{i=1}^r$ be the ordered spiked eigenvalues of $\frac{1}{N^2}\pGa_y(k)\pGa_y(k)^\top$, the symmetrized autocovariance matrices of $\{\py_t\}$ at lag $k>0$. And define $\{\bdel_i(k)\}_{i=1}^r$ as the first $r$ largest eigenvalues of $\frac{1}{N^2}\pGa^\ast_y(k)\pGa^\ast_y(k)^\top$, the bootstrap symmetrized autocovariance matrices of $\{\by_t\}$ at lag $k>0$, where $\pGa^\ast_y(k) = \text{Cov}^\ast(\by_{t},\by_{t+k})$. As a consequence of Theorem \ref{1:t2}, we immediately have the following proposition on the convergence of spiked eigenvalues of the bootstrap symmetrized autocovariance matrices to their population counterparts.
\begin{proposition}\label{1:t3}
Under the Assumptions of Theorem \ref{1:t2}, for $i = 1,2,...,r$ and $k \in \mathbb{N}$, we have 
\begin{align*}
\left\| \pGa^\ast_y(k) - \pGa_y(k) \right\|_2 \overset{p}{\to} 0, \ \ 
\left| \bdel_i(k) - \pdel_i(k) \right| \overset{p}{\to} 0, \ as \ N, T\rightarrow\infty. 
\end{align*}
\end{proposition}

The asymptotic property of spiked eigenvalues of symmetrized autocovariance matrices is significant in many applications. However, there is no literature due to the difficulties and complexities of studying dependent data when $N \to \infty$. Proposition~\ref{1:t3} verifies the consistency of the bootstrap on the spiked eigenvalues of the symmetrized autocovariance matrices and provides statistical tools to study the properties of spiked eigenvalues based on the AR-sieve bootstrap.
\begin{remark}
Despite that $\pGa^\ast_y(k) = \text{Cov}^\ast(\by_{t},\by_{t+k})$ are the autocovariances defined conditionally on the sample observations, the convergence results in Proposition~\ref{1:t3} are on the entire probability space, which allows for the use of autocovariances and their spiked eigenvalues computed from a bootstrap sample $\{\py_t^\ast\}$ to approximate the autocovariances and the corresponding spiked eigenvalues of the original data $\{\py_t\}$.
\end{remark}

%\section{Simulation studies}\label{1:sim}

\section{Simulation studies}\label{simu2}

\color{darkblue}
In this section, we evaluate the finite-sample performance of the proposed AR-sieve bootstrap confidence intervals for the mean statistics. We compute empirical coverage probabilities and discuss the impacts of sample size $T$, data dimension $N$, and factor strength under strong factor scenarios. Comprehensive additional simulation results are provided in Appendix~\ref{1:sec:Appendix_C} of the Supplementary Material. These include: (1) an examination of the proposed method's performance in constructing confidence intervals for the eigenvalues of the symmetrized autocovariance matrix; and (2) a comparison with the standard moving block bootstrap method for vector time series to demonstrate the curse of dimensionality in high-dimensional settings.
%\subsection{AR-sieve bootstrap for mean statistics} \label{1:s2}

To evaluate the finite-sample performance of the proposed AR-sieve bootstrap method, we first examine the empirical coverage and average width of the bootstrap confidence intervals for the mean statistics (Theorem~\ref{1:t1}). We consider the data-generating process (DGP) based on the factor model
\begin{align}\label{sim1}
    \py_t = \pQo \pfo_t + \pu_t,
\end{align}
where the factor loading matrix $\pQo \in \mathbb{R}^{N \times r}$ is generated by QR decomposition of a matrix $N \times r$ with independent standard normal entries that satisfy $\pQ^{o\top} \pQo = \pI_r$. We set the number of factors at $r=2$. The idiosyncratic errors $\{u_{i,t}\}$ are generated as independent white noise $\mathcal{N}(0,1)$.

Latent factors $\pfo_t = (f_{1,t}, f_{2,t})^{\top}$ are generated from independent AR(1) processes to capture the temporal dependence
\begin{align*}
f_{i,t} = 0.5 f_{i,t-1} + e_{i,t}, \quad \text{for}\quad i=1,2.
\end{align*}
%where the innovations $e_{i,t}$ are independent normal variables. To control the factor strength, we follow the setting in Lam et al. (2011) and set the variance of $e_{i,t}$ such that the factors are relatively strong (e.g., corresponding to $\nu=1$, details of weaker factor settings are deferred to the Supplementary Material).
The autoregressive coefficient is set to $0.5$ to reflect a moderate temporal dependence. To study the impact of factor strength, we follow the definition in \citet{lam_factor_2012} and assume that the innovations $e_{i,t}$ follow $\mathcal{N}(0, \sigma_i^2 N^\nu)$, where $\nu \in (0,1]$ controls the strength. Specifically, we set the variances as $\sigma_1^2 = 1$ and $\sigma_2^2 = 0.5$. This scaling ensures that the first two eigenvalues of the accumulated symmetrized autocovariance matrices are spiked and distinct.

We focus on the strong factor scenarios where $\nu = 1$ or $\nu=0.8$. Simulations for weaker factors ($\nu \in \{0.6, 0.4, 0.2\}$) are provided in Appendix~\ref{1:sec:Appendix_C} of the Supplementary Material. For each scenario, we perform 1,000 Monte Carlo replications. In each replication, $B=999$ bootstrap samples are generated to construct confidence intervals for the standardized mean statistic, defined as $\theta_y \coloneqq \frac{\sqrt{T}}{\sqrt{N^{\nu}}} \pone^\top \pQ \pmu_f$. Standardization by $N^{\nu}$ facilitates the comparison of interval lengths across different factor strengths.

Specifically, we compute the bootstrap estimates as $\overline{\overline{y^\ast}} = \frac{\sqrt{T}}{\sqrt{N^{\nu}}} \pone^\top \hQ \overline{\bof}$, where $\overline{\bof}$ is the bootstrap sample mean of the factors and the order $p$ of the AR-sieve is selected via the Akaike Information Criterion (AIC) for each replication. We evaluate the performance using two types of intervals: the nonparametric bootstrap interval using quantiles and the parametric bootstrap interval based on normality. 

Both methods are computationally efficient and widely used. For an arbitrary statistic $\theta$ and its sample estimate $\widehat{\theta}$, the nonparametric bootstrap interval (also known as the reverse percentile interval) is calculated as 
\begin{align*}
	\left(2\widehat{\theta} - \theta^\ast_{(1-\alpha/2)},\ 2\widehat{\theta} - \theta^\ast_{(\alpha/2)}  \right),
\end{align*}
%where $\theta^\ast_{(1-\alpha/2)}$ is the $(1-\alpha/2)$ percentile of the bootstrap estimates $\theta^\ast$ and $\alpha$ denotes a level of significance. The nonparametric bootstrap interval using quantiles is also known as the reverse percentile interval, as the order of upper and lower quantiles is reversed in the formula. The idea of the nonparametric bootstrap interval using quantiles is to use the bootstrap distribution of $(\theta^\ast - \widehat{\theta})$ to approximate the distribution of $(\widehat{\theta} - \theta)$. On the other hand, the parametric bootstrap interval based on normality can be computed as 
where $\theta^\ast_{(\alpha)}$ denotes the $\alpha$-percentile of the bootstrap distribution and $\alpha$ is the significance level. This method relies on the approximation that the distribution of $(\theta^\ast - \widehat{\theta})$ mimics that of $(\widehat{\theta} - \theta)$. Alternatively, the parametric bootstrap interval based on normality is computed as
\begin{align*}
\left(\widehat{\theta} - b^\ast - \sqrt{v^\ast}z_{(1-\alpha/2)}, \widehat{\theta} - b^\ast + \sqrt{v^\ast}z_{(1-\alpha/2)}\right),
\end{align*}
where $b^\ast$ is the bootstrap bias estimate, $v^\ast$ is the bootstrap variance estimate, and $z_{(1-\alpha/2)}$ is the $(1-\alpha/2)$ percentile of the standard normal distribution. %Similar to the nonparametric bootstrap interval using quantiles, the parametric bootstrap interval based on normality also assumes the bootstrap distribution of $(\theta^\ast - \widehat{\theta})$ correctly approximates the distribution of $(\widehat{\theta} - \theta)$, but is constructed in a parametric way. To achieve improved empirical coverage and width of intervals, as discussed by \citet{hall_theoretical_1988}, more sophisticated intervals such as percentile-t and accelerated bias-corrected intervals may be produced using a double bootstrap, with a much higher cost of computations. Since this example's main purpose is to inspect the validity and consistency of our proposed AR-sieve bootstrap method under various cases, we only use these two basic ways of bootstrap intervals as they are simple and computationally efficient. Finally, to get a comprehensive comparison of the performance of two types of intervals, we compute the empirical coverage, average width, and interval score \citep{gneiting_strictly_2007} of bootstrap intervals under various combinations of $N$ and $T$. The interval score of a bootstrap interval $(l,u)$ is computed as
While more sophisticated methods, such as percentile-$t$ or accelerated bias-corrected intervals, could offer higher-order accuracy \citep{hall_theoretical_1988}, they require a double bootstrap procedure, which is computationally expensive. Since our primary goal is to verify the validity and consistency of the proposed AR-sieve bootstrap framework, we focus on these two fundamental and efficient interval types.

To provide a comprehensive performance assessment, we report empirical coverage, average width, and interval score \citep{gneiting_strictly_2007}. The interval score for a confidence interval $(l,u)$ is defined as
\begin{align*}
S_{\alpha} = (u-l) + \frac{2}{\alpha}(l-\theta)\mathbbm{1}{\theta<l} + \frac{2}{\alpha}(\theta-u)\mathbbm{1}{\theta>u}.
\end{align*}
This metric rewards narrower intervals while penalizing those that fail to cover the true parameter $\theta$. It serves as a robust summary statistic for comparing intervals, especially when coverage probabilities and widths are similar.

Table~\ref{1:ta1a} presents the results for the strongest factor scenario ($\nu = 1$). We report the metrics for nominal coverage levels of $95\%$, $90\%$, and $80\%$ across various combinations of $N$ and $T$. As shown in Table~\ref{1:ta1a}, when the sample size $T$ is sufficiently large, the empirical coverage approaches the nominal levels and remains robust to the ratio $N/T$. In particular, the average width of the intervals remains stable or even improves as $N$ increases. This phenomenon reflects the ``blessing of dimensionality'' in factor analysis, where the estimation precision of the common factors improves with the cross-sectional dimension. Overall, the performance benefits from increases in both $N$ and $T$. Comparing the two interval types, the average interval scores are very close across almost all scenarios, indicating that both the nonparametric and parametric approaches perform equally well for strong factors. Table~\ref{1:ta1c} displays the results for the slightly weaker factor case ($\nu = 0.8$), exhibiting a similar consistent performance.
\color{black}
\begin{table}[!htbp]
	\centering
	\caption{Empirical coverage, average width, and interval score of bootstrap intervals using quantiles for $\theta_y$ with $\nu=1$.}
	\label{1:ta1a}
	\resizebox{\textwidth}{!}{%
		\begin{tabular}{|ccccccccccc|}
			\hline
			&  & \multicolumn{3}{c}{95\%} & \multicolumn{3}{c}{90\%} & \multicolumn{3}{c|}{80\%} \\ \hline
			T & N & \begin{tabular}[c]{@{}c@{}}Empirical\\  coverage\end{tabular} & \begin{tabular}[c]{@{}c@{}}Average\\ width\end{tabular} & \begin{tabular}[c]{@{}c@{}}Average\\  interval score\end{tabular} & \begin{tabular}[c]{@{}c@{}}Empirical\\  coverage\end{tabular} & \begin{tabular}[c]{@{}c@{}}Average\\ width\end{tabular} & \begin{tabular}[c]{@{}c@{}}Average\\  interval score\end{tabular} & \begin{tabular}[c]{@{}c@{}}Empirical\\  coverage\end{tabular} & \begin{tabular}[c]{@{}c@{}}Average\\ width\end{tabular} & \begin{tabular}[c]{@{}c@{}}Average\\  interval score\end{tabular} \\ \hline
			\multicolumn{11}{|c|}{Nonparametric bootstrap intervals using quantiles} \\ \hline
			\multirow{5}{*}{200} & 50 & $0.941$ & $8.369$ & $11.572$ & $0.892$ & $7.029$ & $10.686$ & $0.799$ & $5.480$ & $9.449$ \\
			& 100 & $0.948$ & $8.407$ & $11.339$ & $0.901$ & $7.067$ & $10.466$ & $0.811$ & $5.511$ & $9.289$ \\
			& 200 & $0.941$ & $8.366$ & $11.868$ & $0.889$ & $7.038$ & $10.745$ & $0.787$ & $5.488$ & $9.568$ \\
			& 500 & $0.935$ & $8.438$ & $12.514$ & $0.876$ & $7.098$ & $11.470$ & $0.778$ & $5.536$ & $10.394$ \\
			& 1000 & $0.943$ & $8.513$ & $13.615$ & $0.889$ & $7.161$ & $11.759$ & $0.792$ & $5.584$ & $10.170$ \\
			&&&&&&&&&& \\ 
			\multirow{5}{*}{500} & 50 & $0.936$ & $8.501$ & $12.354$ & $0.882$ & $7.160$ & $11.352$ & $0.781$ & $5.579$ & $10.212$ \\
			& 100 & $0.940$ & $8.275$ & $11.693$ & $0.887$ & $6.964$ & $10.804$ & $0.781$ & $5.427$ & $9.808$ \\
			& 200 & $0.943$ & $8.430$ & $12.900$ & $0.891$ & $7.096$ & $11.465$ & $0.792$ & $5.531$ & $9.978$ \\
			& 500 & $0.946$ & $8.354$ & $11.818$ & $0.902$ & $7.023$ & $10.678$ & $0.797$ & $5.484$ & $9.566$ \\
			& 1000 & $0.941$ & $8.147$ & $12.547$ & $0.894$ & $6.850$ & $10.905$ & $0.802$ & $5.344$ & $9.448$ \\
			&&&&&&&&&& \\ 
			\multirow{5}{*}{1000} & 50 & $0.935$ & $8.594$ & $13.142$ & $0.898$ & $7.219$ & $11.658$ & $0.777$ & $5.629$ & $10.220$ \\
			& 100 & $0.944$ & $8.428$ & $13.273$ & $0.892$ & $7.088$ & $11.662$ & $0.777$ & $5.531$ & $10.442$ \\
			& 200 & $0.938$ & $8.194$ & $12.472$ & $0.888$ & $6.889$ & $11.300$ & $0.784$ & $5.371$ & $9.943$ \\
			& 500 & $0.946$ & $8.469$ & $11.918$ & $0.894$ & $7.123$ & $11.077$ & $0.806$ & $5.559$ & $9.865$ \\
			& 1000 & $0.944$ & $8.479$ & $11.928$ & $0.884$ & $7.133$ & $11.177$ & $0.783$ & $5.565$ & $10.141$ \\ \hline
			\multicolumn{11}{|c|}{Parametric bootstrap intervals based on normality} \\ \hline
			\multirow{5}{*}{200} & 50 & $0.944$ & $8.402$ & $11.457$ & $0.893$ & $7.051$ & $10.646$ & $0.795$ & $5.493$ & $9.482$ \\
			& 100 & $0.947$ & $8.449$ & $11.321$ & $0.903$ & $7.090$ & $10.444$ & $0.818$ & $5.524$ & $9.271$ \\
			& 200 & $0.941$ & $8.407$ & $11.698$ & $0.890$ & $7.055$ & $10.657$ & $0.789$ & $5.497$ & $9.534$ \\
			& 500 & $0.935$ & $8.481$ & $12.343$ & $0.878$ & $7.117$ & $11.414$ & $0.775$ & $5.545$ & $10.379$ \\
			& 1000 & $0.942$ & $8.555$ & $13.575$ & $0.888$ & $7.180$ & $11.747$ & $0.793$ & $5.594$ & $10.211$ \\
			&&&&&&&&&& \\
			\multirow{5}{*}{500} & 50 & $0.939$ & $8.548$ & $12.167$ & $0.886$ & $7.174$ & $11.276$ & $0.775$ & $5.590$ & $10.198$ \\
			& 100 & $0.943$ & $8.318$ & $11.609$ & $0.890$ & $6.981$ & $10.837$ & $0.781$ & $5.439$ & $9.814$ \\
			& 200 & $0.942$ & $8.470$ & $12.784$ & $0.897$ & $7.109$ & $11.384$ & $0.794$ & $5.538$ & $9.980$ \\
			& 500 & $0.944$ & $8.395$ & $11.648$ & $0.903$ & $7.046$ & $10.592$ & $0.797$ & $5.489$ & $9.532$ \\
			& 1000 & $0.942$ & $8.190$ & $12.434$ & $0.895$ & $6.873$ & $10.928$ & $0.806$ & $5.355$ & $9.406$ \\
			&&&&&&&&&& \\
			\multirow{5}{*}{1000} & 50 & $0.938$ & $8.632$ & $13.181$ & $0.898$ & $7.244$ & $11.639$ & $0.778$ & $5.644$ & $10.206$ \\
			& 100 & $0.945$ & $8.470$ & $13.145$ & $0.891$ & $7.108$ & $11.652$ & $0.780$ & $5.538$ & $10.419$ \\
			& 200 & $0.942$ & $8.232$ & $12.587$ & $0.891$ & $6.908$ & $11.266$ & $0.786$ & $5.382$ & $9.908$ \\
			& 500 & $0.947$ & $8.516$ & $11.874$ & $0.895$ & $7.147$ & $11.015$ & $0.810$ & $5.568$ & $9.826$ \\
			& 1000 & $0.948$ & $8.525$ & $11.815$ & $0.888$ & $7.154$ & $11.097$ & $0.782$ & $5.574$ & $10.119$ \\ \hline
		\end{tabular}%
	}
\end{table}

\begin{table}[!htbp]
	\centering
	\caption{Empirical coverage, average width, and interval score of bootstrap intervals for $\theta_y$ with $\nu=0.8$.}
	\label{1:ta1c}
	\resizebox{\textwidth}{!}{%
		\begin{tabular}{|ccccccccccc|}
			\hline
			&  & \multicolumn{3}{c}{95\%} & \multicolumn{3}{c}{90\%} & \multicolumn{3}{c|}{80\%} \\ \hline
			T & N & \begin{tabular}[c]{@{}c@{}}Empirical\\  coverage\end{tabular} & \begin{tabular}[c]{@{}c@{}}Average\\ width\end{tabular} & \begin{tabular}[c]{@{}c@{}}Average\\  interval score\end{tabular} & \begin{tabular}[c]{@{}c@{}}Empirical\\  coverage\end{tabular} & \begin{tabular}[c]{@{}c@{}}Average\\ width\end{tabular} & \begin{tabular}[c]{@{}c@{}}Average\\  interval score\end{tabular} & \begin{tabular}[c]{@{}c@{}}Empirical\\  coverage\end{tabular} & \begin{tabular}[c]{@{}c@{}}Average\\ width\end{tabular} & \begin{tabular}[c]{@{}c@{}}Average\\  interval score\end{tabular} \\ \hline
			\multicolumn{11}{|c|}{Nonparametric bootstrap intervals using quantiles} \\ \hline
			\multirow{5}{*}{200} & 50 & $0.948$ & $8.391$ & $11.301$ & $0.897$ & $7.044$ & $10.418$ & $0.807$ & $5.493$ & $9.236$ \\
			& 100 & $0.955$ & $8.432$ & $10.869$ & $0.903$ & $7.087$ & $10.158$ & $0.816$ & $5.530$ & $9.068$ \\
			& 200 & $0.950$ & $8.388$ & $11.421$ & $0.894$ & $7.051$ & $10.315$ & $0.802$ & $5.499$ & $9.181$ \\
			& 500 & $0.940$ & $8.472$ & $11.870$ & $0.887$ & $7.125$ & $10.800$ & $0.788$ & $5.559$ & $9.881$ \\
			& 1000 & $0.951$ & $8.610$ & $12.637$ & $0.899$ & $7.248$ & $10.940$ & $0.807$ & $5.647$ & $9.628$ \\
			&&&&&&&&&& \\
			\multirow{5}{*}{500} & 50 & $0.939$ & $8.507$ & $12.423$ & $0.884$ & $7.161$ & $11.384$ & $0.777$ & $5.578$ & $10.242$ \\
			& 100 & $0.943$ & $8.274$ & $11.467$ & $0.893$ & $6.962$ & $10.626$ & $0.788$ & $5.429$ & $9.662$ \\
			& 200 & $0.943$ & $8.463$ & $12.846$ & $0.894$ & $7.120$ & $11.297$ & $0.800$ & $5.546$ & $9.901$ \\
			& 500 & $0.949$ & $8.392$ & $11.535$ & $0.907$ & $7.048$ & $10.397$ & $0.798$ & $5.506$ & $9.322$ \\
			& 1000 & $0.940$ & $8.173$ & $12.197$ & $0.902$ & $6.876$ & $10.698$ & $0.811$ & $5.363$ & $9.249$ \\
			&&&&&&&&&& \\
			\multirow{5}{*}{1000} & 50 & $0.933$ & $8.590$ & $13.273$ & $0.892$ & $7.216$ & $11.784$ & $0.774$ & $5.631$ & $10.288$ \\
			& 100 & $0.942$ & $8.428$ & $13.244$ & $0.896$ & $7.097$ & $11.703$ & $0.769$ & $5.532$ & $10.485$ \\
			& 200 & $0.936$ & $8.195$ & $12.470$ & $0.894$ & $6.894$ & $11.246$ & $0.784$ & $5.376$ & $9.867$ \\
			& 500 & $0.950$ & $8.490$ & $11.764$ & $0.892$ & $7.138$ & $10.971$ & $0.809$ & $5.571$ & $9.820$ \\
			& 1000 & $0.949$ & $8.498$ & $11.801$ & $0.887$ & $7.147$ & $11.066$ & $0.782$ & $5.571$ & $10.094$ \\ \hline
			\multicolumn{11}{|c|}{Parametric bootstrap intervals based on normality} \\ \hline
			\multirow{5}{*}{200} & 50 & $0.950$ & $8.423$ & $11.218$ & $0.898$ & $7.068$ & $10.412$ & $0.802$ & $5.507$ & $9.258$ \\
			& 100 & $0.955$ & $8.476$ & $10.796$ & $0.905$ & $7.113$ & $10.086$ & $0.818$ & $5.542$ & $9.046$ \\
			& 200 & $0.948$ & $8.427$ & $11.318$ & $0.900$ & $7.072$ & $10.254$ & $0.805$ & $5.510$ & $9.149$ \\
			& 500 & $0.942$ & $8.516$ & $11.765$ & $0.893$ & $7.147$ & $10.735$ & $0.790$ & $5.568$ & $9.860$ \\
			& 1000 & $0.954$ & $8.653$ & $12.561$ & $0.897$ & $7.262$ & $10.991$ & $0.803$ & $5.658$ & $9.636$ \\
			&&&&&&&&&& \\
			\multirow{5}{*}{500} & 50 & $0.942$ & $8.546$ & $12.259$ & $0.880$ & $7.172$ & $11.316$ & $0.778$ & $5.588$ & $10.236$ \\
			& 100 & $0.945$ & $8.316$ & $11.383$ & $0.892$ & $6.979$ & $10.629$ & $0.786$ & $5.437$ & $9.664$ \\
			& 200 & $0.943$ & $8.498$ & $12.714$ & $0.900$ & $7.131$ & $11.261$ & $0.796$ & $5.556$ & $9.882$ \\
			& 500 & $0.949$ & $8.429$ & $11.348$ & $0.910$ & $7.074$ & $10.308$ & $0.804$ & $5.512$ & $9.311$ \\
			& 1000 & $0.945$ & $8.219$ & $12.119$ & $0.900$ & $6.898$ & $10.728$ & $0.815$ & $5.374$ & $9.217$ \\
			&&&&&&&&&& \\
			\multirow{5}{*}{1000} & 50 & $0.933$ & $8.631$ & $13.249$ & $0.891$ & $7.243$ & $11.750$ & $0.773$ & $5.643$ & $10.288$ \\
			& 100 & $0.943$ & $8.472$ & $13.150$ & $0.891$ & $7.110$ & $11.681$ & $0.774$ & $5.540$ & $10.458$ \\
			& 200 & $0.937$ & $8.233$ & $12.517$ & $0.894$ & $6.909$ & $11.213$ & $0.783$ & $5.383$ & $9.839$ \\
			& 500 & $0.951$ & $8.534$ & $11.722$ & $0.894$ & $7.162$ & $10.896$ & $0.808$ & $5.580$ & $9.769$ \\
			& 1000 & $0.952$ & $8.540$ & $11.670$ & $0.888$ & $7.167$ & $10.998$ & $0.784$ & $5.584$ & $10.071$ \\ \hline
		\end{tabular}%
	}
\end{table}

\newpage

\section{Empirical application: Particulate matter concentration}\label{1:sec:6}

We apply the proposed AR-sieve bootstrap method to a real data set. The raw data are observations of $\text{PM}_{10}$ particles in the air, collected on a half-hour basis in Graz, Austria, from 1 October 2010 to 31 March 2011. The particles $\text{PM}_{10}$ represent a common type of air pollutant that can be found in smoke and dust with an aerodynamic diameter of less than $0.01$mm. 

This data set has been studied by \citet{hormann_dynamic_2015} for topics of dynamic functional principal component analysis (FPCA), by \cite{Shang17} for topics of dynamic updating, and by \citet{shang_bootstrap_2018} for comparisons of bootstrap methods for stationary functional time series. The original data are pre-processed by a square-root transformation to stabilize the variance and avoid heavy-tailed observations as directed by \citet{aue_prediction_2015} and \citet{hormann_dynamic_2015}. The square-root of $\text{PM}_{10}$ levels contained in a matrix $48 \times 182$ are then plotted in Figure~\ref{1:f10} as high-dimensional time series over $182$ days with the dimension of $48$ and in Figure~\ref{1:f11} as $182$ days of $48$ half-hourly observations within each day. In general, $\text{PM}_{10}$ concentration levels are relatively high in winter when temperatures are low and pollutants related to daily life, such as traffic and heating, lack space to disperse in the atmosphere. The day-to-day $\text{PM}_{10}$ levels in winter, therefore, are highly temporally dependent, while the half-hourly observations in each day experience similar patterns which are mainly related to people's day-to-day life and temperature. 
\begin{figure}[!htbp]
\centering
\begin{subfigure}[h]{0.43\textwidth}
\includegraphics[scale=0.6]{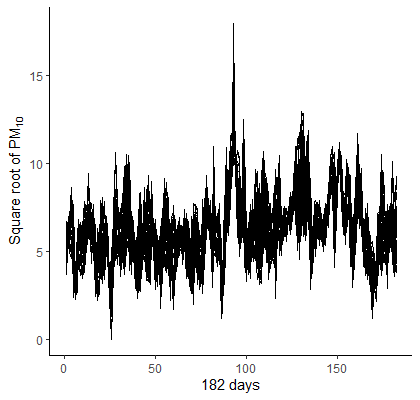}
\caption{Univariate time series plot}\label{1:f10}
\end{subfigure}
\qquad
\begin{subfigure}[h]{0.43\textwidth}
\includegraphics[scale=0.6]{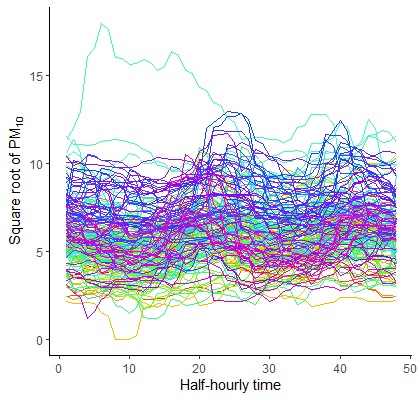}
\caption{Functional time series plot}\label{1:f11}
\end{subfigure}	        
\caption{Observed time series of (square-root) $\text{PM}_{10}$ levels}
\end{figure} 

In \citet{hormann_dynamic_2015} and \citet{shang_bootstrap_2018}, observations of half-hourly $\text{PM}_{10}$ levels as in Figure~\ref{1:f11} are assumed to come from a functional curve. In general, for a functional time series, the original observations are smoothed before further studies such as FPCA and functional bootstrap. Hence, according to \citet{hormann_dynamic_2015} and \citet{shang_bootstrap_2018}, there are $182$ temporal dependent functional curves, each smoothed from $48$ observations. However, as illustrated in Appendix~\ref{appendix0} of the Supplementary Material, the pre-smoothing results rely heavily on the smoothness condition of the functional curve. When the observations are not dense enough, pre-smoothing may cause a loss of information, especially on local patterns. To maintain the original features of time-series observations to the greatest extent, we treat the data as a multivariate or high-dimensional time series. We then perform the proposed AR-sieve bootstrap methods with a factor model on this $48$ by $182$ matrix of time series. This creates a bootstrap confidence interval for the mean levels of (square root) $\text{PM}_{10}$ that are temporal dependent at each half-hourly time point, and to create a bootstrap confidence surface for the lag-$1$ autocovariance matrix of (square root) $\text{PM}_{10}$ levels. 
\begin{figure}[!htbp]
\begin{minipage}[b]{0.45\linewidth}
\centering
\includegraphics[width=\textwidth]{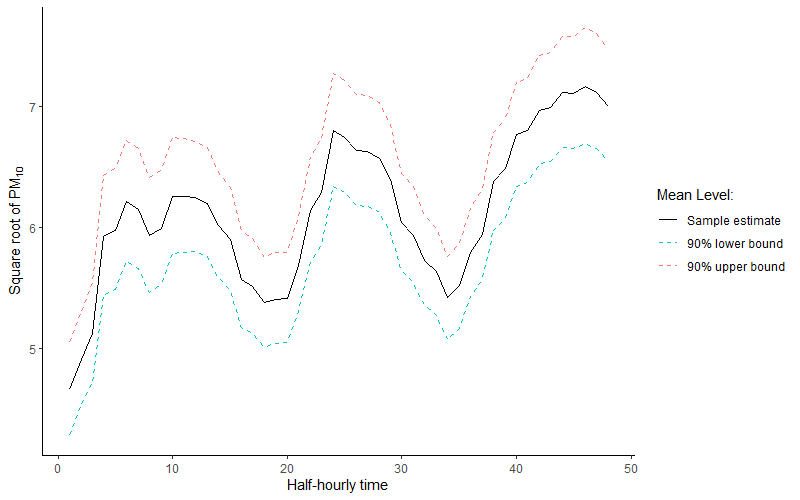}
\caption{$90\%$ AR-sieve bootstrap confidence interval for the mean of temporal dependent (square root) $\text{PM}_{10}$ levels at $48$ half-hourly time}\label{1:f12}
\end{minipage}
\hspace{0.5cm}
\begin{minipage}[b]{0.45\linewidth}
\centering
\includegraphics[width=\textwidth]{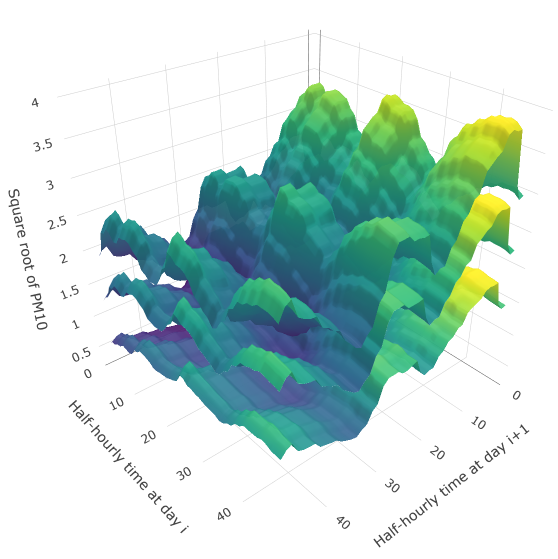}
\caption{$90\%$ AR-sieve bootstrap confidence surface for lag-$1$ autocovariance of temporal dependent (square root) $\text{PM}_{10}$ levels at $48$ half-hourly time point}\label{1:f13}
\end{minipage}
\end{figure}

In Figure~\ref{1:f12}, a $90\%$ nonparametric bootstrap interval using quantiles is created on the mean levels of (square root) $\text{PM}_{10}$, defined as $\theta_y \coloneqq \pQ \pmu_f$ with $\pmu_f$ denoting the population mean of temporal dependent factors $\{\pf_t\}$. From this graph of the sample estimate and the confidence interval of $\theta_y$, it is clear that local patterns, for example between $4$\textsuperscript{th} and $10$\textsuperscript{th} half-hourly time points, are preserved flawlessly by our proposed AR-sieve bootstrap method. Similarly, a sample estimate and a $90\%$ nonparametric bootstrap interval using quantiles for the lag-$1$ autocovariance matrix $\text{Cov}(\py_t,\py_{t+1})$ of temporal dependent (square root) $\text{PM}_{10}$ levels at $48$ half-hourly time points are also computed and presented in Figure~\ref{1:f13}. This nonparametric bootstrap interval using quantiles provides interval estimates on autocovariance of (square root) $\text{PM}_{10}$ levels between two consecutive days, where, as shown in Figure~\ref{1:f13}, local patterns are again completely preserved by the proposed AR-sieve bootstrap method. 

\section{Conclusions and discussions}\label{1:sec:7}

\comm{We utilize factor models to effectively reduce dimensionality and capture temporal dependence, enabling the establishment of the AR-sieve bootstrap for high-dimensional time series.} Specifically, we suggest using autocovariance to estimate the factor model and performing an AR-sieve bootstrap on the estimated factors to provide the ultimate inferences on the original time series. Our proposed AR-sieve bootstrap methods using factor models provide valid statistical inferences on the mean statistic and maintain consistency on bootstrap estimates of spiked eigenvalues of autocovariance matrices. Simulation studies provide numerical evidence on the finite-sample performance of the AR-sieve bootstrap methods. Finally, we apply our methods to the $\text{PM}_{10}$ data to construct bootstrap confidence intervals for the mean vector and the autocovariance matrix, respectively. 

Our work is crucial as a building block for bootstrap methods for high-dimensional time series. We propose a low-rank model for the AR-sieve bootstrap on high-dimensional stationary time series. There are two ways in which the present paper could be further extended: 
\begin{inparaenum}
\item[1)] The asymptotics of the bootstrap validity on the mean statistics can be extended for weaker factor models; 
\item[2)] While the AR-sieve bootstrap is only valid for stationary time series, alternative bootstrap methods can be considered on the factors where the dimension has been reduced.
\end{inparaenum}
%studying the differences and connections between studies on functional and high-dimensional time series. Our future work includes further exploration and justification of the density or sparsity of functional time series observations and the impact of pre-smoothing, which are fundamental for data analysis.

\section*{Acknowledgment}

The authors thank the insightful comments provided by the two reviewers.

\bigskip
%\begin{center}
%{\large\bf SUPPLEMENTARY MATERIAL}
%\end{center}

%\begin{description}
%\item[Appendices:] Discussions and justifications for assumptions, additional simulation results, technical proofs of results, discussions of applying the proposed AR-sieve bootstrap on sparsely observed functional time series, and additional simulations on the AR-sieve bootstrap.
%\item[R-codes:] R-codes for performing the AR-sieve bootstrap and reproducing the simulations.
%\item[$\text{PM}_{10}$ data:] Data set used in Section~\ref{1:sec:6}.
%\end{description}

%\appendix
\begin{appendices}

\section{Technical proof of theorems}\label{1:sec:Appendix_A}

\begin{proof}[Proof of Lemma~\ref{1:l5}]
The upper bound $d$ for all $\omega \in (0,2\pi]$ follows directly from the norm summability condition stated in Assumption~\ref{1:a2}. The assumption of strong factors in Assumption~\ref{1:c1} implies the positivity on eigenvalues of the spectral density matrix $W_{\pf} (\omega)$. Denoted by $\sigma_i(\omega)$, the minimum eigenvalue of $W_{\pf} (\omega)$ for $i=1,2,\dots,r$, then $\sigma_i(\omega)$ is continuous in $(0,2\pi]$ and strictly positive. Denoted by $\sigma_\text{min}=\text{min}_{\omega \in (0,2\pi]} (\sigma_i (\omega))$, the minimum eigenvalue of the spectral density matrix of $\{\pf_t\}$, then there exists a constant $c > 0$ so that $\sigma_\text{min} \ge c$ for all frequencies $\omega \in (0,2\pi]$.
\end{proof}

\color{darkblue}
\begin{proof}[Proof of Theorem~\ref{1:t1}]
Let $\boof_t = \sum_{l=1}^{p} \tA_{l,p} \boof_{t-l} + \booe_{t,p}$, where $\{\tA_{l,p},\ l=1,2,...,p\}$ are the estimators of the AR coefficient matrices based on true factors $\{\pf_t\}$, and $\{\booe_{t,p}, t=p+1,p+2,...,T\}$ are generated by i.i.d.\ resampling from the centered residuals $(\te_{t,p} - \overline{\te}_{{T'},p})$ with $\te_{t,p}=\pf_t-\sum_{l=1}^{p} \tA \pf_{t-l}$ and $\overline{\te}_{{T'},p} = \frac{1}{T-p} \sum_{t=p+1}^T \te_{t,p}$. Therefore, $\{\boof_t\}$ are bootstrap pseudo-variables generated based on the true factors $\{\pf_t\}$ rather than $\{\hf_t\}$. Recall that $\{\bof_t\}$ are bootstrapped based on the centered residuals $\{(\he_{t,p} - \overline{\he}_{{T'},p})\}$ \comm{for $t=p+1,...,T$}, with $\he_{t,p}=\hf_t-\sum_{l=1}^{p} \hA_{l,p} \hf_{t-l}$ and $\overline{\he}_{{T'},p} = \frac{1}{T-p} \sum_{t=p+1}^T \he_{t,p}$, and we define $\mathbb{E}^\ast$ and $\text{Cov}^\ast$ as the expectation and covariance with respect to the measure that assigns probability $1/(T-p)$ to each observation, respectively. 
\color{darkblue}
Recall that $T' = T-p$ is the effective sample size, the bootstrap empirical mean is defined as $\overline{\bof_{T'}} = \frac{1}{T'}\sum_{t=p+1}^{T}\bof_{t}$, and the sample mean of the estimated factors is defined as $\overline{\hf_{T'}} = \frac{1}{T'}\sum_{t=p+1}^{T}\hf_{t}$. By definition, $\mathbb{E}^\ast \overline{\bof_T}= \overline{\hf_T}$.
    %Therefore, $\mathbb{E}^\ast \overline{\bof_T}= \overline{\hf_T}$ by definition and we can write
    
Then, we decompose the standardized statistic as
\begin{align*}
\sqrt{T'} \pc^\top \hQ \left( \overline{\bof_{T'}} - \mathbb{E}^\ast \overline{\bof_{T'}} \right)
&=
\sqrt{T'} \pc^\top \pQ \left( \overline{\boof_{T'}} - \mathbb{E}^\ast \overline{\boof_{T'}} \right) +
\sqrt{T'} \pc^\top \left(\hQ-\pQ\right) \left( \overline{\bof_{T'}} - \mathbb{E}^\ast \overline{\bof_{T'}} \right) \\
&+
\sqrt{T'}\pc^\top \pQ \left[\left( \overline{\bof_{T'}} - \mathbb{E}^\ast \overline{\bof_{T'}} \right)-\left( \overline{\boof_{T'}} - \mathbb{E}^\ast \overline{\boof_{T'}} \right)\right]
\eqqcolon \mathcal{M}_1 + \mathcal{M}_2 + \mathcal{M}_3,
\end{align*}
with obvious definitions of $\mathcal{M}_1, \mathcal{M}_2$ and $\mathcal{M}_3$.
	
For the term $\mathcal{M}_1$, under Assumptions~\ref{1:c1} (iii), \ref{1:a2} and the additional assumption in Theorem \ref{1:t1} that $\lim_{T \to \infty} \mathbb{V} (\sqrt{{T'}}\overline{\pf_{T'}} ) = \sum_{k\in\mathbf{Z}} \pGa_{\pf} (k) < \infty$, using Theorem 2.1 in \citet{politis_subsampling_1997}, we have the CLT for $\sqrt{{T'}}\ \overline{\pf_{T'}}$ as $\sqrt{{T'}}\left( \overline{\pf_{T'}} - \mathbb{E}\overline{\pf_{T'}} \right) \overset{d}\to \mathcal{N}\ \left(0,\ \sum_{k\in\mathbf{Z}} \pGa_{\pf} (k)\right).$ Moreover, under the additional assumptions in Theorem~\ref{1:t1}, $\pc^\top \pQ$ is an $r$-dimensional vector such that $\| \pc^\top \pQ \|_{\ell_1} < \infty$ for a fixed $r$, Therefore, under Assumptions~\ref{1:c1} (ii) and~\ref{1:a2}, we can use Cramer-Wold Theorem \citep{cramer_theorems_1936} to conclude for the scalar $\sqrt{{T'}} \pc^\top \pQ \overline{\pf_{T'}}$ that
\begin{align*}
\sqrt{{T'}} \pc^\top \pQ \left( \overline{\pf_{T'}} - \mathbb{E}\overline{\pf_{T'}} \right)
\overset{d}\to \mathcal{N}\ \left(0,\  \pc^\top \pQ \left( \sum_{k\in\mathbf{Z}} \pGa_{\pf}(k) \right)  \pQ^\top \pc \right),
\end{align*}
when $T,N \to \infty$. 
	
In addition, under the strong mixing condition on the true factors $\{\pf_t\}$, the empirical moments of $\{\pe_t\}$ converge to their population counterpart. Therefore, under all the assumptions of \ref{1:t1}, we fulfill all the conditions of Theorem~4.1 in \citet{meyer_vector_2015}. Consequently, we can use Theorem~4.1 in \citet{meyer_vector_2015} to conclude that the general VAR-sieve bootstrap is valid for $\sqrt{{T'}} \pc^\top \pQ \overline{\pf_{T'}}$ since $\sqrt{{T'}} \pc^\top \pQ \overline{\pf_{T'}}$ shares the same CLT with its counterpart generated from the companion process as discussed in \citet{meyer_vector_2015}. Hence,
\begin{align*}
d_K\left(\mathcal{L} \left( \left. \sqrt{{T'}} \pc^\top \hQ \left( \overline{\bof_{T'}} - \mathbb{E}^\ast \overline{\bof_{T'}} \right) \right| \py_1,\py_2,\dots,\py_{T} \right), \mathcal{L} \left(\sqrt{{T'}} \pc^\top \pQ \left( \overline{\pf_{T'}} - \mathbb{E} \overline{\pf_{T'}} \right) \right)   \right) =o_P\left(1\right)
\end{align*}
as $T, N \to \infty$. 
	
Therefore, to see the assertion in Theorem~\ref{1:t1}, we first need to show that when $T,N \to \infty$, both $\mathcal{M}_2$ and $\mathcal{M}_3$ tend to $0$ in probability, then apply Slutsky's theorem. To show $\mathcal{M}_2 \to 0$ in probability for $T,N \to \infty$, we first notice that
$
\sqrt{{T'}}\pc^\top \left(\hQ-\pQ\right) \left( \overline{\bof_{T'}} - \mathbb{E}^\ast \overline{\bof_{T'}} \right) = \frac{1}{\sqrt{{T'}}} \pc^\top \left(\hQ-\pQ\right) \sum_{t=p+1}^{T}\left(\bof_t - \overline{\hf_{T'}} \right).
$
Therefore, we can show that
\begin{align*}
& \mathbb{E} \left[ \sqrt{{T'}} \pc^\top \left(\hQ-\pQ\right) \left( \overline{\bof_{T'}} - \mathbb{E} \overline{\bof_{T'}} \right) \right]^2
=
\mathbb{E} \left[ \frac{1}{{T'}} \pc^\top \left(\hQ-\pQ\right) \sum_{t=p+1}^{T}\left(\bof_t - \overline{\hf_{T'}} \right) \right]  \left[ \sum_{s=p+1}^{T}\left(\bof_s - \overline{\hf_{T'}} \right)^\top  \left(\hQ-\pQ\right)^\top \pc \right]\\
=& 
\left[ \frac{1}{{T'}} \pc^\top \left(\hQ-\pQ\right) \sum_{t=p+1}^{T} \sum_{s=p+1}^{T} \mathbb{E} \left(\bof_t - \overline{\hf_{T'}} \right)  \left(\bof_s - \overline{\hf_{T'}} \right)^\top  \left(\hQ-\pQ\right)^\top \pc \right]
\end{align*}
\begin{align*}
&\le
\frac{1}{{T'}} \left\| \pc^\top \left(\hQ-\pQ\right) \right\|^2 \left\| \sum_{t=p+1}^{T} \sum_{s=p+1}^{T} \mathbb{E} \left(\bof_t - \overline{\hf_{T'}} \right)  \left(\bof_s - \overline{\hf_{T'}} \right)^\top  \right\|_F\\
&=
O_P\left(\frac{1}{{T'}^2} \left\| \sum_{t=p+1}^{T} \sum_{s=p+1}^{T} \mathbb{E} \left(\bof_t - \overline{\hf_{T'}} \right)  \left(\bof_s - \overline{\hf_{T'}} \right)^\top  \right\|_F\right).
\end{align*}
	%below have some weird characters
	%	where the last line follows from the fact that $\| \pc^\top \pQ \|_{\ell_1} < \infty$ for $N,T \to \infty$ under the additional assumptions in Theorem \ref{1:t1}, $\| \pQ \|_2 \asymp \sqrt{N}$, and $\left\| \hQ − \pQ \right\|_2 = O_P \left(N^{1/2}T^{−1/2}\right)$ by Lemma \ref{1:l2}. 	
Define ${\Sigma}_{e,p}^\ast \coloneqq \mathbb{E}^\ast \left( \be_{t}  \pe_{t}^{\ast\top} \right)$, then
$
\sum_{t=p+1}^{T} \sum_{s=p+1}^{T} \mathbb{E}^\ast \left(\bof_t - \overline{\hf_{T'}} \right)  \left(\bof_s - \overline{\hf_{T'}} \right)^\top 
=
\sum_{t=p+1}^{T} \sum_{s=p+1}^{T} \mathbb{E}^\ast \left( \left( \sum_{l_1=0}^{\infty} \hPsi_{l_1,p} \be_{t-l_1} \right) \right. \\ \left. \left( \sum_{l_2=0}^{\infty}  \hPsi_{l_2,p} \be_{s-l_2}\right)^\top \right)
=
\sum_{t=p+1}^{T} \sum_{s=p+1}^{T} \mathbb{E}^\ast  \sum_{l_1=0}^{\infty} \sum_{l_2=0}^{\infty} \left( \hPsi_{l_1,p} \be_{t-l_1}  \pe_{s-l_2}^{\ast\top} \hPsi_{l_2,p}^\top \right)
=
\sum_{t=p+1}^{T} \sum_{s=p+1}^{T}   \sum_{l=0}^{\infty}  \hPsi_{l,p} \mathbb{E}^\ast \left( \be_{t-l}  \pe_{t-l}^{\ast\top} \right) \hPsi_{s-t+l,p}^\top
$
where $\be_{t-l_1}$ and $\be_{t-l_2}$ are i.i.d.\ bootstrapped, therefore $\mathbb{E}^\ast \left( \be_{t-l_1}  \pe_{t-l_2}^{\ast\top} \right) = \mathbf{0}$ for $l_1 \ne l_2$. Hence, we can show that
\begin{equation*}
\frac{1}{{T'}^2} \left\|  \sum_{t=p+1}^{T} \sum_{s=p+1}^{T} \mathbb{E}^\ast \left(\bof_t - \overline{\hf_{T'}} \right)  \left(\bof_s - \overline{\hf_{T'}} \right)^\top  \right\|_F
\le
\frac{1}{{T'}^2} \left\| {\Sigma}_{e,p}^\ast \right\|_F \sum_{l=0}^{\infty}\left\| \hPsi_{l,p} \right\|_F \sum_{t=p+1}^{T} \sum_{s=p+1}^{T} \left\| \hPsi_{s-t+l,p} \right\|_F
= O_P\left(\frac{1}{{T'}}\right),
\end{equation*}
where we note that Lemmas~\ref{1:l6} and~\ref{1:k1} imply the summability of $\left\| \hPsi_{l,p} \right\|_F$, hence $\sum_{s=p+1}^{T} \left\| \hPsi_{s-t+l,p} \right\|_F$ is bounded for $T \to \infty$. Therefore, $\frac{1}{{T'}} \sum_{l=0}^{\infty}\left\| \hPsi_{l,p} \right\|_F \sum_{t=p+1}^{T} \sum_{s=p+1}^{T} \left\| \hPsi_{s-t+l,p} \right\|_F$ is bounded for $T \to \infty$, and we can conclude that $\mathbb{E}^\ast \left[ \sqrt{{T'}} \pc^\top \left(\hQ-\pQ\right) \left( \overline{\bof_{T'}} - \mathbb{E}^\ast \overline{\bof_{T'}} \right) \right]^2 \to 0$ in probability, which suffices for $\mathcal{M}_2 \to 0$ in probability conditional on the sample. 

For $\mathcal{M}_3$, we first write
\begin{align*}
&\mathbb{E}^\ast \left[ \sqrt{{T'}} \pc^\top \pQ \left\{\left( \overline{\bof_{T'}} - \mathbb{E}^\ast \overline{\bof_{T'}} \right)-\left( \overline{\boof_{T'}} - \mathbb{E}^\ast \overline{\boof_{T'}} \right)\right\} \right]^2
=
\mathbb{E}^\ast \left\| \sqrt{{T'}} \pc^\top \pQ \left\{\left( \overline{\bof_{T'}} - \overline{\hf_{T'}} \right)-\left( \overline{\boof_{T'}} - \overline{\tf_{T'}} \right)\right\} \right\|^2\\
\le&
\| \pc^\top \pQ \|^2  \frac{1}{{T'}} \sum_{t=p+1}^{T} \sum_{s=p+1}^{T} \mathbb{E}^\ast \left\| \left\{\left( \bof_t - \overline{\hf_{T'}} \right)-\left( \boof_t - \overline{\tf_{T'}} \right)\right\} \left\{\left( \bof_s - \overline{\hf_{T'}} \right)-\left( \boof_s - \overline{\tf_{T'}} \right)\right\}^\top \right\|_F\\
=&
O_{P}\left( \frac{1}{{T'}} \sum_{t=p+1}^{T} \sum_{s=p+1}^{T} \mathbb{E}^\ast \left\| \left\{\left( \bof_t - \overline{\hf_{T'}} \right)-\left( \boof_t - \overline{\tf_{T'}} \right)\right\} \left\{\left( \bof_s - \overline{\hf_{T'}} \right)-\left( \boof_s - \overline{\tf_{T'}} \right)\right\}^\top \right\|_F\right),
\end{align*}
where the last line follows from the fact that $\| \pc^\top \pQ \|^2$ is bounded when $N \to \infty$. To proceed, first note that
\begin{align*}
&\frac{1}{{T'}} \sum_{t=p+1}^{T} \sum_{s=p+1}^{T} \mathbb{E}^\ast \left\{\left( \bof_t - \overline{\hf_{T'}} \right)-\left( \boof_t - \overline{\tf_{T'}} \right)\right\} \left\{\left( \bof_s - \overline{\hf_{T'}} \right)-\left( \boof_s - \overline{\tf_{T'}} \right)\right\}^\top\\
=&
\frac{1}{{T'}} \sum_{t=p+1}^{T} \sum_{s=p+1}^{T} \mathbb{E}^\ast \left\{\sum_{l_1=0}^{\infty} \hPsi_{l_1,p} \be_{t-l_1,p} - \tPsi_{l_1,p} \booe_{t-l_1,p} \right\} \left\{\sum_{l_2=0}^{\infty} \hPsi_{l_2,p} \be_{s-l_2,p} - \tPsi_{l_2,p} \booe_{s-l_2,p} \right\}^{\top} \\
=&\frac{1}{{T'}} \sum_{t=p+1}^{T} \sum_{s=p+1}^{T} \mathbb{E}^\ast \left\{\sum_{l_1=0}^{\infty} \hPsi_{l_1,p} \be_{t-l_1,p} \right\} \left\{\sum_{l_2=0}^{\infty} \hPsi_{l_2,p} \be_{s-l_2,p} - \tPsi_{l_2,p} \booe_{s-l_2,p} \right\}^\top\\
+&\frac{1}{{T'}} \sum_{t=p+1}^{T} \sum_{s=p+1}^{T} \mathbb{E}^\ast \left\{\sum_{l_1=0}^{\infty} \tPsi_{l_1,p} \booe_{t-l_1,p} \right\} \left\{\sum_{l_2=0}^{\infty} \tPsi_{l_2,p} \booe_{s-l_2,p} - \hPsi_{l_2,p} \be_{s-l_2,p} \right\}^\top
\eqqcolon
\frac{1}{{T'}} \sum_{t=p+1}^{T} \sum_{s=p+1}^{T} \left( \mathcal{H}_1 + \mathcal{H}_2 \right),
\end{align*}
with an obvious notation for $\mathcal{H}_1$ and $\mathcal{H}_2$. Then, we only consider $\mathcal{H}_1$ as $\mathcal{H}_2$ can be dealt with similarly. 
	
For $\mathcal{H}_1$, we can further decompose it as
\begin{align*}
\mathcal{H}_1	=&
		\mathbb{E}^\ast \left\{\sum_{l_1=0}^{\infty} \hPsi_{l_1,p} \be_{t-l_1,p} \right\} \left\{\sum_{l_2=0}^{\infty} \hPsi_{l_2,p} \be_{s-l_2,p} - \tPsi_{l_2,p} \be_{s-l_2,p} \right\}^\top
		+
		\mathbb{E}^\ast \left\{\sum_{l_1=0}^{\infty} \hPsi_{l_1,p} \be_{t-l_1,p} \right\} \\
  &\left\{\sum_{l_2=0}^{\infty} \tPsi_{l_2,p} \be_{s-l_2,p} - \tPsi_{l_2,p} \booe_{s-l_2,p} \right\}^\top
		=
		\sum_{l=0}^{\infty} \hPsi_{l,p} \mathbb{E}^\ast \left\{ \be_{t-l,p} \pe_{t-l,p}^{\ast\top} \right\} \left\{ \hPsi_{l+s-t,p} - \tPsi_{l+s-t,p} \right\}^\top\\
		+&
		\sum_{l=0}^{\infty} \hPsi_{l,p} \mathbb{E}^\ast \left\{ \be_{t-l,p} (\be_{t-l,p}- \booe_{t-l,p})^\top  \right\}  \tPsi_{l+s-t,p} ^\top
		\eqqcolon \mathcal{H}_{11} + \mathcal{H}_{12}.
	\end{align*}
	where the second last equation follows from the bootstrap independence for $l_1 \ne l_2$. Hence we have for $\mathcal{H}_{11}$,
	\begin{align*}
		\frac{1}{{T'}} &\sum_{t=p+1}^{T} \sum_{s=p+1}^{T} \left\|  \mathcal{H}_{11} \right \|_F
		=
		\frac{1}{{T'}} \sum_{t=p+1}^{T} \sum_{s=p+1}^{T} \left\| \sum_{l=0}^{\infty} \hPsi_{l,p} {\Sigma}_{e,p}^\ast \left\{ \hPsi_{l+s-t,p} - \tPsi_{l+s-t,p} \right\}^\top \right\|_F\\
		\le&
		\left\| {\Sigma}_{e,p}^\ast \right\|_F \frac{1}{{T'}}  \sum_{l=0}^{\infty} \left\| \hPsi_{l,p} \right\|_F  \sum_{t=p+1}^{T} \sum_{s=p+1}^{T} \left\| \hPsi_{l+s-t,p} - \tPsi_{l+s-t,p}  \right\|_F
		=
		O_P\left(p^{\frac{3}{2}} \left\| \hA_{p} - \tA_{p} \right\|_F \right)
		=o_P\left(1\right),
	\end{align*}
where the second last equation follows from the results in Lemmas~\ref{1:l6} and~\ref{1:k1}, and the last equation follows the result in Lemma~\ref{1:k1}. For $\mathcal{H}_{12}$ we can show that
\begin{align*}
\frac{1}{{T'}} \sum_{t=p+1}^{T} \sum_{s=p+1}^{T}\left\|  \mathcal{H}_{12} \right \|_F
&\le
\sqrt{\mathbb{E}^\ast \left\| \be_{t,p} \right\|^2} \sqrt{\mathbb{E}^\ast \left\| \be_{t,p}- \booe_{t,p} \right\|^2} \frac{1}{{T'}} \sum_{l=0}^{\infty} \left\| \hPsi_{l,p} \right\|_F  \sum_{t=p+1}^{T} \sum_{s=p+1}^{T} \left\| \hPsi_{l+s-t,p} \right\|_F\\
&=
O_P\left(\sqrt{\mathbb{E}^\ast \left\| \be_{t,p}- \booe_{t,p} \right\|^2}\right),
\end{align*}
where the last equation follows from the same arguments on summability properties in Lemmas \ref{1:l6}. Hence, it remains to show $\mathbb{E}^\ast \left\| \be_{t,p}- \booe_{t,p} \right\|^2 \to 0$ in probability. Recall that $\mathbb{E}^\ast$ defines expectation with respect to the measure that assigns probability $1/(T-p) = 1/T'$ to each observation, which follows as follows.
\begin{align*}
\mathbb{E}^{\ast} \left\| \be_{t,p}- \booe_{t,p} \right\|^2
= \ & 
		\mathbb{E}^\ast \left\{ \left( \be_{t,p}- \booe_{t,p} \right) \left( \be_{t,p}- \booe_{t,p} \right)^\top \right\}  \\
		=& \frac{1}{{T'}} \sum_{t=p+1}^{T} \left\{ ( \he_{t,p} - \overline{\he}_{{T'},p} ) - ( \te_{t,p} - \overline{\te}_{{T'},p} ) \right\} \left\{ ( \he_{t,p} - \overline{\he}_{{T'},p} ) - ( \te_{t,p} - \overline{\te}_{{T'},p} ) \right\}^\top \\
		=& \frac{1}{{T'}} \sum_{t=p+1}^{T} \left\{ ( \he_{t,p} - \te_{t,p} ) - ( \overline{\he}_{{T'},p} - \overline{\te}_{{T'},p} ) \right\} \left\{ ( \he_{t,p} - \te_{t,p} ) - ( \overline{\he}_{{T'},p} - \overline{\te}_{{T'},p} ) \right\}^\top \\
		\le&
		\frac{2}{{T'}} \sum_{t=p+1}^{T} \left\| \he_{t,p} - \te_{t,p}  \right\|^2 + 2\left\{ \left\| \overline{\te}_{{T'},p} \right\|^2 + \left\| \overline{\he}_{{T'},p} \right\|^2 - 2 \left\| \overline{\te}_{{T'},p} \right\| \left\| \overline{\he}_{{T'},p} \right\|  \right\} \\
		\le&
		\frac{2}{{T'}} \sum_{t=p+1}^{T} \left\| \he_{t,p} - \te_{t,p}  \right\|^2 + 4\left\{ \left\| \overline{\te}_{{T'},p} \right\|^2 + \left\| \overline{\he}_{{T'},p} \right\|^2  \right\}.
	\end{align*}
	Recall that when $\{\pf_t\}$ and $\{\hf_t\}$ have non-zero means, $\te_{t,p} = \left(\pf_t - \overline{\pf_{T'}}\right) - \sum_{l=1}^{p} \tA_{l,p} \left(\pf_{t-l} - \overline{\pf}_{{T'}}\right)$ and $\he_{t,p} = \left(\hf_t - \overline{\hf_{T'}}\right) - \sum_{l=1}^{p} \hA_{l,p} \left(\hf_{t-l} - \overline{\hf}_{{T'}}\right)$. Without altering the idea of proof, to simplify the notation used, we use $\{\pf_t\}$ and $\{\hf_t\}$ to denote the demeaned factors $\left(\pf_t - \overline{\pf_{T'}}\right)$ and their sample counterparts $\left(\hf_t - \overline{\hf_{T'}}\right)$, respectively. Therefore, with the same arguments as in the proof of Lemma~\ref{1:k2}, we have 
	\begin{align}\label{1:ee1}
		&\frac{2}{{T'}} \sum_{t=p+1}^{T} \left\| \he_{t,p} - \te_{t,p}  \right\|^2 \nonumber
		=
		\frac{2}{{T'}} \sum_{t=p+1}^{T} \left\| (\hf_{t} - \pf_{t})  + \sum_{l=1}^{p} (\tA_{l,p} \pf_{t-l} - \hA_{l,p} \hf_{t-l}) \right\|^2 \nonumber 
		\le
		\frac{4}{{T'}} \sum_{t=p+1}^{T} \left\| \hf_{t} - \pf_{t} \right\|^2  \\+& \frac{4}{{T'}} \sum_{t=p+1}^{T} \left\| \sum_{l=1}^{p}  \tA_{l,p} \pf_{t-l} - \hA_{l,p} \hf_{t-l} \right\|^2 \nonumber 
		\le
		\frac{4}{{T'}} \sum_{t=p+1}^{T} \left\| \hf_{t} - \pf_{t} \right\|^2  + 8 \sum_{l=1}^{p} \left\| \hA_{l,p} \right\|_F^2 \frac{1}{{T'}} \sum_{t=p+1}^{T} \left\|  \hf_{t-l} -  \pf_{t-l} \right\|^2 \nonumber 
		\\+& 8 \left\| \sum_{l=1}^{p} \left( \hA_{l,p} - \tA_{l,p} \right)    \frac{1}{{T'}} \sum_{t=p+1}^{T}  \pf_{t-l} \right\|_F^2 \nonumber 
		=
		O_P\left( \sup_{p+1 \le t \le T} \left\|  \hf_{t} -  \pf_{t} \right\|^2\right) + 
		O_P\left( \left\| \sum_{l=1}^{p} \left( \hA_{p} - \tA_{p} \right) \right\|_F^2 \right) \nonumber 
		\\=&
		O_P\left( \left(\frac{1}{\sqrt{T}} + \frac{1}{\sqrt{N}} \right)^2 \right)+ 
		O_P\left(p^{8} \left(\frac{1}{\sqrt{T}} + \frac{1}{\sqrt{N}} \right)^2\right)  
		=o_P(1),
	\end{align}
	where the third last equation follows from the fact that $\left\| \hA_{l,p} \right\|_F^2$ is summable, which is implied by Assumption~\ref{1:a3} and Lemma \ref{1:l8}. The second last equation is then a direct result of Lemmas \ref{1:l8} and \ref{1:l9}, and Assumption~\ref{1:a4} implies the last equation. 
	
	Furthermore, $\overline{\he}_{{T'},p} = \frac{1}{{T'}} \sum_{t=p+1}^{T} \he_{t,p} = \frac{1}{{T'}} \sum_{t=p+1}^{T} \left(\hf_t - \sum_{l=1}^{p} \hA_{l,p} \hf_{t-l} \right)$ and we can show that
	\begin{align}\label{1:ee2}
		\left\| \overline{\he}_{{T'},p} \right\|^2
		\le&
		2 \left\| \frac{1}{{T'}} \sum_{t=p+1}^{T} \hf_t  \right\|^2 + 2 \left\|  \sum_{l=1}^{p} \hA_{l,p} \frac{1}{{T'}} \sum_{t=p+1}^{T} \hf_{t-l}   \right\|^2 =o_P(1).
	\end{align}
	This is because firstly
	\begin{align*}
		&\left\| \frac{1}{{T'}} \sum_{t=p+1}^{T} \hf_t  \right\|^2
		\le
		2 \left\| \frac{1}{{T'}} \sum_{t=p+1}^{T} \pf_t  \right\|^2 + 
		2 \left\| \frac{1}{{T'}} \sum_{t=p+1}^{T} \left( \hf_t - \pf_t\right)  \right\|^2
		\\&=
		O_P\left( \frac{1}{{T'}} \right) + 
		O_P\left( \frac{1}{{T'}} \sum_{t=p+1}^{T} \left\| \hf_t - \pf_t \right\|^2  \right)
		=
		O_P\left( \frac{1}{{T'}} \right) + 
		O_P\left( \left(\frac{1}{\sqrt{T}} + \frac{1}{\sqrt{N}} \right)^2 \right)
		=o_P(1),
	\end{align*}
	where the second last equation follows as we have assumed that the population mean of $\{\pf_t\}$ is $0$ for technical convenience. Moreover,
	\begin{align*}
		\left\|  \sum_{l=1}^{p} \hA_{l,p} \frac{1}{{T'}} \sum_{t=p+1}^{T} \hf_{t-l}   \right\| 
		\le
		\sum_{l=1}^{p} \left\| \hA_{l,p} \right\|_F \left\| \frac{1}{{T'}} \sum_{t=p+1}^{T} \hf_{t-l}   \right\|
		=
		O_P\left(1\right) \times
		O_P\left( \frac{1}{\sqrt{{T'}}} + \frac{1}{\sqrt{T}} + \frac{1}{\sqrt{N}} \right)
		=o_P(1),
	\end{align*}
	where the second last equation follows from the summability conditions in Lemma \ref{1:l6}, the order of $\left\| \hf_t-\pf_t \right\|$ in Lemma \ref{1:l8} and the fact that the mean of $\{\hf_t\}$ is assumed to be $0$ for technical convenience. 
	
	Lastly, we can show that $\left\| \overline{\te}_{T} \right\|^2 \to 0$ in probability with the same technique as stated above for $\left\| \overline{\he}_{T} \right\|$. Hence, with~\eqref{1:ee1} and~\eqref{1:ee2}, we can conclude that $\frac{1}{{T'}} \sum_{t=p+1}^{T} \sum_{s=p+1}^{T} \left\| \mathcal{H}_{12} \right\|_F \to 0$ in probability. Together with the result that $\frac{1}{{T'}} \sum_{t=p+1}^{T} \sum_{s=p+1}^{T} \left\| \mathcal{H}_{11} \right\|_F \to 0$ in probability, we have $\frac{1}{{T'}} \sum_{t=p+1}^{T} \sum_{s=p+1}^{T} \left\| \mathcal{H}_1 \right\|_F \to 0$ in probability. Therefore, it suffices to conclude that $\mathcal{M}_3 \to 0$ in probability is conditional on the sample.
	Consequently, utilizing Slutsky's theorem, conditional on the sample, we can conclude that 
	\begin{align*}
		d_K\left(\mathcal{L} \left( \left. \sqrt{{T'}} \pc^\top \hQ \left( \overline{\bof_{T'}} - \mathbb{E}^\ast \overline{\bof_{T'}} \right) \right| \py_1,\py_2,...,\py_T \right), \mathcal{L} \left(\sqrt{{T'}} \pc^\top \pQ \left( \overline{\pf_{T'}} - \mathbb{E} \overline{\pf_{T'}} \right) \right)   \right)  \overset{p}{\to} 0.
	\end{align*}
\end{proof}

\color{black}
\begin{proof}[Proof of Theorem~\ref{1:t2}]
	Without loss of generality, we again assume $\{\pf_t\}$ are the demeaned factors (or the means of factors are all 0) in this proof to simplify the notations.
	
	Firstly, notice that $\bof_t = \sum_{l=1}^{p} \hA_{l,p} \bof_{t-l} + \pe_t^\ast = \sum_{l=1}^{\infty} \hPsi_{l,p} \pe_{t-l}^\ast + \pe_t^\ast = \sum_{l=0}^{\infty} \hPsi_{l,p} \pe_{t-l}^\ast.$ We can represent $\pGa^\ast_f(k)$ as
	\begin{align*}
		\pGa^\ast_f(k) 
		&= \text{Cov}^\ast (\bof_t, \bof_{t+k})
		= \text{Cov}^\ast \left(\sum_{l_1=0}^{\infty} \hPsi_{l_1,p} \pe_{t-l_1}^\ast, \sum_{l_2=0}^{\infty} \hPsi_{l_2,p} \pe_{t+k-l_2}^\ast \right)
		= \sum_{l_1=0}^{\infty} \sum_{l_2=0}^{\infty} \hPsi_{l_1,p} \text{Cov}^\ast(\pe_{t-l_1}^\ast,\pe_{t+k-l_2}^\ast)  \hPsi_{l_2,p}^\top \\
		%&= \sum_{l_1=0}^{\infty} \sum_{l_2=l_1+k} \hPsi_{l_1,p} \text{Cov}^\ast(\pe_{t-l_1}^\ast,\pe_{t+k-l_2}^\ast)  \hPsi_{l_2,p}^\top \\
		&= \sum_{l_1=0}^{\infty} \hPsi_{l_1,p} \text{Cov}^\ast(\pe_{t-l_1}^\ast,\pe_{t-l_1}^\ast)  \hPsi_{l_1+k,p}^\top 
		= \sum_{l=0}^{\infty} \hPsi_{l,p} \hS_{e,p}  \hPsi_{l+k,p}^\top,
	\end{align*}
	where we stress the fact that $\text{Cov}^\ast(\pe_{t-l_1}^\ast,\pe_{t-l_2}^\ast) = 0 $ for $l_1 \ne l_2$ and $\text{Cov}^\ast(\pe_{t-l_1}^\ast,\pe_{t-l_1}^\ast) = \mathbb{E}^\ast (\pe_{t}^\ast \pe_{t}^{\ast\top}) = \hS_{e,p}$ for all $l_1 \in \mathbf{Z},$ since $\be_t$ is uniformly distributed on the set of centered residuals $(\he_{t,p} - \overline{\he}_{{T'}})$.
 Similarly,
	\begin{align*}
		\pGa_f(k) 
		&= \text{Cov} \left(\pf_t, \pf_{t+k}\right)
		= \text{Cov} \left(\sum_{l_1=0}^{\infty} \pPsi_{l_1} \pe_{t-l_1}, \sum_{l_2=0}^{\infty} \pPsi_{l_2} \pe_{t+k-l_2}\right)
		= \sum_{l_1=0}^{\infty} \sum_{l_2=0}^{\infty} \pPsi_{l_1} \text{Cov}\left(\pe_{t-l_1},\pe_{t+k-l_2}\right)  \pPsi_{l_2}^\top \\
		%&= \sum_{l_1=0}^{\infty} \sum_{l_2=l_1+k} \pPsi_{l_1} \text{Cov}\left(\pe_{t-l_1},\pe_{t+k-l_2}\right)  \pPsi_{l_2}^\top \\
		&= \sum_{l_1=0}^{\infty} \pPsi_{l_1} \text{Cov}\left(\pe_{t-l_1},\pe_{t-l_1}\right)  \pPsi_{l_1+k}^\top 
		= \sum_{l=0}^{\infty} \pPsi_{l} \pS_{e}  \pPsi_{l+k}^\top,
	\end{align*} 
	where we write $\pS_e = \text{Cov} (\pe_{t}, \pe_{t})$ and use the fact that $\pf_t = \sum_{l=1}^{\infty} \pA_{l} \pf_{t-l} + \pe_t = \sum_{l=1}^{\infty} \pPsi_{l} \pe_{t-l} + \pe_t = \sum_{l=0}^{\infty} \pPsi_{l} \pe_{t-l}.$
	
	To see the assertion in this theorem, we first of all define an intermediate term $\pGa_{f,p}(k) \coloneqq \sum_{l=0}^{\infty} \pPsi_{l,p} \pS_{e,p}  \pPsi_{l+k,p}^\top$, where $\{\pPsi_{l,p},\ l\in\mathbf{N}\}$ are the power series coefficients matrices of $\left(\pI_r - \sum_{l=1}^p \pA_{l,p} z^l\right)^{-1}$ for $|z| \le 1$, and $\pS_{e,p} = \text{Cov} (\pe_{t,p}, \pe_{t,p})$ where $\pe_{t,p} = \pf_t - \sum_{l=1}^p \pA_{l,p}\pf_{t-l}$ with $\{ \pA_{l,p},\ l\in \mathbf{N}\}$ the finite predictor coefficients matrices of $\{ \pA_{l},\ l\in \mathbf{N}\}$. Hence by triangular inequality, we have
	$
		\left\| \pGa^\ast_f(k) - \pGa_f(k) \right\|_2 \le \left\| \pGa^\ast_f(k) - \pGa_{f,p}(k)  \right\|_2 + \left\| \pGa_{f,p}(k) - \pGa_f(k)  \right\|_{2}.
	$
	It is then sufficient to show both terms on the right side converge to 0 in probability. For $\left\| \pGa^\ast_f(k) - \pGa_{f,p}(k) \right\|_2$, we have
	\begin{align*}
		\left\| \pGa^\ast_f(k) - \pGa_{f,p}(k)  \right\|_2 
		&= \left\| \sum_{l=0}^{\infty} \hPsi_{l,p} \hS_{e,p}  \hPsi_{l+k,p}^\top -  \sum_{l=0}^{\infty} \pPsi_{l,p} \pS_{e,p}  \pPsi_{l+k,p}^\top \right\|_2 
		\\
		&= \left\| \sum_{l=0}^{\infty} \left[  \left(\hPsi_{l,p} - \pPsi_{l,p} \right) \hS_{e,p}  \hPsi_{l+k,p}^\top +  \pPsi_{l,p} \left(\hS_{e,p} - \pS_{e,p}\right)  \hPsi_{l+k,p}^\top \right.\right.  
		+ \left.\left.  \pPsi_{l,p} \pS_{e,p}  \left(\hPsi_{l+k,p}- \pPsi_{l+k,p}\right)^\top \right] \right\|_2   
		\\
        &= O_P\left( \sum_{l=1}^{\infty} \left\| \hPsi_{l,p} - \pPsi_{l,p} \right\|_F\right) + O_P\left(\left\| \hS_{e,p} - \pS_{e,p} \right\|_F\right),
	\end{align*}
	where the second last equation follows from the norm summable conditions on $\hPsi_{l,p}$ and $\pPsi_{l,p}$. Hence we can use the results of Lemma \ref{1:k1} and \ref{1:k2} to conclude that $\left\| \pGa^\ast_f(k) - \pGa_{f,p}(k)  \right\|_2 \to 0$ in probability. Similarly, we have
	\begin{align*}
		\left\| \pGa_{f,p}(k) - \pGa_{f}(k) \right\|_2 
		&= \left\| \sum_{l=0}^{\infty} \pPsi_{l,p} \pS_{e,p}  \pPsi_{l+k,p}^\top -  \sum_{l=0}^{\infty} \pPsi_{l} \pS_{e}  \pPsi_{l+k}^\top \right\|_2 
		\\
		&= \left\| \sum_{l=0}^{\infty} \left[  \left(\pPsi_{l,p} - \pPsi_{l} \right) \pS_{e,p}  \pPsi_{l+k,p}^\top +  \pPsi_{l} \left(\pS_{e,p} - \pS_{e}\right)  \pPsi_{l+k,p}^\top \right.\right. + \left.\left.  \pPsi_{l} \pS_{e}  \left(\pPsi_{l+k,p}- \pPsi_{l+k}\right)^\top \right] \right\|_2   
		\\&= O_P\left( \sum_{l=1}^{\infty} \left\|\pPsi_{l,p} - \pPsi_{l} \right\|_F\right) + O_P\left(\left\| \pS_{e,p} - \pS_{e} \right\|_F\right),
	\end{align*}
	since $\pPsi_{l,p}$ and $\pPsi_{l}$ are norm summable. Hence $\left\| \pGa_{f,p}(k) - \pGa_{f}(k) \right\|_2 \to 0$ in probability by Lemmas~\ref{1:k1} and~\ref{1:k2}. Therefore we can conclude that $\left\| \pGa^\ast_f(k) - \pGa_f(k) \right\|_2 \to 0$ in probability.
\end{proof}

\begin{proof}[Proof of Proposition~\ref{1:t3}]
	To see the assertions, we first note that,
	\begin{align*}
		\left\| \pGa^\ast_y(k) - \pGa_y(k) \right\|_2
		&=  \left\| \hQ \pGa^\ast_f(k) \hQ^T  - \pQ \pGa_f(k) \pQ^T \right\|_2 
		\\&\le \left\| \left(\hQ - \pQ\right) \pGa^\ast_f(k) \hQ^T \right\|_2
		+ \left\|  \pQ \left(\pGa^\ast_f(k) -\pGa_f(k)\right) \hQ^T \right\|_2 
		+ \left\|  \pQ \pGa_f(k) \left(\hQ-\pQ\right)^T \right\|_2 \\
		&= O_P\left(N^{1/2}\left\|\hQ - \pQ\right\|_2\right) + O_P\left(N\left\| \pGa^\ast_f(k) - \pGa_f(k) \right\|_2\right) 
		=o_P(1),
	\end{align*}
where the last equation follows from Assumption~\ref{1:c1}, Lemma~\ref{1:l1} and Theorem~\ref{1:t2}. To see that $\left| \bdel_i(k) - \pdel_i(k) \right| \overset{p}{\to} 0$ for $N \to \infty$ and $T \to \infty$, we can apply Weyl's Eigenvalue Theorem \citep{fan_large_2013}, that is
	$
		\left| \bdel_i(k) - \pdel_i(k) \right| \le \frac{1}{N^2} \left\| \pGa^\ast_y(k)\pGa^\ast_y(k)^\top - \pGa_y(k)\pGa_y(k)^\top \right\|_2.
	$
	Furthermore,
	\begin{align*}
		\frac{1}{N^2} \left\| \pGa^\ast_y(k)\pGa^\ast_y(k)^\top - \pGa_y(k)\pGa_y(k)^\top \right\|_2 
		&= \frac{1}{N^2} \left\| \left[\pGa^\ast_y(k) - \pGa_y(k)\right] \pGa^\ast_y(k)^\top + \pGa_y(k) \left[\pGa^\ast_y(k) - \pGa_y(k)\right]^\top \right\|_2 
  \\
		&\le \frac{1}{N^2} \left\| \left[\pGa^\ast_y(k) - \pGa_y(k)\right] \pGa^\ast_y(k)^\top \right\|_2 
		+ \frac{1}{N^2} \left\| \pGa_y(k) \left[\pGa^\ast_y(k) - \pGa_y(k)\right]^\top \right\|_2.
	\end{align*}
	It is then sufficient to consider one of the two terms on the right side since the other one can be dealt with similarly. To study $\frac{1}{N^2} \left\| \left[\pGa^\ast_y(k) - \pGa_y(k)\right] \pGa^\ast_y(k)^\top \right\|_2$, we first notice that from Assumption~\ref{1:c1}, Lemma \ref{1:l1} and Theorem \ref{1:t2}, $\left\| \pGa^\ast_y(k) \right\|_2=\left\| \hQ \pGa^\ast_f(k) \hQ^T  \right\|_2 \asymp N$. Therefore, we have
	\begin{align*}
		\frac{1}{N^2} \left\| \left[\pGa^\ast_y(k) - \pGa_y(k)\right] \pGa^\ast_y(k)^\top \right\|_2 
		&= O_P\left(\frac{1}{N}\left\| \pGa^\ast_y(k) - \pGa_y(k) \right\|_2\right) 
		= O_P\left(N^{-1/2} \left\|\hQ - \pQ \right\|_2\right) 
		+ O_P\left(\left\| \pGa^\ast_f(k) - \pGa_f(k) \right\|_2\right),
	\end{align*}
	where both terms on the right side converge to 0 in probability as shown in Lemma~\ref{1:l2} and Theorem~\ref{1:t2}.	
\end{proof}

\bibliographystyle{agsm}
\bibliography{reference1}

%\pagebreak
%\pagenumbering{arabic}
% \if0\blind
% {
%\begin{center}
  %\title{\bf \large Appendices to ``AR-sieve Bootstrap for High-dimensional Time Series''
  %\newline}
% \author{Daning Bi \orcidlink{0000-0001-9825-7290}, 
% Han Lin Shang \orcidlink{0000-0003-1769-6430}, 
% Yanrong Yang \orcidlink{0000-0002-3629-5803},
% Huanjun Zhu \orcidlink{0000-0003-0575-4525}}
 %\end{center}
% \maketitle
% } \fi

% \if1\blind
% {
%    \title{\bf \large Supplementary Material to ``AR-sieve Bootstrap for High-dimensional Time Series''}
%    \author{}
%    \maketitle
% } \fi

\newpage

% \begin{appendices}
\begin{center}
{\large\bf SUPPLEMENTARY MATERIAL}
\end{center}

\color{darkblue}
The supplementary material contains discussions and justifications for assumptions, auxiliary lemmas and proofs, additional simulation results for the AR-sieve bootstrap, and discussions of applying the proposed AR-sieve bootstrap on sparsely observed functional time series.

Appendix~\ref{assu} includes comments and justifications for all assumptions. %Appendix~\ref{simu2} contains simulation studies on the performance of AR-sieve bootstrap confidence intervals for the mean statistics. 
%In Appendix~\ref{1:sec:Appendix_B}, proofs of main theorems are present, while 
The auxiliary lemmas and their proofs that support the proofs of the main theorems are left in Appendix~\ref{1:sec:Appendix_B}. Appendix~\ref{1:sec:Appendix_C} presents additional simulations, including bootstrapping the mean statistic for the case where the factors are relatively weak, bootstrapping the spiked eigenvalues of the symmetrized autocovariance matrix, and the comparison with the moving block bootstrap. Finally, in Appendix~\ref{appendix0}, we introduce the smoothing problem on sparsely observed functional time series and then propose treating it as high-dimensional data when applying the AR-sieve bootstrap. Some simulations are also provided for bootstrapping sparsely observed functional time series.
\color{black}

\section{Discussions and justifications for assumptions}\label{assu}

\textbf{Comments and justifications for Assumption~\ref{1:c1}}:
\begin{enumerate}
\item 
Assumption~\ref{1:c1} (i) states the strict stationarity on $\{\pf_t\}$, which has been used in the literature of factor models, such as \citet{fan_large_2013} and is commonly seen in AR-sieve bootstrap literature, such as \citet{kreiss_range_2011} and \citet{meyer_vector_2015}. Apart from stationarity, the Assumption~\ref{1:c1} (i) also states that factor time series $\{\pf_t\}$ and error terms $\{\pu_t\}$ are independent of any time lags, which is stronger than the assumption in \citet{lam_estimation_2011}, but requires us to apply bootstrap methods by resampling innovations $\{\pe_t\}$ in Wold representation of $\{\pf_t\}$ as in~\eqref{1:e3}, since AR-sieve bootstrap does not work for high-dimensional noises $\{\pu_t\}$. 
% \item We impose Assumption~\ref{1:c1} (ii) to identify the factor components $\{\pQ\pf_t\}$ from the original high-dimensional data. The conditions that $\frac{1}{N} \pQ^\top\pQ = \pI_r$ and eigenvalues $\{\plam_i(\pf),\ i=1,2,...,r\}$ of $\sum_{k=1}^{k_0} \pGa_f(k)\pGa_f(k)^\top$ fulfil $\infty>\lambda_1(\pf) \ge \lambda_2(\pf) \ge \cdots \ge \lambda_r(\pf) >0$ as $N \to \infty$ are sufficient for $\{\pQ\pf_t\}$ to be identifiable from $\{\pu_t\}$ when $N \to \infty$, since the $N \times N$ matrix $\pL$ can be represented as
% \begin{equation}\label{1:e13}
% \pL=\sum_{k=1}^{k_0} \pGa_y(k)\pGa_y(k)^\top= N \pQ \left\{ \sum_{k=1}^{k_0} \pGa_{\pf}(k) \pGa_{\pf}(k)^\top \right\} \pQ^\top,
% \end{equation} 
% with the first $r$ eigenvalues of $\frac{1}{N^2}\pL$ non-vanishing. In other words, the columns of $\pQ$ can be considered as the eigenvectors of $\pL$ corresponding to $r$ nonzero eigenvalues scaled by $\sqrt{N}$. As a consequence, Assumption~\ref{1:c1} (ii) implies the pervasiveness of $r$ factors $\{\pf_t\}$ when $N$ goes to infinity, which is equivalent to the strong factors' case according to the definition in \citet{lam_estimation_2011}. 
\commHS{
\item The integer $k_0$ in Assumption~\ref{1:c1} (ii) is a prescribed parameter that determines the accumulation of autocovariance matrices. Theoretically, to ensure the identifiability of the factor space, $k_0$ is treated as a fixed constant. As stated in Assumption~\ref{1:c1} (ii), $k_0$ must be chosen such that the accumulated signal strength of the factors is sufficient to distinguish them from the noise (i.e., the first $r$ eigenvalues are distinct and significant). In practice, the choice of $k_0$ involves a trade-off. A larger $k_0$ incorporates information from more time lags, potentially improving identification if the dependence persists, but may also introduce estimation noise from higher lags where the signal is weak. As discussed in \citet{zhang_factor_2024}, the estimation results are generally robust to the choice of $k_0$, provided it captures the major temporal dependence structure. For many applications, a small value, such as $k_0=1$, is sufficient and computationally efficient. In our empirical analysis and simulations, we define $k_0$ as a fixed small integer (e.g., $k_0=1$) to satisfy the identification condition.
}
\item The $\psi$-mixing in Assumption~\ref{1:c1} (iii) is introduced to specify the \comm{weak} dependence structure of $\{\pf_t\}$, which is also considered in \citet{lam_estimation_2011} to simplify the technical proof of consistency on the loading matrix $\pQ$. However, it is not the weakest possible. Meanwhile, Assumption~\ref{1:c1} (ii) together with the mixing condition in (iii) is also sufficient for the absolute summability condition on $\{\pf_t\}$ when $N \to \infty$, which is preliminary for AR-sieve bootstrap to be applicable on $\{\pf_t\}$, since otherwise the Wold representation is not guaranteed to exist \citep{cheng_baxters_1993}. 
\end{enumerate}

\textbf{Justification for Assumption~\ref{1:a2}}:

Assumption~\ref{1:a2} is introduced to \comm{fulfill} the requirement of the existence of a general representation of VAR~\eqref{1:e2}. This type of condition is commonly used in the literature of AR-sieve bootstraps, such as \citet{kreiss_range_2011} and \citet{meyer_vector_2015}. %and it is worth noting that $\{\pe_t\}$, the innovation processes in the VAR representation of factor processes $\{\pf_t\}$, are linearly independent as $\{\pf_t\}$ is of full rank\footnote{Daning: The logic of this sentence seems to be incorrect. If $\{\pf_t\}$ is of full rank, the error component may be not linearly independent.} \citep{wiener_prediction_1958}. 
In addition, following the heredity of mixing properties in Assumption~\ref{1:c1}, $\{\pf_t\}$ is strict stationary and also $\psi-$mixing, which in turn implies the decaying of $\pGa_{\pf} (k)$ as $k \to \infty$. The matrix norm summability condition on $\pGa_f(k)$, as in Assumption~\ref{1:a2}, then specifies the rate of decay that is required for a vector AR representation to be valid as stated in the next lemma. %For simplicity, we also assume $\mathbb{E} \pf_t = 0$ which will not affect the results we derived on the consistency and validity of the proposed sieve bootstrap method. %Besides, since we assume the dimension $r$ of $\{\pf_t\}$ is finite, we can set $\mathbb{E} \pf_t = 0$ in Assumption~\ref{1:c1} will not affect the results we derived on the consistency and validity of the proposed sieve bootstrap method.\footnote{Daning: I am unclear about what would you like to express. Emphasizing the fixed dimensions on $f_t$? } 
Besides, the assumption $\mathbb{E} \pf_t = 0$ can be relaxed to $\mathbb{E} \pf_t = \mu_f$ with the cost of a more lengthy proof of theorems in this work.

\textbf{Justification for Assumption~\ref{1:a3}}:
Assumption~\ref{1:a3} requires $p \to \infty$ at a relatively slower rate of sample size $T$, which is required for the convergence of the Yule-Walker estimator of $\pA_{p} = (\pA_{1,p},...,\pA_{p,p})$. In other words, the order $p$ of the AR terms in the AR-sieve bootstrap depends on the sample size $T$ and has to be chosen properly. For $\{\pf_t\}$ fulfilling Assumption~\ref{1:a2}, Assumption~\ref{1:a3} is also satisfied if we choose $p = O\left((T/\ln T)^{1/6}\right)$ \citep[e.g.,][]{meyer_vector_2015}. Assumptions \ref{1:a2} and \ref{1:a3} are widely discussed in the literature of AR-sieve bootstrap, for example, in \citet{kreiss_range_2011} and \citet{meyer_vector_2015}. In summary, Assumption~\ref{1:a2} ensures the existence of a VAR representation in~\eqref{1:e2} and specifies the rate of decaying for the coefficient matrices and Assumption~\ref{1:a3} relates to the convergence of Yule-walker estimators $\{\tA_{l,p}\}$ to the finite predictor coefficient matrices $\{\pA_{l,p}\}$.

\textbf{Justification for Assumption~\ref{1:a4}}:
In addition to Assumption~\ref{1:a3}, Assumption~\ref{1:a4} is introduced as the bootstrap procedure is performed on the estimated factors $\{\hf_{t}\}$ rather than true unobservable factors $\{\pf_t\}$, where the error comes from both the estimation of factors and finite order approximation of AR-sieve representations. In other words, we need to control the error imposed by the bootstrap procedure by restricting the speed at which the AR order $p$ goes to infinity. On the other hand, the order of dimension $N$ in Assumption~$\ref{1:a4}$ also indicates ``blessing of dimensionality'', since the increase of the dimension $N$ will enhance the strength of common factors $\{\pf_t\}$.

\section{Auxiliary lemmas and proofs}\label{1:sec:Appendix_B}

We present some auxiliary results that facilitate the proofs of theorems in this paper. Those auxiliary results are divided into two subsections according to the related topics. In the first subsection, we present some results for factor models' estimates, and in the second subsection, the results for the AR-sieve bootstrap of factor models are summarized.

\subsection{Auxiliary results for estimates of factor models}

\begin{lemma}\label{1:l1}
	Denoted by $\left\| \mathbf{V} \right\|_{\text{min}}$ the positive square root of the minimum eigenvalue of $\mathbf{V}\mathbf{V}^\top$ or $\mathbf{V}^\top \mathbf{V}$, under Assumption~\ref{1:c1}, we have
	\begin{equation*}
		\left\| \pGa_f(k) \right\|_2 \asymp 1 \asymp \left\| \pGa_f(k) \right\|_{\text{min}},\quad \text{and}\quad \left\| \tGa_f(k)  - \pGa_f(k) \right\|_2 = O_P\left(T^{-1/2}\right).
	\end{equation*}
\end{lemma}

Lemma~\ref{1:l1} is a modification of the results in Lemma 1 and 2 of \citet{lam_estimation_2011} for the strong factors' case since we have assumed $\pQ^\top \pQ = N \pI_r$ but not $\pQ^\top \pQ = \pI_r$ %as in \citet{lam_estimation_2011}
. Therefore, the proof of Lemma~\ref{1:l1} is similar to the proofs of Lemmas~1 and~2 in \citet{lam_estimation_2011}, hence omitted.

%\begin{proof}[Proof.]
%(\ref{1:e14}) is implied by the factor strength and identification condition in Assumption~\ref{1:c1} and
%(\ref{1:e15}) is  implied by the mixing condition in Assumption~\ref{1:c1}, since the $r \times r$ sample estimate $\tGa_f(k)$ converges to $\pGa_f(k)$ in the rate of $O\left(T^{-1/2}\right)$ elementwisely.
%\end{proof}
%\begin{remark}
%	Lemma~\ref{1:l1} is a modification of the results in \citet{lam_estimation_2011} under the case that all factors are strong, since we have assumed $\pQ^\top \pQ = N \pI_r$ but not $\pQ^\top \pQ = \pI_r$ as in \citet{lam_estimation_2011}. This modification is intended only to fulfill the requirement of sieve bootstrap. It will not alter the condition for identifying factors from noises since we have only scaled the eigenvectors and the estimated factors after we have identified them. The result in Lemma \ref{1:l1} entails the boundedness condition required on the matrix norm of $\pGa_f (k)$ and the convergence rate of its sample estimate $\tGa_f (k)$. 
%\end{remark}

\begin{lemma} \label{1:l2}
Under Assumption~\ref{1:c1}, $$\left\| \hQ - \pQ \right\|_2 = O_P \left(N^{1/2}T^{-1/2}\right), \quad \text{and} \quad N^{-1/2} \left\| \hQ \hf_t - \pQ \pf_t \right\|_2 = O_P \left(T^{-1/2} + N^{-1/2} \right).$$ 
\end{lemma}
Although we scale the columns in $\pQ$ by $\sqrt{N}$ in our factor models' setting, the above convergence rate is the same as that of the strong factors' case in Theorem~3 of \citet{lam_estimation_2011}. Besides, the proof of Lemma~\ref{1:l2} is the case for strong factors in the proof of Theorem~3 in \citet{lam_estimation_2011} with the only difference in the scaled factor loading matrix $\pQ$ and factors $\pf_t$. Therefore, the proof is omitted here. %To prove Lemma \ref{1:l2}, we firstly introduce the following lemma in \citet{golub_matrix_1996}.

\begin{lemma} \label{1:l8}
	Define $\hGa_f (k) = \frac{1}{T-k} \sum_{t=1}^{T-k} \hf_t \hf_{t+k}$ and $\tGa_f (k)= \frac{1}{T-k} \sum_{t=1}^{T-k} \pf_t \pf_{t+k}$, for some $k \le p$, where $p$ fulfills Assumption \ref{1:a4}. It then holds that
	\begin{align*}
		\left\| \hGa_f (k) - \tGa_f (k) \right\|_2 = O_P\left(N^{-1/2} + T^{-1/2}\right).
	\end{align*}
\end{lemma}

Lemma~\ref{1:l8} illustrates the convergence rate on autocovariance matrices of estimated factors under the strong factors' case, which is an extension to the convergence rate of estimated factors obtained in Theorem 3 in \citet{lam_estimation_2011}.
\begin{proof}[Proof of Lemma~\ref{1:l8}]
	First of all, we notice that
	$
		\hGa_f (k) - \tGa_f (k) 
		= \frac{1}{T-k} \sum_{t=1}^{T-k} \left( \hf_t \hf_{t+k} - \pf_t \pf_{t+k} \right)
		= \frac{1}{T-k} \sum_{t=1}^{T-k} \left[ \left(\hf_t - \pf_t\right) \right.$\\$\left. \hf_{t+k} +  \pf_t \left(\hf_{t+k} - \pf_{t+k}\right) \right].
	$
	Hence,
	\begin{align*}
		&\left\| \hGa_f (k) - \tGa_f (k) \right\|_2 
		\le \left\| \frac{1}{T-k} \sum_{t=1}^{T-k}  \left(\hf_t - \pf_t\right) \hf_{t+k} \right\|_2 + \left\| \frac{1}{T-k} \sum_{t=1}^{T-k} \pf_t \left(\hf_{t+k} - \pf_{t+k}\right) \right\|_2 \\
		\le&  \frac{1}{T-k} \sum_{t=1}^{T-k}  \left\| \left(\hf_t - \pf_t\right) \hf_{t+k} \right\|_2 +  \frac{1}{T-k} \sum_{t=1}^{T-k} \left\| \pf_t \left(\hf_{t+k} - \pf_{t+k}\right) \right\|_2. 	
	\end{align*}
	It is sufficient to consider only one of the two terms on the right-hand side above since the other one can be dealt with in precisely the same way. %as $\left\| \hf_t \right\|_2 = \frac{1}{N} \left\| \hQ^\top \py_t \right\|_2 \asymp 1 \asymp \left\| \pf_t \right\|_2$. 
	For the first term on the right-hand side above, notice that under the factor model defined in~\eqref{1:e2}, we have
	\begin{align*}
		\hf_t - \pf_t 
		&= \frac{1}{N} \hQ^\top \py_t - \pf_t
		= \frac{1}{N} \left(\hQ- \pQ \right)^\top \py_t + \frac{1}{N} \pQ^\top \py_t - \pf_t\\
		&= \frac{1}{N} \left(\hQ- \pQ \right)^\top \py_t + \frac{1}{N} \pQ^\top \py_t - \frac{1}{N} \pQ^\top \pQ \pf_t
		%&= \frac{1}{N} \hQ^\top \left(\pQ \pf_t +\pu_t\right) - \pf_t\\
		%&= \frac{1}{N} \left(\hQ^\top \pQ - \pQ^\top \pQ\right) \pf_t + \frac{1}{N} \hQ^\top \pu_t\\
		= \frac{1}{N} \left(\hQ- \pQ \right)^\top \py_t + \frac{1}{N} \pQ^\top \pu_t.
	\end{align*}
	Hence
	$
		\left\| \hf_t - \pf_t \right\|_2
		\le \left\| \frac{1}{N} \left(\hQ- \pQ \right)^\top \py_t \right\|_2 + \left\| \frac{1}{N} \pQ^\top \pu_t \right\|_2, 	
	$
	by the triangular inequality. To study $\left\| \frac{1}{N} \pQ^\top \pu_t \right\|_2,$ first consider the random variables $\frac{1}{\sqrt{N}} \pq_i^\top \pu_t$ for each $ \frac{1}{\sqrt{N}} \pq_i$ in $ \frac{1}{\sqrt{N}}  \pQ = \left(\frac{1}{\sqrt{N}} \pq_1,\frac{1}{\sqrt{N}}\pq_2,...,\frac{1}{\sqrt{N}} \pq_r\right)$, where $\frac{1}{\sqrt{N}} \pq_i$ for $i=1,2,...,r$ are unscaled eigenvectors estimated from $\hL$. Observe that $\mathbb{E} \left(\frac{1}{\sqrt{N}}  \pq_i^\top \pu_t\right) = 0$ and $\mathbb{V}\left( \frac{1}{\sqrt{N}}  \pq_i^\top \pu_t\right) = \frac{1}{N}  \pq_i^\top  \pS_u \pq_i \le \lambda_\text{max}\left(\pS_u\right) < \infty,$ since $\left\| \frac{1}{\sqrt{N}} \pq_i \right\|_2 = 1$ and $\lambda_\text{max}\left(\pS_u\right)$ is the largest eigenvalue of $\pS_u$. Consequently, $\frac{1}{\sqrt{N}} \pq_i^\top \pu_t = O_P\left(1\right)$ and $\left\| \frac{1}{N} \pQ^\top \pu_t \right\|_2 = \sqrt{\frac{1}{{N}} \sum_{i=1}^{r} \left(\frac{1}{\sqrt{N}} \pq_i^\top \pu_t\right)^2} = O_P\left(N^{-1/2}\right)$, as the eigenvalues of $\pS_u$ are assumed to be bounded when $N \to \infty$ under Assumption \ref{1:c1}.
	
	Recall that $\left\| \hQ - \pQ \right\|_2 = O_P \left(N^{1/2}T^{-1/2}\right)$ by Lemma~\ref{1:l2}, we then have 
    \begin{align*}
        \left\| \frac{1}{N} \left(\hQ- \pQ \right)^\top \py_t \right\|_2 \le  \frac{1}{N} \left\| \left(\hQ- \pQ \right)^\top \right\|_2 \left\| \py_t \right\|_2 = O_P\left(T^{-1/2}\right) ,   \end{align*}
    and 
	\begin{align*}
		\left\| \hf_t - \pf_t \right\|_2
		\le \left\| \frac{1}{N} \left(\hQ- \pQ \right)^\top \py_t \right\|_2 + \left\| \frac{1}{N} \pQ^\top \pu_t \right\|_2
		= O_P\left(N^{-1/2} + T^{-1/2}\right),
	\end{align*}
	uniformly for $t$. Finally, we can conclude that
	\begin{align*}
		\left\| \hGa_f (k) - \tGa_f (k) \right\|_2
		\le  \frac{1}{T-k} \sum_{t=1}^{T-k}  \left\|\left(\hf_t - \pf_t\right)  \hf_{t+k} \right\|_2 
		+  \frac{1}{T-k} \sum_{t=1}^{T-k} \left\| \pf_t\left(\hf_{t+k} - \pf_{t+k}\right)  \right\|_2 
		= O_P\left(N^{-1/2} + T^{-1/2}\right).		
	\end{align*}
\end{proof}

\subsection{Auxiliary results for AR-sieve bootstrap of factor models}
\begin{lemma} \label{1:l9} 
	Let $\tA_{p} =\left(\tA_{1,p}, \tA_{2,p}, ..., \tA_{p,p}\right)$ be the matrix of the Yule-Walker estimators of the finite predictor coefficients on true factors $\{\pf_t\}$, and $\hA_{p} = \left(\hA_{1,p}, \hA_{2,p}, ..., \hA_{p,p}\right)$ be the matrix of the Yule-Walker estimators of the finite predictor coefficients on estimated factors $\{\hf_t\}$, then
	\begin{align*}
		\left\| \hA_{p}- \tA_{p} \right\|_F = O_P\left(p^{4} \left(N^{-1/2} + T^{-1/2}\right)\right).
	\end{align*}
\end{lemma}

\begin{proof}[Proof of Lemma \ref{1:l9}]
	Recall that the Yule-Walker estimators are solved from the Yule-Walker equations on the finite predictors' coefficient matrices as 
	$
		\pA_{p} = \left(\pA_{1,p},\pA_{2,p},...,\pA_{p,p}\right) = \pPi_1 \pPi_{0,p}^{-1},
	$
	where $\pPi_1 = \left(\pGa_{\pf}(1),\pGa_{\pf}(2),...,\pGa_{\pf}(p)\right)$ is an $r \times (rp)$ block matrix of autocovariance matrices and
	\begin{align*}
		\pPi_{0,p} = 
		\begin{pmatrix}
			\pGa_{\pf}(0)  &\pGa_{\pf}(1)  & \cdots   &\pGa_{\pf}(p-1)\\
			\pGa_{\pf}(-1)  &\pGa_{\pf}(0)  & \cdots   &\pGa_{\pf}(p-2)\\
			\vdots   &        \vdots   &       \ddots        & \vdots\\
			\pGa_{\pf}(-p+1)  &\pGa_{\pf}(-p+2)  & \cdots   &\pGa_{\pf}(0)\\
		\end{pmatrix},
	\end{align*}
	is then an $(rp) \times (rp)$ block matrix of autocovariance matrices \citep{brockwell_time_1991}. Write $\hA_{p} = \left(\hA_{1,p},\hA_{2,p},...,\hA_{p,p}\right) = \hPi_1 \hPi_{0,p}^{-1} $ with $\hPi_1$ and $\hPi_{0,p}$ the same matrices as $\pPi_1$ and $\pPi_{0,p}$ but defined based on $\hGa_f$ rather than $\pGa_f$. Similarly, $\tA_{p} = \left(\tA_{1,p},\tA_{2,p},...,\tA_{p,p}\right) = \tPi_1 \tPi_{0,p}^{-1} $ with $\tPi_1$ and $\tPi_{0,p}$ defined based on $\tGa_f$ rather than $\pGa_f$. Recall that $\hGa_f$ and $\tGa_f$ are sample lag-$k$ autocovariance matrices defined in Lemma \ref{1:l8}, then we have
	\begin{align}\label{eb1}
		\left\| \hA_{p}- \tA_{p} \right\|_F \le \left\| \hPi_{0,p}^{-1} - \tPi_{0,p}^{-1}  \right\|_F  \left\| \hPi_1\right\|_F +  \left\| \tPi_{0,p}^{-1}  \right\|_F    \left\| \hPi_1 - \tPi_1 \right\|_F.
	\end{align}
	To find $\left\| \tPi_{0,p}^{-1}  \right\|_F$, we first compute $\left\| \pPi_{0,p}^{-1}  \right\|_F$.
	Recall the recursive derivation based on the partitioned inverse formula for $\pPi_{0,p+1}^{-1}$ as in \citet{sowell_decomposition_1989},
	\begin{align*} 
		\pPi_{0,p+1}^{-1} 
		= 
		\begin{pmatrix}	
			\pPi_{0,p}^{-1} + \mathcal{J}_p \overline{\pA}_p \overline{\pv}_p^{-1} \overline{\pA}_p^{\top} \mathcal{J}_p & -\mathcal{J}_p \overline{\pA}_p \overline{\pv}_p^{-1}\\
			-\overline{\pv}_p^{-1} \overline{\pA}_p^{\top}  \mathcal{J}_p & \overline{\pv}_p^{-1}\\
		\end{pmatrix}
		=	
		\begin{pmatrix}	
			\pPi_{0,p}^{-1}  & 0\\
			0 & 0\\
		\end{pmatrix}
		+
		\begin{pmatrix}	
			0 & -\mathcal{J}_p \overline{\pA}_p \overline{\pv}_p^{-1/2}\\
			0 & \overline{\pv}_p^{-1/2}\\
		\end{pmatrix}	
		\begin{pmatrix}	
			0 & 0\\
			-\overline{\pv}_p^{-1/2} \overline{\pA}_p^{\top}  \mathcal{J}_p & \overline{\pv}_p^{-1/2}\\
		\end{pmatrix}, \numberthis \label{1:2.91}
	\end{align*}	
	where $\mathcal{J}_p = \pJ_p \otimes \pI_r $ with $\pJ_p$ the $p \times p$ matrix with ones on the anti-diagonal and $\pI_r$ the $r \times r$ identity matrix, $\overline{\pv} = \mathbb{E}\left(\pf_t - \sum_{l=1}^p \overline{\pA}_{l,p} \pf_{t+l}\right) \left(\pf_t - \sum_{l=1}^p \overline{\pA}_{l,p} \pf_{t+l}\right)^\top$ and $\overline{\pA}_{p} = \left(\overline{\pA}_{1,p}^\top, \overline{\pA}_{2,p}^\top,...,\overline{\pA}_{p,p}^\top\right)$ the coefficient matrices minimizing the forward prediction variance $\mathbb{E}\left(\pf_t - \sum_{l=1}^p \mathbf{F}_{l,p} \pf_{t+l}\right) \left(\pf_t - \sum_{l=1}^p \mathbf{F}_{l,p} \pf_{t+l}\right)^\top$. 
    Denoted by $\mathcal{S}_p$ the second term on the right-hand side of~\eqref{1:2.91}. %we can then get the recursive expression of $\pPi_{0,p}^{-1}$ as
	% $
	% 	\pPi_{0,p}^{-1}
	% 	=
	% 	\begin{pmatrix}	
	% 		\pGa_{\pf}(0)^{-1}  & 0\\
	% 		0 & 0\\
	% 	\end{pmatrix}
	% 	+ \sum_{l=1}^{p-1} \mathcal{S}_l.
	% $
	% For $\mathcal{S}_l$, note that
	% \begin{align*}
	% 	\left\| \mathcal{S}_l \right\|_F 
	% 	\le \left\| \overline{\pv}_l^{-1/2} \right\|_F^2 \left(1+ \left\| \mathcal{J}_l \overline{\pA}_l \right\|_F\right)^2
	% 	\le \left\| \overline{\pv}_l^{-1/2} \right\|_F^2 \left(1+ \sum_{j=1}^l \left\| \overline{\pA}_{j,l} \right\|_F\right)^2
	% 	= O\left(1\right),
	% \end{align*}
	% uniformly for $l=1,2,...,p$, where we use the definition of $\overline{\pv}_l$ and Lemma \ref{1:l6}. Hence $\left\| \sum_{l=1}^{p-1} \mathcal{S}_l \right\|_F \le \sum_{l=1}^{p-1} \left\| \mathcal{S}_l \right\|_F = O\left(p\right)$. Besides, $\left\| \pGa_{\pf}(0)^{-1} \right\|_F = \sqrt{\sum_{i=1}^r \lambda_i^{-2}} \le \sqrt{r} \lambda_{\text{min}}^{-1} = O\left(1\right)$, where $\lambda_i$ is the $i$\textsuperscript{th} eigenvalue of $\pGa_{\pf}(0)$, $\lambda_{\text{min}}$ is the smallest eigenvalue of $\pGa_{\pf}(0)$ and we use Assumption \ref{1:a2} that $\pGa_{\pf}(0)$ is full rank. Thus, we have shown $\left\| \pPi_{0,p}^{-1} \right\|_F = O\left(p\right).$
    \color{darkblue}
    $\mathcal{S}_p$ represents the update term when expanding the inverse covariance matrix from dimension $rp$ to $r(p+1)$. Instead of a direct matrix summation, which involves varying dimensions, we consider the Frobenius norm. By applying the triangular inequality to the recursive formula~\eqref{1:2.91}, we have
    \begin{equation*}
        \|\pPi_{0,p+1}^{-1}\|_F \le \left\| \begin{pmatrix} \pPi_{0,p}^{-1} & 0 \\ 0 & 0 \end{pmatrix} \right\|_F + \|\mathcal{S}_p\|_F = \|\pPi_{0,p}^{-1}\|_F + \|\mathcal{S}_p\|_F.
    \end{equation*}
    By iterating this inequality from $l=1$ to $p-1$, we obtain
    \begin{equation*}
        \|\pPi_{0,p}^{-1}\|_F \le \|\pGa_f(0)^{-1}\|_F + \sum_{l=1}^{p-1} \|\mathcal{S}_l\|_F.
    \end{equation*}
    For the term $\mathcal{S}_l$, note that
    \begin{align*}
		\left\| \mathcal{S}_l \right\|_F 
		\le \left\| \overline{\pv}_l^{-1/2} \right\|_F^2 \left(1+ \left\| \mathcal{J}_l \overline{\pA}_l \right\|_F\right)^2
		\le \left\| \overline{\pv}_l^{-1/2} \right\|_F^2 \left(1+ \sum_{j=1}^l \left\| \overline{\pA}_{j,l} \right\|_F\right)^2
		= O\left(1\right),
	\end{align*}
	uniformly for $l=1,2,...,p$, where we use the definition of $\overline{\pv}_l$ and Lemma \ref{1:l6}.
    Therefore, the summation of norms is bounded by
    \begin{equation*}
        \sum_{l=1}^{p-1} \|\mathcal{S}_l\|_F = O(p).
    \end{equation*}
    Also, $\|\pGa_f(0)^{-1}\|_F \le \sqrt{r}\lambda_{\min}^{-1} = O(1)$ due to the full rank assumption. Consequently, we have shown that $\|\pPi_{0,p}^{-1}\|_F = O(p)$.
    \color{black}
	
    To find $\| \hPi_{0,p}^{-1} - \tPi_{0,p}^{-1}  \|_F$, note that for invertible matrices $\hPi_{0,p}$ and $\tPi_{0,p}$,
	\begin{align*}
		\left\| \hPi_{0,p}^{-1} - \tPi_{0,p}^{-1}  \right\|_F 
		&= \left\| \hPi_{0,p}^{-1} (\tPi_{0,p}-\hPi_{0,p}) \tPi_{0,p}^{-1}  \right\|_F 
		= \left\| (\hPi_{0,p}^{-1} - \tPi_{0,p}^{-1})  (\tPi_{0,p}-\hPi_{0,p}) \tPi_{0,p}^{-1}  + \tPi_{0,p}^{-1} (\tPi_{0,p}-\hPi_{0,p}) \tPi_{0,p}^{-1}  \right\|_F \\
		&\le \left\| \hPi_{0,p}^{-1} - \tPi_{0,p}^{-1} \right\|_F \left\| \tPi_{0,p}-\hPi_{0,p} \right\|_F \left\| \tPi_{0,p}^{-1} \right\|_F  
		+ \left\| \tPi_{0,p}-\hPi_{0,p}\right\|_F \left\| \tPi_{0,p}^{-1}  \right\|_F^2.
	\end{align*}
	And for large enough $N$ and $T$ such as $\left\| \hGa_f (k) - \tGa_f (k) \right\|_2 \to 0$ and $\left\| \tPi_{0,p}-\hPi_{0,p}\right\|_F \to 0$ in probability, we can write 
	\begin{align*}
		\left\| \hPi_{0,p}^{-1} - \tPi_{0,p}^{-1}  \right\|_F 
		&\le \frac{\left\| \tPi_{0,p}^{-1}  \right\|_F^2   \left\| \tPi_{0,p}-\hPi_{0,p}\right\|_F }{1- \left\| \tPi_{0,p}^{-1} \right\|_F  \left\| \tPi_{0,p}-\hPi_{0,p} \right\|_F } 
		\le \frac{\left\| \pPi_{0,p}^{-1}  \right\|_F^2   \left\| \tPi_{0,p}-\hPi_{0,p}\right\|_F }{1- \left\| \tPi_{0,p}^{-1} \right\|_F  \left\| \tPi_{0,p}-\hPi_{0,p} \right\|_F } 
		+ \frac{\left\| \tPi_{0,p}^{-1} -\pPi_{0,p}^{-1} \right\|_F^2   \left\| \tPi_{0,p}-\hPi_{0,p}\right\|_F }{1- \left\| \tPi_{0,p}^{-1} \right\|_F  \left\| \tPi_{0,p}-\hPi_{0,p} \right\|_F }\\
		&=O_P\left(\left\| \pPi_{0,p}^{-1}  \right\|_F^2   \left\| \tPi_{0,p}-\hPi_{0,p}\right\|_F \right),
	\end{align*}
	where the last equation follows since when $N,T \to \infty$, $\| \tPi_{0,p}-\hPi_{0,p}\|_F \to 0$ in probability, and the first term in the second inequality is the leading term.
	In addition, we have
	\begin{align}\label{1:ee3}
		\left\| \tPi_{0,p}-\hPi_{0,p}\right\|_F 
		 \le \sum_{l=1}^{p} \sum_{j=1}^{p} \left\| \hGa_f (l-j) - \tGa_f (l-j)  \right\|_F
		 \le p^2 \max_{|k| \le p-1} \left\| \hGa_f (k) - \tGa_f (k) \right\|_F
		=O_P\left(p^{\comm{2}} \left(N^{-1/2} + T^{-1/2}\right)\right),
	\end{align}
	where for $r \times r$ matrices $\hGa_f (k)$ and $\tGa_f (k)$, $\left\| \hGa_f (k) - \tGa_f (k) \right\|_F \asymp \left\| \hGa_f (k) - \tGa_f (k) \right\|_2 = O_P\left(N^{-1/2} + T^{-1/2}\right)$ as shown in Lemma \ref{1:l8}. Therefore, with~\eqref{1:ee3} we can conclude that
	\begin{align}\label{1:ee4}
		\left\| \hPi_{0,p}^{-1} - \tPi_{0,p}^{-1}  \right\|_F 
		 = O_P\left(\left\| \pPi_{0,p}^{-1}  \right\|_F^2   \left\| \tPi_{0,p}-\hPi_{0,p}\right\|_F \right)
		= O_P\left(p^{4} \left(N^{-1/2} + T^{-1/2}\right)\right).
	\end{align}
	Lastly,
	\begin{align}\label{1:ee5}
		\left\| \hPi_1\right\|_F 
		\le \sum_{k=1}^{p} \left\| \hGa_f (k)\right\|_F 
		\le \sum_{k=1}^{p} \left\| \pGa_f (k)\right\|_F + \sum_{k=1}^{p} \left\| \hGa_f (k) - \pGa_f (k) \right\|_F
		=O\left(1\right) + O_P\left(p\left(N^{-1/2} + T^{-1/2}\right)\right),
	\end{align}
	where the first term follows from the summability condition in Assumption~\ref{1:a2}. Moreover,
	$
		\left\| \hPi_1 - \tPi_1 \right\|_F 
		\le \sum_{k=1}^{p} \left\| \hGa_f (k) - \tGa_f (k) \right\|_F
		=  O_P\left(p\left(N^{-1/2} + T^{-1/2}\right)\right).
	$
	Hence, we can conclude that the first term in~\eqref{eb1} is the leading term, and 
	$
		\left\| \hA_{p}- \tA_{p} \right\|_F = O_P\left(p^{4} \left(N^{-1/2} + T^{-1/2}\right)\right),
	$
	by~\eqref{1:ee4} and~\eqref{1:ee5}.
\end{proof}

\begin{lemma}\label{1:l6} 
	Let $\{\pf_t\}$ be factor processes fulfilling Assumptions \ref{1:c1} and \ref{1:a2} for some $\gamma \ge 0.$ Write $\left\{\pA_{l,p}, l = 1,2,...,p\right\}$ and $\left\{\pPsi_{l,p}, l = 1,2,...,p\right\}$ as the finite predictor coefficients matrices of the AR coefficients  $\left\{\pA_{l}, l \in \mathbb{N}\right\}$ and the MA coefficients $\left\{\pPsi_{l}, l \in \mathbb{N}\right\}$ as in~\eqref{1:e2} and~\eqref{1:e3}, respectively.
\begin{enumerate}
		\item[(i)] Norm summability:
		The coefficients matrices $\pA_{l}$ and $\pPsi_{l}$ \comm{fulfill} the following summability properties: $\sum_{l=1}^{\infty} (1+l)^\gamma \left\| \pA_l \right\|_F < \infty$ and $\sum_{l=1}^{\infty} (1+l)^\gamma \left\| \pPsi_l \right\|_F < \infty.$
		\item[(ii)] (Lemma 3.1 of \cite{meyer_vector_2015}) 
		For some $\gamma \ge 0$ as in Assumption~\ref{1:a2}, there exist $p_0 \in \mathbb{N}$ and $d < \infty$ such that
		\begin{align*}
			\sum_{l=1}^{p} (1+l)^\gamma \left\| \pA_{l,p} - \pA_{l} \right\|_F \le d \sum_{l=p+1}^{\infty} (1+l)^\gamma \left\| \pA_{l} \right\|_F,\ \text{for } p \ge p_0,
		\end{align*}
		and the right side converges to 0 when $p \to \infty$.
		\item[(iii)] (Lemma 3.2 of \cite{meyer_vector_2015}) 
		Let $\pA_p (z) \coloneqq \boldsymbol{I}_r -  \sum_{l=1}^{p} \pA_{l,p} z^l$, then there exist $p_1 \in \mathbb{N}$ and $c < \infty$ such that
		\begin{align*}
			\inf_{|z| \le 1+ 1/p} \left|\det\left(\pA_p (z)\right)\right| \ge c,\quad \text{for } p \ge p_1.
		\end{align*}
		\item[(iv)] (Lemma 3.3 of \cite{meyer_vector_2015})
		Let $\{\pPsi_{l,p}, l \in \mathbb{N}\}$ be the power series coefficients matrices of $\left( \boldsymbol{I}_r -  \sum_{l=1}^{p} \pA_{l,p} z^l \right)^{-1}$, for $|z| \le 1$. For $p_1$ as defined in (iii) and some $\gamma \ge 0$ in Assumption~\ref{1:a2}, there exist $p_2 \ge p_1$ and $d < \infty$ such that
		\begin{align*}
			\sum_{l=1}^{\infty} (1+l)^\gamma \left\| \pPsi_{l,p} - \pPsi_{l} \right\|_F \le d \sum_{l=p+1}^{\infty} (1+l)^\gamma \left\| \pA_{l} \right\|_F,\quad \text{for } p \ge p_2,
		\end{align*}
		and the right side converges to 0 when $p \to \infty$.
\end{enumerate}
\end{lemma}
Lemma~\ref{1:l6} (ii) is the vector form of Baxter's inequality on the AR coefficient matrices $\{\pA_l\}$ and its finite predictor coefficient matrices $\{\pA_{l,p}\}$, while Lemma~\ref{1:l6} (iv) relates Baxter's inequality of AR coefficients to the MA coefficient matrices $\{\pPsi_l\}$ and its finite predictor coefficient matrices $\{\pPsi_{l,p}\}$. The proofs of Lemma~\ref{1:l6} can be found in \citet{meyer_vector_2015}, hence it is omitted here.

\begin{lemma} \label{1:l7} (Lemma~3.5 of \citet{meyer_vector_2015})
Let $\{\pf_t\}$ be factor processes defined under the assumptions of Lemma~\ref{1:l6} and also fulfill Assumption~\ref{1:a3}. Define $\pPsi_{l,p}$ as the coefficients matrices in the power series of $\left( \boldsymbol{I}_r -  \sum_{l=1}^{p} \pA_{l,p} z^l \right)^{-1}$, for $|z| \le 1$ with $\pPsi_{0,q} \coloneqq \boldsymbol{I}_r$  and $\tPsi_{l,p}$ as the power series coefficients matrices of $\left( \boldsymbol{I}_r -  \sum_{l=1}^{p} \tA_{l,p} z^l \right)^{-1},\ \text{for } |z| \le 1$ with $\tPsi_{0,q} \coloneqq \boldsymbol{I}_r$. Then, there exists $p_3 \in \mathbb{N}$ such that \comm{for all $l \in \mathbb{N}$} and for all $p \ge p_3$,
\begin{equation*}
\left\| \tPsi_{l,p}-\pPsi_{l,p}\right\|_F \le \left(1+\frac{1}{p}\right)^{-l} \frac{1}{p^2} O_P\left(1\right).
\end{equation*}
\end{lemma}

The proof of Lemma~\ref{1:l7} can be found in \citet{meyer_vector_2015}.

\begin{lemma} \label{1:k1} 
Let $\{\pf_t\}$ be factor processes fulfilling Assumptions \ref{1:c1}, \ref{1:a2} ($\gamma = 1$), \ref{1:a3} and \ref{1:a4}. Define $\{\pPsi_{l,p}\}$ as the coefficients matrices in the power series of $\left( \boldsymbol{I}_r -  \sum_{l=1}^{p} \pA_{l,p} z^l \right)^{-1}$, for $|z| \le 1$ with $\pPsi_{0,q} \coloneqq \boldsymbol{I}_r$. Similarly, define $\{\tPsi_{l,p}\}$ as the power series coefficient matrices of $\left( \boldsymbol{I}_r -  \sum_{l=1}^{p} \tA_{l,p} z^l \right)^{-1},$ for $|z| \le 1$ with $\tPsi_{0,q} \coloneqq \boldsymbol{I}_r$, and $\{\hPsi_{l,p}\}$ as the power series coefficient matrices of $\left( \boldsymbol{I}_r -  \sum_{l=1}^{p} \hA_{l,p} z^l \right)^{-1},$ for $|z| \le 1$ with $\hPsi_{0,q} \coloneqq \boldsymbol{I}_r$. Then, there exists $p_3 \in \mathbb{N}$ such that for all $p \ge p_3$ as in Lemma~\ref{1:l7}. As $N, T\rightarrow\infty$, 
\begin{align*}
&\sum_{l=1}^{\infty} \left\| \tPsi_{l,p}-\pPsi_{l,p}\right\|_F = O_P\left(\frac{1}{p}\right) = o_P(1), \quad\quad
& \sum_{l=1}^{\infty} \left\| \pPsi_{l,p}-\pPsi_{l}\right\|_F = o\left(1\right),\\
&\sum_{l=1}^{\infty} \left\| \hPsi_{l,p}-\tPsi_{l,p}\right\|_F = O_P\left(p^{3/2} \left\| \hA_{p} - \tA_{p} \right\|_F\right)  = o_P(1), \quad\quad
& \sum_{l=1}^{\infty} \left\| \hPsi_{l,p}-\pPsi_{l,p}\right\|_F = o_P(1).
\end{align*}

\end{lemma}	

\begin{proof}[Proof of Lemma \ref{1:k1}]
	For large enough $N$, $T$ and $p>p_3$ as in Lemma \ref{1:l7}, $\sum_{l=1}^{\infty} \left\|\tPsi_{l,p}-\pPsi_{l,p}\right\|_F$ follows directly from Lemma~\ref{1:l7} as
	\begin{align*}
		\sum_{l=1}^{\infty} \left\| \tPsi_{l,p}-\pPsi_{l,p}\right\|_F 
		\le \frac{1}{p^2} \sum_{l=1}^{\infty}\left(1+\frac{1}{p} \right)^{-l} O_P\left(1\right)
		\le \frac{1}{p^2} \frac{p}{1+p}(1+p) O_P\left(1\right)
		=O_P\left(\frac{1}{p}\right).
	\end{align*}
	The order of $\sum_{l=1}^{\infty} \left\| \pPsi_{l,p}-\pPsi_{l}\right\|_F$ follows directly from Lemma \ref{1:l6} (i) and (iv), as
	\begin{align*}
		\sum_{l=1}^{\infty} \left\| \pPsi_{l,p}-\pPsi_{l}\right\|_F 
		\le \sum_{l=1}^{\infty} (1+l)^\gamma \left\| \pPsi_{l,p} -\pPsi_{l} \right\|_F 
		\le d \sum_{l=p+1}^{\infty} (1+l)^\gamma \left\| \pA_{l} \right\|_F
		=o\left(1\right).
	\end{align*}
	To show $\sum_{l=1}^{\infty} \left\| \hPsi_{l,p}-\tPsi_{l,p}\right\|_F = o_P\left(1\right),$ first notice that
	$
		\sum_{l=1}^{\infty} \left\|\hPsi_{l,p}-\tPsi_{l,p}\right\|_F
		\le 	\sum_{l=1}^{\infty} \sum_{u=1}^{r} \sum_{v=1}^{r} \left| \widehat{\Psi}_{l,p}^{(u,v)} - \widetilde{\Psi}_{l,p}^{(u,v)} \right|,
	$
	where $\widehat{\Psi}_{l,p}^{(u,v)}$ and $\widetilde{\Psi}_{l,p}^{(u,v)}$ are the $(u,v)$\textsuperscript{th} elements of the matrices $\hPsi_{l,p}$ and $\tPsi_{l,p}$, respectively. We then apply Cauchy's inequality for holomorphic functions on the $(u,v)$\textsuperscript{th} element of $\tPsi_{l,p}$ and $\pPsi_{l,p}$, that is
	\begin{align*}
		\left| \widehat{\Psi}_{l,p}^{(u,v)}-\widetilde{\Psi}_{l,p}^{(u,v)} \right| 
		&\le \left(1+\frac{1}{p}\right)^{-l} \max_{|z| = 1+\frac{1}{p}}  \left\| \hA_{p}^{-1}(z) - \tA_{p}^{-1}(z) \right\|_F
		\le \left(1+\frac{1}{p}\right)^{-l} \left[ \max_{|z| = 1+\frac{1}{p}} \frac{1}{|\det(\hA_{p}(z))|} \left\| \hA_{p}^{adj}(z) - \tA_{p}^{adj}(z) \right\|_F \right. \\
		&+ \left. \max_{|z| = 1+\frac{1}{p}} \left| \frac{1}{\det(\hA_{p}(z))} - \frac{1}{\det(\tA_{p}(z))} \right|  \left\| \tA_{p}^{adj}(z) \right\|_F \right]
		\eqqcolon \left(1+\frac{1}{p} \right)^{-l} \left[ \max_{|z| = 1+\frac{1}{p}}\mathcal{K}_{1,z} + \max_{|z| = 1+\frac{1}{p}}\mathcal{K}_{2,z} \right],
	\end{align*}
	where we use $\pA^{adj}$ to denote the adjugate matrix of $\pA$, and write the two terms above as $\mathcal{K}_{1,z}$ and $\mathcal{K}_{2,z}$.
	
	To study $\mathcal{K}_{1,z}$, with Assumption~\ref{1:a3}, Lemmas \ref{1:l2} and \ref{1:l9}, we show that with sufficiently large $N$ and $T$, we can choose $p>p_3$ such that $\left\| \hA_{p}- \tA_{p} \right\|_F = o_P(1) $ and $\sup_{|z|\le 1+\frac{1}{p}} \left\| \hA_{p}(z)- \tA_{p}(z) \right\|_F = o_P(1)$. Furthermore, since determinants are continuous functions of the elements, it can be extended to $\sup_{|z|\le 1+\frac{1}{p}} \left| \det\hA_{p}(z)- \det\tA_{p}(z) \right| \to 0$ in probability, with
	$
		\left|\det\left(\tA_p (z)\right)\right| \ge c\ \text{and}\ \left|\det\left(\hA_p (z)\right)\right| \ge c\ \text{in probability, for }|z|\le1+\frac{1}{p},
	$
	and for some $c>0$ as in Lemma \ref{1:l6}. Then, for $p > p_3$ and any $|z|=1+1/p$ we can show that
	\begin{align*}
		\mathcal{K}_{1,z} 
		&\le \frac{1}{c} \left\| \hA_{p}^{adj}(z) - \tA_{p}^{adj}(z) \right\|_F
		\le \frac{1}{c} \sum_{u=1}^{r} \sum_{v=1}^{r} \left| \hA_{p}^{adj}(z)^{(u,v)} - \tA_{p}^{adj}(z)^{(u,v)} \right| \\
		&\le \frac{1}{c} \sum_{u=1}^{r} \sum_{v=1}^{r} \sup_{|z|\le 1+\frac{1}{p}} \left| \det \hA_{p}^{(-v,-u)}(z) - \det \tA_{p}^{(-v,-u)}(z) \right|
		\le \frac{1}{c} \sum_{u=1}^{r} \sum_{v=1}^{r} \sup_{|z|\le 1+\frac{1}{p}} r \left\| \hA_{p}(z) - \tA_{p}(z) \right\|_F O_P\left(1\right)\\
		&\le \sup_{|z|\le 1+\frac{1}{p}} \left\| \hA_{p}(z) - \tA_{p}(z) \right\|_F,
	\end{align*}
	where $\tA_{p}^{(-v,-u)}(z)$ is a matrix generated by removing the $v$\textsuperscript{th} row and the $u$\textsuperscript{th} column of $\tA_{p}(z)$.
	
	And for $\sup_{|z|\le 1+\frac{1}{p}} \left\| \hA_{p}(z) - \tA_{p}(z) \right\|_F,$ we have
	\begin{align*}
		\sup_{|z|\le 1+\frac{1}{p}} \left\| \hA_{p}(z) - \tA_{p}(z) \right\|_F
		\le \sup_{|z|\le 1+\frac{1}{p}} \sum_{l=1}^{p}  \left\| \hA_{l,p} - \tA_{l,p} \right\|_F |\comm{z}|^l
		\le \left(1+\frac{1}{p}\right)^p \sum_{l=1}^{p}  \left\| \hA_{l,p} - \tA_{l,p} \right\|_F
		= O_P\left(\sqrt{p} \left\| \hA_{p} - \tA_{p} \right\|_F\right).
	\end{align*}
	Hence we can conclude that for $\mathcal{K}_{1,z}$,
	$
		\max_{|z| = 1+\frac{1}{p}}\mathcal{K}_{1,z}
		= O_P\left(\sqrt{p} \left\| \hA_{p} - \tA_{p} \right\|_F\right),
	$
	since the bound does not depend on $z$.
	
	For $\mathcal{K}_{2,z}$, note that $\max_{|z| = 1+\frac{1}{p}} \left\| \pA_{p}(z) \right\|_F \le (1+1/p)^p \sum_{l=1}^{p} \left\| \pA_{l,p} \right\|_F = O_P\left(1\right)$ by Lemma \ref{1:l6}, therefore, $\max_{|z| = 1+\frac{1}{p}} \left\| \tA_{p}(z) \right\|_F = O_P\left(1\right)$ by Assumption~\ref{1:a3}. Similarly, for some constants $c$,
	\begin{align*}
		\max_{|z| = 1+\frac{1}{p}} \mathcal{K}_{2,z} 
		\le \frac{1}{c^2} \max_{|z| = 1+\frac{1}{p}} \left| \det \hA_{p}(z) - \det \tA_{p}(z) \right| \left\| \tA_{p}^{adj}(z) \right\|_F
		= O_P\left(\sqrt{p} \left\| \hA_{p} - \tA_{p} \right\|_F\right).
	\end{align*}
	As a result, 
	$
		\sum_{l=1}^{\infty} \left\| \hPsi_{l,p}-\tPsi_{l,p}\right\|_F 
		\le \sum_{l=1}^{\infty} \sum_{u=1}^{r} \sum_{v=1}^{r} | \widehat{\Psi}_{l,p}^{(u,v)}-\widetilde{\Psi}_{l,p}^{(u,v)} |
		= O_P\left(p^{3/2} \left\| \hA_{p} - \tA_{p} \right\|_F\right).
	$
	Then, we can conclude that
	\begin{align*}
		\sum_{l=1}^{\infty} \left\| \hPsi_{l,p}-\pPsi_{l,p}\right\|_F
		\le 
		\sum_{l=1}^{\infty} \left\| \tPsi_{l,p}-\pPsi_{l,p}\right\|_F +
		\sum_{l=1}^{\infty} \left\| \hPsi_{l,p}-\tPsi_{l,p}\right\|_F
		= O_P\left(\frac{1}{p}\right) + O_P\left(p^{3/2} \left\| \hA_{p}-\tA_{p} \right\|_F\right).
	\end{align*}
\end{proof}

\begin{lemma} \label{1:k2} 
	Let $\{\pf_t\}$ be factor processes defined under the assumptions of Lemma \ref{1:k1}. Write $\pe_t = \pf_t - \sum_{l=1}^{\infty} \pA_{l}\pf_{t-l}$, $\pe_{t,p} = \pf_t - \sum_{l=1}^{p} \pA_{l,p}\pf_{t-l}$, $\te_{t,p} = \pf_t - \sum_{l=1}^{p} \tA_{l,p}\pf_{t-l}$ and $\he_{t,p} = \hf_t - \sum_{l=1}^{p} \hA_{l,p}\hf_{t-l}$. Furthermore, define the corresponding covariance $\tS_{e,p} = \mathbb{E}^\ast (\te_{t,p} - \overline{\te}_{{T'},p})(\te_{t,p} - \overline{\te}_{{T'},p})^\top$ with $\overline{\te}_{{T'},p} = \frac{1}{{T'}} \sum_{t=p+1}^T \te_{t,p}$, and $\hS_{e,p} = \mathbb{E}^\ast (\he_{t,p} - \overline{\he}_{{T'},p})(\he_{t,p} - \overline{\he}_{{T'},p})^\top$ with $\overline{\he}_{{T'},p} = \frac{1}{{T'}} \sum_{t=p+1}^T \he_{t,p}$, where $\mathbb{E}^\ast$ is the expectation defined on the measure of assigning probability $\frac{1}{{T'}}$ to each observation. 
	
	If we additionally assume that the empirical distribution of $\{\pe_t\}$ converges weakly to the distribution function of $\mathcal{L}(\pe_t)$, then, there exists $p_3 \in \mathbb{N}$ such that for all $p \ge p_3$ as in Lemma~\ref{1:l7}, when $N \to \infty$ and $T \to \infty$,
	\begin{align*}
		&\left\|\tS_{e,p}-\pS_{e,p}\right\|_F   = o_P(1),&\quad\quad
		\left\|\pS_{e,p}-\pS_{e}\right\|_F   = o(1),\\
		&\left\|\hS_{e,p}-\tS_{e,p}\right\|_F = O_P\left(p^{3/2} \left\|  \hA_{p} - \tA_{p} \right\|_F\right)  = o_P(1),&\quad\quad
		\left\|\hS_{e,p}-\pS_{e,p}\right\|_F  = o_P(1).
	\end{align*}	
\end{lemma}	

\begin{proof}[Proof of Lemma \ref{1:k2}]
	To show $\left\| \tS_{e,p}-\pS_{e,p} \right\|_F \to 0$ in probability, first note that by definition,
	\begin{align*}
		\left\| \tS_{e,p}-\pS_{e,p} \right\|_F 
		&= \left\| \frac{1}{{T'}} \sum_{t=p+1}^{T} \left( \te_{t,p} \te_{t,p}^\top - \pe_{t,p} \pe_{t,p}^\top \right) \right\|_F 
		+ \left\| \frac{1}{{T'}} \sum_{t=p+1}^{T} \pe_{t,p} \pe_{t,p}^\top - \mathbb{E}\left( \pe_{t,p} \pe_{t,p}^\top\right) \right\|_F 
		+ \left\| \overline{\te}_{{T'},p} \overline{\te}_{{T'},p}^\top \right\|_F\\
		&\eqqcolon \mathcal{E}_1 + \mathcal{E}_2 + \mathcal{E}_3,
	\end{align*}
	with straightforward notations for $\mathcal{E}_1,\ \mathcal{E}_2$ and $\mathcal{E}_3$. Next, we show that the three terms above converge to zero in probability. For $\mathcal{E}_1$, we know that by triangular inequality,
	$
		\mathcal{E}_1 \le \left\| \frac{1}{{T'}} \sum_{t=p+1}^{T} \left(\te_{t,p}  - \pe_{t,p}\right) \te_{t,p}^\top \right\|_F 
		+ \left\| \frac{1}{{T'}} \sum_{t=p+1}^{T} \pe_{t,p} \left( \te_{t,p} - \pe_{t,p} \right)^\top\right\|_F
		\eqqcolon \mathcal{E}_{1,1} + \mathcal{E}_{1,2},
	$
	with obvious notations for $\mathcal{E}_{1,1}$ and $\mathcal{E}_{1,2}$. It is then sufficient to show $\mathcal{E}_{1,1} \to 0$ in probability since $\mathcal{E}_{1,2}$ can be dealt with similarly. We can now bound $\mathcal{E}_{1,1}$ by
	\begin{align*}
		\mathcal{E}_{1,1} 
		&\le \left\| \frac{1}{{T'}} \sum_{t=p+1}^{T} \sum_{l=1}^p \left(\tA_{l,p}  - \pA_{l,p}\right) \pf_{t-l} \te_{t,p}^\top \right\|_F 
		+  \left\| \frac{1}{{T'}} \sum_{t=p+1}^{T} \sum_{l=1}^p \left(\pA_{l,p}  - \pA_{l}\right) \pf_{t-l} \te_{t,p}^\top \right\|_F \\
		&+ \left\| \frac{1}{{T'}} \sum_{t=p+1}^{T} \sum_{l=p+1}^\infty \pA_{l} \pf_{t-l} \te_{t,p}^\top \right\|_F.
	\end{align*}
	Since both $\{\pf_t\}$ and $\{\te_{t,p}\}$ are $r \times 1$ vectors, by Assumption~\ref{1:a3} and Lemma \ref{1:l6}, we have
	\begin{align*}
		\mathcal{E}_{1,1} = O_P\left(\left\| \sum_{l=1}^p \left(\tA_{l,p}  - \pA_{l,p}\right) \right\|_F +  \sum_{l=p+1}^\infty (1+l)\left\| \pA_{l} \right\|_F\right),
	\end{align*}
	which tends to zero in probability. 
	$\mathcal{E}_2 \to 0$ in probability can be shown similarly, since $\{\pf_t\}$ is stationary.
 
 For $\mathcal{E}_3$, first write that
	\begin{align*}
		\mathcal{E}_3 
		= \left\| \overline{\te}_{{T'},p} \overline{\te}_{{T'},p}^\top \right\|_F 
		\le \left\| \left( \overline{\te}_{{T'},p} - \overline{\pe}_{{T'},p}\right) \left( \overline{\te}_{{T'},p} - \overline{\pe}_{{T'},p}\right)^\top \right\|_F + 2 \left\| \left( \overline{\te}_{{T'},p} - \overline{\pe}_{{T'},p}\right) \overline{\pe}_{{T'},p}^\top \right\|_F + \left\| \overline{\pe}_{{T'},p} \overline{\pe}_{{T'},p}^\top \right\|_F,
	\end{align*}
	where $\left\| \overline{\pe}_{{T'},p} \right\| = O_P\left(\left({T'}\right)^{-1/2}\right)$. Hence it is sufficient to consider $\left\| \overline{\te}_{{T'},p} - \overline{\pe}_{{T'},p} \right\|$ as
	\begin{align*}
		&\left\| \overline{\te}_{{T'},p} - \overline{\pe}_{{T'},p} \right\|
		= \left\| \frac{1}{{T'}} \sum_{t=p+1}^{T} \left(\te_{{T'},p} - \pe_{{T'},p}\right) \right\| 
		= \left\| \frac{1}{{T'}} \sum_{t=p+1}^{T} \left(\sum_{l=1}^{p} \tA_{l,p} \pf_{t-l} - \sum_{l=1}^{\infty} \pA_{l} \pf_{t-l}\right)\right\| \\
		&\le \left\| \frac{1}{{T'}} \sum_{t=p+1}^{T} \sum_{l=1}^{p} \left(\tA_{l,p} -\pA_{l,p}\right) \pf_{t-l} \right\| 
		+ \left\| \frac{1}{{T'}} \sum_{t=p+1}^{T} \sum_{l=1}^{p} \left(\pA_{l,p} -\pA_{l}\right) \pf_{t-l} \right\| 
		+ \left\| \frac{1}{{T'}} \sum_{t=p+1}^{T} \sum_{l=p+1}^{\infty} \pA_{l} \pf_{t-l} \right\| \\
		&=O_P\left(\left\| \sum_{l=1}^p \left(\tA_{l,p}  - \pA_{l,p}\right) \right\|_F \right) + O_P\left(\sum_{l=p+1}^{\infty} (1+l) \left\| \pA_{l} \right\|_F\right) \overset{p}{\to} 0,
	\end{align*}
	where the last line follows from Assumption~\ref{1:a3} and Lemma \ref{1:l6}, and we use the same arguments for $\mathcal{E}_{1,1}$ as above. Therefore, we can conclude that $\left\| \tS_{e,p}-\pS_{e,p}\right\|_F \to 0$ in probability.
	
	To see $\left\| \pS_{e,p}-\pS_{e}\right\|_F  \to 0$, note that
	$
		\left\| \pS_{e,p}-\pS_{e} \right\|_F
		= \left\| \mathbb{E} \left( \pe_{t,p}\pe_{t,p}^\top - \pe_{t}\pe_{t}^\top\right) \right\|_F 
		\le \left\| \mathbb{E} \left\{ \left( \pe_{t,p}- \pe_{t}\right)\pe_{t,p}^\top \right\} \right\|_F + \left\| \mathbb{E} \left\{ \pe_{t,p}\left( \pe_{t,p} - \pe_{t}\right)^\top \right\} \right\|_F.
	$
	Hence it suffices to show $\left\| \mathbb{E} \left\{ \left( \pe_{t,p}- \pe_{t}\right)\pe_{t,p}^\top \right\} \right\|_F \to 0$. For this, by the triangle inequality, we have
	$
		\left\| \mathbb{E} \left\{ \left( \pe_{t,p}- \pe_{t}\right)\pe_{t,p}^\top \right\} \right\|_F 
		\le  \left\| \mathbb{E} \sum_{l=1}^{p} \left( \pA_{l,p}- \pA_{l}\right)\pf_{t-l}\pe_{t,p}^\top \right\|_F +  \left\| \mathbb{E}  \sum_{l=p+1}^{\infty} \pA_{l}\pf_{t-l}\pe_{t,p}^\top \right\|_F 
		= O\left( \sum_{l=1}^{p} \left\| \pA_{l,p} \right.\right.$\\$\left.\left.- \pA_{l} \right\|_F\right) + O\left( \sum_{l=p+1}^{\infty} \left\| \pA_{l}\right\|_F\right) \to 0,
	$
	where we stress the fact that $\|\pf_t\| \asymp \| \pe_{t,p} \| \asymp 1$ and use the results in Lemma~\ref{1:l6}.
	
	With similar arguments, we can show that $\left\| \hS_{e,p}-\tS_{e,p}\right\|_F \to 0$ in probability. Firstly, notice that $\left( \hS_{e,p}-\tS_{e,p}\right)$ can be expressed as
	\begin{align*}
		&\hS_{e,p}-\tS_{e,p}
		=  \frac{1}{{T'}} \sum_{t=p+1}^{T} \left[  \left( \he_{t,p} - \overline{\he}_{{T'},p} \right) \left( \he_{t,p} - \overline{\he}_{{T'},p} \right)^\top - \left( \te_{t,p} - \overline{\te}_{{T'},p} \right) \left( \te_{t,p} - \overline{\te}_{{T'},p} \right)^\top \right]   \\
		=& \frac{1}{{T'}} \sum_{t=p+1}^{T} \left[ \left( \he_{t,p} - \overline{\he}_{{T'},p} \right) - \left( \te_{t,p} - \overline{\te}_{{T'},p} \right)\right] \left( \he_{t,p} - \te_{t,p} \right)^\top  
		- \frac{1}{{T'}} \sum_{t=p+1}^{T} \left[ \left( \he_{t,p} - \overline{\he}_{{T'},p} \right) - \left( \te_{t,p} - \overline{\te}_{{T'},p} \right)\right] \left( \overline{\he}_{{T'},p} - \overline{\te}_{{T'},p} \right)^\top \\
		+& \frac{1}{{T'}} \sum_{t=p+1}^{T}\left[  \left( \he_{t,p} - \overline{\he}_{{T'},p} \right) \left( \te_{t,p} - \overline{\te}_{{T'},p} \right)^\top\right]  
		+ \frac{1}{{T'}} \sum_{t=p+1}^{T}\left[  \left( \te_{t,p} - \overline{\te}_{{T'},p} \right) \left( \he_{t,p} - \overline{\he}_{{T'},p} \right)^\top\right].
	\end{align*}
	Recall that $\overline{\te}_{{T'},p} = \frac{1}{{T'}} \sum_{t=p+1}^T \te_{t,p}$ and $\overline{\he}_{{T'},p} = \frac{1}{{T'}} \sum_{t=p+1}^T \he_{t,p}$, therefore, by triangular inequality, it is sufficient to study the leading term $\frac{1}{{T'}} \sum_{t=p+1}^{T} \left[ (\he_{t,p} - \overline{\he}_{t,p}) - (\te_{t,p} - \overline{\te}_{t,p})\right] (\he_{t,p} - \te_{t,p})^\top$. For this, it is sufficient to consider the order of $\left\| \frac{1}{{T'}} \sum_{t=p+1}^{T} (\he_{t,p}  - \te_{t,p}) (\he_{t,p} - \te_{t,p})^\top \right\|_F$. We then have the bound
	\begin{align*}
		\frac{1}{{T'}} &\sum_{t=p+1}^{T} \left\| \he_{t,p}  - \te_{t,p} \right\|^2 
		\le 3 \sum_{l=1}^{p} \left\| \hA_{l,p} -\tA_{l,p} \right\|_F^2 \frac{1}{{T'}} \sum_{t=p+1}^{T} \left\| \hf_{t-l} \right\|^2 
		+ \frac{3}{{T'}} \sum_{t=p+1}^{T} \left\| \hf_{t} - \pf_{t} \right\|^2 \\
		&+ 3 \sum_{l=1}^{p} \left\| \tA_{l,p} \right\|_F^2 \frac{1}{{T'}} \sum_{t=p+1}^{T} \left\| \hf_{t-l} - \pf_{t-l} \right\|^2 
		= O_P\left(\left\| \hA_{p}- \tA_{p} \right\|_F^2\right) + O_P\left(p \left\| \hf_{t} - \pf_{t} \right\|^2\right),
	\end{align*}
	which converges to 0 in probability by the results of Lemmas \ref{1:l8} and \ref{1:l9}. Hence we can conclude that $\left\| \hS_{e,p}-\tS_{e,p}\right\|_F \to 0$ in probability.
	
	Lastly, $\left\| \hS_{e,p}-\pS_{e,p}\right\|_F = o_P\left(1\right)$ follows directly from $\left\| \hS_{e,p}-\tS_{e,p}\right\|_F = o_P\left(1\right)$, $\left\| \tS_{e,p}-\pS_{e,p}\right\|_F = o_P\left(1\right)$, and the triangular inequality.
\end{proof}

\section{Additional simulations on AR-sieve bootstrap}\label{1:sec:Appendix_C}

\subsection{AR-sieve bootstrap for mean statistics} \label{1:s2}

We examine the performance of the AR-sieve bootstrap for relatively weak factors. To achieve that, we evaluate the empirical coverage and average width of bootstrap confidence intervals for the mean statistics based on the same data-generating process as discussed in the main paper.

\begin{table}[!htb]
	\centering
	\caption{Empirical coverage, average width, and interval score of bootstrap intervals for $\theta_y$ with $\nu=0.6$.}
	\label{1:ta1e}
	\resizebox{\textwidth}{!}{%
		\begin{tabular}{|ccccccccccc|}
			\hline
			&  & \multicolumn{3}{c}{95\%} & \multicolumn{3}{c}{90\%} & \multicolumn{3}{c|}{80\%} \\ \hline
			T & N & \begin{tabular}[c]{@{}c@{}}Empirical\\  coverage\end{tabular} & \begin{tabular}[c]{@{}c@{}}Average\\ width\end{tabular} & \begin{tabular}[c]{@{}c@{}}Average\\  interval score\end{tabular} & \begin{tabular}[c]{@{}c@{}}Empirical\\  coverage\end{tabular} & \begin{tabular}[c]{@{}c@{}}Average\\ width\end{tabular} & \begin{tabular}[c]{@{}c@{}}Average\\  interval score\end{tabular} & \begin{tabular}[c]{@{}c@{}}Empirical\\  coverage\end{tabular} & \begin{tabular}[c]{@{}c@{}}Average\\ width\end{tabular} & \begin{tabular}[c]{@{}c@{}}Average\\  interval score\end{tabular} \\ \hline
			\multicolumn{11}{|c|}{Nonparametric bootstrap intervals using quantiles} \\ \hline
			\multirow{5}{*}{200}  & 50   & $0.957$ & $8.423$ & $10.729$ & $0.911$ & $7.080$ & $9.856$ & $0.819$ & $5.522$ & $8.810$ \\
			& 100  & $0.965$ & $8.551$ & $10.317$ & $0.913$ & $7.186$ & $9.506$ & $0.830$ & $5.601$ & $8.642$ \\
			& 200  & $0.965$ & $8.490$ & $10.791$ & $0.928$ & $7.136$ & $9.597$ & $0.839$ & $5.570$ & $8.476$ \\
			& 500  & $0.970$ & $8.742$ & $10.666$ & $0.927$ & $7.351$ & $9.590$ & $0.828$ & $5.732$ & $8.717$ \\
			& 1000 & $0.968$ & $9.090$ & $11.055$ & $0.946$ & $7.643$ & $9.641$ & $0.854$ & $5.954$ & $8.444$ \\
			&&&&&&&&&& \\
			\multirow{5}{*}{500}  & 50   & $0.939$ & $8.521$ & $12.639$ & $0.880$ & $7.164$ & $11.564$ & $0.774$ & $5.583$ & $10.313$ \\
			& 100  & $0.949$ & $8.288$ & $11.068$ & $0.893$ & $6.970$ & $10.284$ & $0.791$ & $5.438$ & $9.466$ \\
			& 200  & $0.947$ & $8.543$ & $12.417$ & $0.904$ & $7.183$ & $10.928$ & $0.818$ & $5.597$ & $9.619$ \\
			& 500  & $0.960$ & $8.525$ & $10.822$ & $0.929$ & $7.157$ & $9.732$ & $0.829$ & $5.591$ & $8.822$ \\
			& 1000 & $0.952$ & $8.343$ & $11.234$ & $0.916$ & $7.016$ & $10.067$ & $0.836$ & $5.472$ & $8.676$ \\
			&&&&&&&&&& \\
			\multirow{5}{*}{1000} & 50   & $0.931$ & $8.581$ & $13.487$ & $0.886$ & $7.213$ & $11.923$ & $0.774$ & $5.631$ & $10.433$ \\
			& 100  & $0.944$ & $8.441$ & $13.101$ & $0.889$ & $7.105$ & $11.734$ & $0.768$ & $5.538$ & $10.550$ \\
			& 200  & $0.937$ & $8.209$ & $12.268$ & $0.891$ & $6.905$ & $11.084$ & $0.792$ & $5.383$ & $9.744$ \\
			& 500  & $0.953$ & $8.547$ & $11.405$ & $0.900$ & $7.189$ & $10.701$ & $0.815$ & $5.603$ & $9.635$ \\
			& 1000 & $0.954$ & $8.584$ & $11.455$ & $0.891$ & $7.214$ & $10.683$ & $0.795$ & $5.630$ & $9.868$ \\ \hline
			\multicolumn{11}{|c|}{Parametric bootstrap intervals based on normality} \\ \hline
			\multirow{5}{*}{200}  & 50   & $0.961$ & $8.465$ & $10.704$ & $0.910$ & $7.104$ & $9.847$ & $0.822$ & $5.535$ & $8.816$ \\
			& 100  & $0.966$ & $8.590$ & $10.243$ & $0.921$ & $7.209$ & $9.485$ & $0.830$ & $5.617$ & $8.632$ \\
			& 200  & $0.968$ & $8.533$ & $10.748$ & $0.932$ & $7.162$ & $9.574$ & $0.839$ & $5.580$ & $8.453$ \\
			& 500  & $0.966$ & $8.783$ & $10.614$ & $0.927$ & $7.371$ & $9.592$ & $0.829$ & $5.743$ & $8.697$ \\
			& 1000 & $0.970$ & $9.127$ & $10.937$ & $0.948$ & $7.659$ & $9.675$ & $0.854$ & $5.968$ & $8.458$ \\
			&&&&&&&&&& \\
			\multirow{5}{*}{500}  & 50   & $0.940$ & $8.555$ & $12.396$ & $0.879$ & $7.180$ & $11.407$ & $0.774$ & $5.594$ & $10.283$ \\
			& 100  & $0.950$ & $8.330$ & $11.046$ & $0.896$ & $6.991$ & $10.306$ & $0.786$ & $5.447$ & $9.439$ \\
			& 200  & $0.947$ & $8.575$ & $12.331$ & $0.909$ & $7.196$ & $10.927$ & $0.822$ & $5.607$ & $9.607$ \\
			& 500  & $0.965$ & $8.560$ & $10.681$ & $0.931$ & $7.184$ & $9.697$ & $0.831$ & $5.597$ & $8.802$ \\
			& 1000 & $0.957$ & $8.389$ & $11.215$ & $0.917$ & $7.040$ & $10.082$ & $0.842$ & $5.485$ & $8.660$ \\
			&&&&&&&&&& \\
			\multirow{5}{*}{1000} & 50   & $0.932$ & $8.629$ & $13.388$ & $0.888$ & $7.242$ & $11.920$ & $0.774$ & $5.642$ & $10.429$ \\
			& 100  & $0.946$ & $8.481$ & $13.061$ & $0.891$ & $7.118$ & $11.670$ & $0.769$ & $5.546$ & $10.503$ \\
			& 200  & $0.942$ & $8.249$ & $12.290$ & $0.895$ & $6.922$ & $11.074$ & $0.791$ & $5.394$ & $9.684$ \\
			& 500  & $0.953$ & $8.585$ & $11.349$ & $0.899$ & $7.205$ & $10.641$ & $0.817$ & $5.614$ & $9.597$ \\
			& 1000 & $0.954$ & $8.626$ & $11.343$ & $0.894$ & $7.239$ & $10.612$ & $0.799$ & $5.640$ & $9.817$ \\ \hline
		\end{tabular}%
	}
\end{table}

As shown in Tables~\ref{1:ta1e} to~\ref{1:ta1i}, when $\nu$ is further reduced from $0.6$ to $0.2$ and the factors are weakened, the empirical coverage tends to increase with $N/T$, and the bootstrap intervals become wider and wider. This suggests that the AR-sieve bootstrap overestimates the standard error of the (standardized) mean statistic when $N$ increases. When the factors become weaker, the spikiness of the first two largest eigenvalues of accumulated symmetrized autocovariance matrices decreases. The number of factors can be overestimated, which brings the noise into the bootstrap samples. As a result, neither of the two types of bootstrap intervals performs well when factors are very weak (especially when $\nu = 0.2$) and $N/T$ is large. The bootstrap distribution of the (standardized) mean statistic suffers from comparably fatter tails. This phenomenon can be observed especially for large $T$ in Table~\ref{1:ta1i}, where both the average widths and the empirical coverages of bootstrap intervals are increasing with sample size $N$ while the average interval scores are decreasing.

\begin{table}[!htbp]
	\centering
	\caption{Empirical coverage, average width, and interval score of bootstrap intervals for $\theta_y$ with $\nu=0.4$.}
	\label{1:ta1g}
	\resizebox{\textwidth}{!}{%
		\begin{tabular}{|ccccccccccc|}
			\hline
			&  & \multicolumn{3}{c}{95\%} & \multicolumn{3}{c}{90\%} & \multicolumn{3}{c|}{80\%} \\ \hline
			T & N & \begin{tabular}[c]{@{}c@{}}Empirical\\  coverage\end{tabular} & \begin{tabular}[c]{@{}c@{}}Average\\ width\end{tabular} & \begin{tabular}[c]{@{}c@{}}Average\\  interval score\end{tabular} & \begin{tabular}[c]{@{}c@{}}Empirical\\  coverage\end{tabular} & \begin{tabular}[c]{@{}c@{}}Average\\ width\end{tabular} & \begin{tabular}[c]{@{}c@{}}Average\\  interval score\end{tabular} & \begin{tabular}[c]{@{}c@{}}Empirical\\  coverage\end{tabular} & \begin{tabular}[c]{@{}c@{}}Average\\ width\end{tabular} & \begin{tabular}[c]{@{}c@{}}Average\\  interval score\end{tabular} \\ \hline
			\multicolumn{11}{|c|}{Nonparametric bootstrap intervals using quantiles} \\ \hline
			\multirow{5}{*}{200}  & 50   & $0.969$ & $8.513$ & $9.931$  & $0.933$ & $7.154$ & $9.136$ & $0.845$ & $5.585$ & $8.188$ \\
			& 100  & $0.980$ & $8.821$ & $9.634$  & $0.944$ & $7.417$ & $8.730$ & $0.865$ & $5.782$ & $7.949$ \\
			& 200  & $0.982$ & $8.868$ & $10.416$ & $0.960$ & $7.451$ & $8.854$ & $0.887$ & $5.817$ & $7.638$ \\
			& 500  & $0.989$ & $9.648$ & $10.149$ & $0.973$ & $8.111$ & $8.870$ & $0.915$ & $6.323$ & $7.439$ \\
			& 1000 & $0.992$ & $10.190$ & $10.407$ & $0.980$ & $8.557$ & $9.079$ & $0.939$ & $6.672$ & $7.607$ \\
			&&&&&&&&&& \\
			\multirow{5}{*}{500}  & 50   & $0.943$ & $8.567$ & $12.859$ & $0.874$ & $7.197$ & $11.536$ & $0.765$ & $5.614$ & $10.393$ \\
			& 100  & $0.962$ & $8.367$ & $10.292$ & $0.903$ & $7.030$ & $9.770$ & $0.796$ & $5.487$ & $9.049$ \\
			& 200  & $0.957$ & $8.743$ & $11.581$ & $0.925$ & $7.352$ & $10.183$ & $0.846$ & $5.733$ & $9.016$ \\
			& 500  & $0.978$ & $8.974$ & $9.992$  & $0.945$ & $7.549$ & $8.930$ & $0.863$ & $5.889$ & $7.966$ \\
			& 1000 & $0.984$ & $8.998$ & $10.014$ & $0.959$ & $7.580$ & $8.876$ & $0.898$ & $5.916$ & $7.694$ \\
			&&&&&&&&&& \\
			\multirow{5}{*}{1000} & 50   & $0.934$ & $8.624$ & $13.608$ & $0.885$ & $7.250$ & $12.105$ & $0.771$ & $5.653$ & $10.618$ \\
			& 100  & $0.943$ & $8.486$ & $12.923$ & $0.891$ & $7.142$ & $11.624$ & $0.785$ & $5.570$ & $10.552$ \\
			& 200  & $0.941$ & $8.277$ & $11.882$ & $0.888$ & $6.959$ & $10.814$ & $0.805$ & $5.426$ & $9.487$ \\
			& 500  & $0.967$ & $8.709$ & $10.811$ & $0.917$ & $7.320$ & $9.967$ & $0.842$ & $5.711$ & $9.070$ \\
			& 1000 & $0.972$ & $8.939$ & $11.083$ & $0.919$ & $7.525$ & $9.874$ & $0.831$ & $5.875$ & $9.059$ \\ \hline
			\multicolumn{11}{|c|}{Parametric bootstrap intervals based on normality} \\ \hline
			\multirow{5}{*}{200}  & 50   & $0.971$ & $8.555$ & $9.934$  & $0.934$ & $7.180$ & $9.109$ & $0.843$ & $5.594$ & $8.200$ \\
			& 100  & $0.979$ & $8.868$ & $9.685$  & $0.947$ & $7.442$ & $8.695$ & $0.862$ & $5.798$ & $7.935$ \\
			& 200  & $0.985$ & $8.915$ & $10.326$ & $0.956$ & $7.481$ & $8.877$ & $0.889$ & $5.829$ & $7.644$ \\
			& 500  & $0.989$ & $9.685$ & $10.176$ & $0.975$ & $8.128$ & $8.891$ & $0.918$ & $6.333$ & $7.463$ \\
			& 1000 & $0.993$ & $10.228$ & $10.403$ & $0.982$ & $8.583$ & $9.054$ & $0.939$ & $6.688$ & $7.594$ \\
			&&&&&&&&&& \\
			\multirow{5}{*}{500}  & 50   & $0.945$ & $8.597$ & $12.755$ & $0.876$ & $7.215$ & $11.405$ & $0.763$ & $5.621$ & $10.363$ \\
			& 100  & $0.958$ & $8.403$ & $10.410$ & $0.908$ & $7.052$ & $9.762$ & $0.799$ & $5.494$ & $9.016$ \\
			& 200  & $0.960$ & $8.775$ & $11.546$ & $0.928$ & $7.364$ & $10.254$ & $0.846$ & $5.737$ & $9.006$ \\
			& 500  & $0.978$ & $9.016$ & $10.091$ & $0.947$ & $7.566$ & $8.916$ & $0.866$ & $5.895$ & $7.941$ \\
			& 1000 & $0.986$ & $9.054$ & $10.033$ & $0.959$ & $7.599$ & $8.919$ & $0.897$ & $5.920$ & $7.688$ \\
			&&&&&&&&&& \\
			\multirow{5}{*}{1000} & 50   & $0.932$ & $8.666$ & $13.475$ & $0.883$ & $7.273$ & $12.096$ & $0.775$ & $5.666$ & $10.598$ \\
			& 100  & $0.944$ & $8.531$ & $12.906$ & $0.894$ & $7.159$ & $11.577$ & $0.779$ & $5.578$ & $10.503$ \\
			& 200  & $0.945$ & $8.317$ & $11.877$ & $0.893$ & $6.979$ & $10.803$ & $0.802$ & $5.438$ & $9.443$ \\
			& 500  & $0.968$ & $8.749$ & $10.722$ & $0.923$ & $7.343$ & $9.938$ & $0.846$ & $5.721$ & $9.033$ \\
			& 1000 & $0.972$ & $8.994$ & $10.924$ & $0.926$ & $7.548$ & $9.814$ & $0.835$ & $5.881$ & $9.025$ \\ \hline
		\end{tabular}%
	}
\end{table}

\begin{table}[!thbp]
	\centering
	\caption{Empirical coverage, average width, and interval score of bootstrap intervals for $\theta_y$ with $\nu=0.2$.}
	\label{1:ta1i}
	\resizebox{\textwidth}{!}{%
		\begin{tabular}{|ccccccccccc|}
			\hline
			&  & \multicolumn{3}{c}{95\%} & \multicolumn{3}{c}{90\%} & \multicolumn{3}{c|}{80\%} \\ \hline
			T & N & \begin{tabular}[c]{@{}c@{}}Empirical\\  coverage\end{tabular} & \begin{tabular}[c]{@{}c@{}}Average\\ width\end{tabular} & \begin{tabular}[c]{@{}c@{}}Average\\  interval score\end{tabular} & \begin{tabular}[c]{@{}c@{}}Empirical\\  coverage\end{tabular} & \begin{tabular}[c]{@{}c@{}}Average\\ width\end{tabular} & \begin{tabular}[c]{@{}c@{}}Average\\  interval score\end{tabular} & \begin{tabular}[c]{@{}c@{}}Empirical\\  coverage\end{tabular} & \begin{tabular}[c]{@{}c@{}}Average\\ width\end{tabular} & \begin{tabular}[c]{@{}c@{}}Average\\  interval score\end{tabular} \\ \hline
			\multicolumn{11}{|c|}{Nonparametric bootstrap intervals using quantiles} \\ \hline
			\multirow{5}{*}{200}  & 50   & $0.980$ & $8.677$ & $9.368$  & $0.952$ & $7.291$ & $8.417$ & $0.877$ & $5.687$ & $7.451$ \\
			& 100  & $0.989$ & $9.119$ & $9.628$  & $0.971$ & $7.647$ & $8.297$ & $0.900$ & $5.967$ & $7.243$ \\
			& 200  & $0.994$ & $9.297$ & $9.702$  & $0.980$ & $7.819$ & $8.338$ & $0.944$ & $6.098$ & $6.886$ \\
			& 500  & $1.000$ & $10.850$ & $10.850$ & $0.998$ & $9.119$ & $9.131$ & $0.985$ & $7.120$ & $7.272$ \\
			& 1000 & $0.997$ & $12.374$ & $12.521$ & $0.994$ & $10.399$ & $10.670$ & $0.988$ & $8.101$ & $8.424$ \\
			&&&&&&&&&& \\
			\multirow{5}{*}{500}  & 50   & $0.940$ & $8.714$ & $12.959$ & $0.888$ & $7.330$ & $11.577$ & $0.786$ & $5.711$ & $10.325$ \\
			& 100  & $0.973$ & $8.591$ & $9.978$  & $0.930$ & $7.229$ & $9.165$ & $0.837$ & $5.632$ & $8.327$ \\
			& 200  & $0.981$ & $9.123$ & $10.594$ & $0.957$ & $7.673$ & $9.256$ & $0.897$ & $5.977$ & $7.953$ \\
			& 500  & $0.997$ & $9.799$ & $9.868$  & $0.984$ & $8.236$ & $8.625$ & $0.942$ & $6.433$ & $7.148$ \\
			& 1000 & $0.999$ & $10.222$ & $10.344$ & $0.998$ & $8.591$ & $8.743$ & $0.977$ & $6.700$ & $6.994$ \\
			&&&&&&&&&& \\
			\multirow{5}{*}{1000} & 50   & $0.938$ & $8.793$ & $13.505$ & $0.878$ & $7.395$ & $12.148$ & $0.775$ & $5.759$ & $10.722$ \\
			& 100  & $0.950$ & $8.668$ & $12.120$ & $0.887$ & $7.288$ & $11.303$ & $0.787$ & $5.691$ & $10.382$ \\
			& 200  & $0.961$ & $8.495$ & $11.213$ & $0.910$ & $7.133$ & $10.264$ & $0.826$ & $5.561$ & $9.021$ \\
			& 500  & $0.989$ & $9.152$ & $9.786$  & $0.962$ & $7.686$ & $8.759$ & $0.880$ & $5.986$ & $7.852$ \\
			& 1000 & $0.990$ & $9.789$ & $10.293$ & $0.972$ & $8.216$ & $8.920$ & $0.910$ & $6.416$ & $7.662$ \\ \hline
			\multicolumn{11}{|c|}{Parametric bootstrap intervals based on normality} \\ \hline
			\multirow{5}{*}{200}  & 50   & $0.983$ & $8.717$ & $9.392$  & $0.951$ & $7.316$ & $8.443$ & $0.880$ & $5.700$ & $7.458$ \\
			& 100  & $0.990$ & $9.150$ & $9.628$  & $0.968$ & $7.679$ & $8.317$ & $0.901$ & $5.983$ & $7.247$ \\
			& 200  & $0.993$ & $9.347$ & $9.719$  & $0.980$ & $7.844$ & $8.342$ & $0.945$ & $6.112$ & $6.911$ \\
			& 500  & $1.000$ & $10.907$ & $10.907$ & $0.998$ & $9.153$ & $9.157$ & $0.985$ & $7.131$ & $7.277$ \\
			& 1000 & $0.997$ & $12.421$ & $12.583$ & $0.994$ & $10.424$ & $10.694$ & $0.988$ & $8.122$ & $8.438$ \\
			&&&&&&&&&& \\
			\multirow{5}{*}{500}  & 50   & $0.946$ & $8.751$ & $12.804$ & $0.892$ & $7.344$ & $11.463$ & $0.786$ & $5.722$ & $10.259$ \\
			& 100  & $0.970$ & $8.635$ & $10.047$ & $0.934$ & $7.246$ & $9.145$ & $0.843$ & $5.646$ & $8.309$ \\
			& 200  & $0.982$ & $9.160$ & $10.503$ & $0.960$ & $7.687$ & $9.220$ & $0.893$ & $5.989$ & $7.935$ \\
			& 500  & $0.996$ & $9.850$ & $9.968$  & $0.986$ & $8.266$ & $8.619$ & $0.946$ & $6.440$ & $7.129$ \\
			& 1000 & $0.999$ & $10.263$ & $10.343$ & $0.998$ & $8.613$ & $8.772$ & $0.977$ & $6.710$ & $7.004$ \\
			&&&&&&&&&& \\
			\multirow{5}{*}{1000} & 50   & $0.937$ & $8.833$ & $13.426$ & $0.881$ & $7.413$ & $12.134$ & $0.776$ & $5.776$ & $10.703$ \\
			& 100  & $0.952$ & $8.713$ & $12.132$ & $0.891$ & $7.312$ & $11.276$ & $0.784$ & $5.697$ & $10.336$ \\
			& 200  & $0.959$ & $8.527$ & $11.119$ & $0.913$ & $7.156$ & $10.135$ & $0.824$ & $5.576$ & $8.983$ \\
			& 500  & $0.989$ & $9.179$ & $9.783$  & $0.965$ & $7.703$ & $8.700$ & $0.882$ & $6.002$ & $7.826$ \\
			& 1000 & $0.992$ & $9.832$ & $10.339$ & $0.979$ & $8.251$ & $8.918$ & $0.914$ & $6.429$ & $7.648$ \\ \hline
		\end{tabular}%
	}
\end{table}

\subsection{AR-sieve bootstrap for spiked eigenvalues of the symmetrized autocovariance matrix}\label{1:s3}

The study on spiked eigenvalues of high-dimensional covariance matrices has received massive attention in the past decades. For time-series data, researchers are particularly interested in the spiked eigenvalues of the symmetrized autocovariance matrix. However, the theoretical results of these spiked eigenvalues of the symmetrized autocovariance matrix for high-dimensional time series are much more involved and hard to apply for practical analysis. As an alternative, the AR-sieve bootstrap can be considered for real data applications when the theoretical results do not exist or are hard to implement. As discussed in Proposition~\ref{1:t3}, the bootstrap estimates $\delta_i^\ast(k)$ are generally consistent to $\delta_i(k)$. However, without a central limit theorem (CLT) on $\widehat{\delta}_i(k)$, the spiked eigenvalues of the symmetrized sample autocovariance matrix, it is generally hard to derive the validity of the AR-sieve bootstrapped estimate theoretically. We use simulations to study our AR-sieve bootstrap method's performance on estimating $\delta_i(k)$. To be more specific, the data we generated are based on the strongest factor model where $\nu=1$. We continue the study on the validity and consistency of our AR-sieve bootstrap method by assessing the empirical coverage of bootstrap intervals on the first two largest eigenvalues $\delta_1$ and $\delta_2$ of the symmetrized lag-$1$ autocovariance matrix. To make a comprehensive comparison based on average width and interval score of bootstrap intervals for various combination of $N$ and $T$, the bootstrap intervals are created based on standardized eigenvalues $\delta_1^0 = \frac{\sqrt{T}}{{N}^2} \delta_1$ and $\delta_2^0 = \frac{\sqrt{T}}{{N}^2} \delta_2$ rather than $\delta_1$ and $\delta_2$.  

\begin{table}[!htbp]
	\centering
	\caption{Empirical coverage, average width, and interval score of bootstrap intervals for $\delta_1^0$ with $\nu=1$.}
	\label{1:ta3a}
	\resizebox{\textwidth}{!}{%
		\begin{tabular}{|ccccccccccc|}
			\hline
			&  & \multicolumn{3}{c}{95\%} & \multicolumn{3}{c}{90\%} & \multicolumn{3}{c|}{80\%} \\ \hline
			T & N & \begin{tabular}[c]{@{}c@{}}Empirical\\  coverage\end{tabular} & \begin{tabular}[c]{@{}c@{}}Average\\ width\end{tabular} & \begin{tabular}[c]{@{}c@{}}Average\\  interval score\end{tabular} & \begin{tabular}[c]{@{}c@{}}Empirical\\  coverage\end{tabular} & \begin{tabular}[c]{@{}c@{}}Average\\ width\end{tabular} & \begin{tabular}[c]{@{}c@{}}Average\\  interval score\end{tabular} & \begin{tabular}[c]{@{}c@{}}Empirical\\  coverage\end{tabular} & \begin{tabular}[c]{@{}c@{}}Average\\ width\end{tabular} & \begin{tabular}[c]{@{}c@{}}Average\\  interval score\end{tabular} \\ \hline
			\multicolumn{11}{|c|}{Nonparametric bootstrap intervals using quantiles} \\ \hline
			\multirow{5}{*}{200}  & 50   & $0.846$ & $11.881$ & $27.227$ & $0.819$ & $9.775$ & $18.896$ & $0.771$ & $7.470$ & $13.697$ \\
			& 100  & $0.855$ & $11.999$ & $26.475$ & $0.835$ & $9.895$ & $18.426$ & $0.794$ & $7.587$ & $13.157$ \\
			& 200  & $0.854$ & $11.732$ & $26.724$ & $0.837$ & $9.676$ & $18.463$ & $0.798$ & $7.390$ & $13.200$ \\
			& 500  & $0.874$ & $11.805$ & $25.093$ & $0.846$ & $9.730$ & $17.536$ & $0.789$ & $7.444$ & $12.811$ \\
			& 1000 & $0.858$ & $12.077$ & $26.443$ & $0.841$ & $9.967$ & $18.380$ & $0.795$ & $7.623$ & $13.385$ \\
			&&&&&&&&&& \\
			\multirow{5}{*}{500}  & 50   & $0.887$ & $11.377$ & $22.661$ & $0.858$ & $9.481$ & $16.962$ & $0.777$ & $7.347$ & $13.539$ \\
			& 100  & $0.892$ & $11.326$ & $22.973$ & $0.873$ & $9.441$ & $16.895$ & $0.800$ & $7.317$ & $12.991$ \\
			& 200  & $0.891$ & $11.444$ & $23.008$ & $0.864$ & $9.541$ & $17.078$ & $0.797$ & $7.391$ & $13.353$ \\
			& 500  & $0.885$ & $11.426$ & $23.913$ & $0.858$ & $9.521$ & $17.425$ & $0.782$ & $7.366$ & $13.648$ \\
			& 1000 & $0.884$ & $11.357$ & $23.069$ & $0.866$ & $9.478$ & $17.031$ & $0.775$ & $7.339$ & $13.398$ \\
			&&&&&&&&&& \\
			\multirow{5}{*}{1000} & 50   & $0.943$ & $11.440$ & $17.729$ & $0.907$ & $9.582$ & $14.185$ & $0.810$ & $7.446$ & $12.196$ \\
			& 100  & $0.935$ & $11.322$ & $17.117$ & $0.901$ & $9.490$ & $14.079$ & $0.803$ & $7.372$ & $12.127$ \\
			& 200  & $0.934$ & $11.263$ & $18.888$ & $0.886$ & $9.422$ & $15.027$ & $0.809$ & $7.324$ & $12.756$ \\
			& 500  & $0.920$ & $11.281$ & $18.128$ & $0.891$ & $9.457$ & $15.059$ & $0.804$ & $7.347$ & $12.544$ \\
			& 1000 & $0.928$ & $11.221$ & $18.426$ & $0.888$ & $9.395$ & $14.828$ & $0.795$ & $7.299$ & $12.433$ \\ \hline
			\multicolumn{11}{|c|}{Parametric bootstrap intervals based on normality} \\ \hline
			\multirow{5}{*}{200}  & 50   & $0.901$ & $12.147$ & $19.992$ & $0.873$ & $10.194$ & $16.073$ & $0.796$ & $7.943$ & $13.071$ \\
			& 100  & $0.907$ & $12.304$ & $19.677$ & $0.878$ & $10.326$ & $15.899$ & $0.809$ & $8.045$ & $12.788$ \\
			& 200  & $0.904$ & $12.012$ & $19.824$ & $0.876$ & $10.081$ & $15.926$ & $0.820$ & $7.854$ & $12.677$ \\
			& 500  & $0.915$ & $12.088$ & $19.296$ & $0.896$ & $10.145$ & $15.414$ & $0.824$ & $7.904$ & $12.334$ \\
			& 1000 & $0.919$ & $12.365$ & $20.463$ & $0.890$ & $10.377$ & $16.048$ & $0.813$ & $8.085$ & $13.031$ \\
			&&&&&&&&&& \\
			\multirow{5}{*}{500}  & 50   & $0.928$ & $11.518$ & $18.289$ & $0.887$ & $9.666$ & $15.434$ & $0.800$ & $7.531$ & $13.180$ \\
			& 100  & $0.927$ & $11.463$ & $18.825$ & $0.890$ & $9.620$ & $15.334$ & $0.819$ & $7.495$ & $12.654$ \\
			& 200  & $0.927$ & $11.582$ & $18.808$ & $0.883$ & $9.720$ & $15.619$ & $0.814$ & $7.573$ & $13.063$ \\
			& 500  & $0.930$ & $11.553$ & $19.663$ & $0.881$ & $9.696$ & $15.970$ & $0.799$ & $7.554$ & $13.382$ \\
			& 1000 & $0.924$ & $11.501$ & $18.687$ & $0.877$ & $9.652$ & $15.431$ & $0.785$ & $7.520$ & $13.158$ \\
			&&&&&&&&&& \\
			\multirow{5}{*}{1000} & 50   & $0.953$ & $11.535$ & $15.854$ & $0.915$ & $9.681$ & $13.768$ & $0.826$ & $7.542$ & $12.143$ \\
			& 100  & $0.953$ & $11.426$ & $15.294$ & $0.915$ & $9.589$ & $13.457$ & $0.810$ & $7.471$ & $11.985$ \\
			& 200  & $0.942$ & $11.349$ & $16.814$ & $0.909$ & $9.524$ & $14.387$ & $0.809$ & $7.421$ & $12.580$ \\
			& 500  & $0.941$ & $11.380$ & $16.037$ & $0.901$ & $9.550$ & $14.190$ & $0.809$ & $7.441$ & $12.428$ \\
			& 1000 & $0.944$ & $11.310$ & $16.318$ & $0.906$ & $9.492$ & $13.903$ & $0.808$ & $7.395$ & $12.276$ \\ \hline
		\end{tabular}%
	}
\end{table}

First of all, we compute the empirical coverage, average width, and interval score for nonparametric bootstrap intervals using quantiles and parametric bootstrap intervals based on normality for $\delta_1^0$ and $\delta_2^0$. As shown in Tables~\ref{1:ta3a} to~\ref{1:ta4a}, neither of the two types of bootstrap intervals can provide the desired result as the empirical coverage probabilities are consistently lower than the nominal probabilities for each interval, especially when $T$ is small. While the ``blessing of dimensionality'' may improve the empirical coverage of both intervals on $\delta_1$ and $\delta_2$ for large $N$, the results are not as good for the (standardized) mean statistic. They consistently underestimated empirical coverage probabilities mainly due to the skewness of the sampling distribution of $\widehat{\delta}_i(k)$, especially for a relatively small $T$. In general, the parametric bootstrap interval based on normality, which is symmetric, and the nonparametric bootstrap interval using quantiles, which is reversely skewed, perform well when the sampling distributions are symmetric but do not perform well when the sample statistic follows a skewed distribution \citep[see,][for discussions]{hall_theoretical_1988}. To consider this skewness, an unreversed nonparametric bootstrap interval using quantiles, computed as 
\begin{align*}
	\left( \theta^\ast_{(\alpha/2)},\ \theta^\ast_{(1-\alpha/2)} \right),
\end{align*} 
%This is also referred to as \theta{BACK} by Peter Hall, except that this is an unequal-tailed interval
can also be computed and compared since the skewness of sample statistics is retained by the bootstrap estimates. It is noteworthy that, unlike the aforementioned two bootstrap intervals which follow the idea of using the bootstrap distribution of $(\theta^\ast - \widehat{\theta})$ to approximate the distribution of $(\widehat{\theta} - \theta)$, this unreversed nonparametric bootstrap interval is constructed based on the idea of using the bootstrap distribution of $\widehat{\theta}$ to create a confidence interval of $\theta$ directly. Therefore, this interval has its tails unreversed and hence is more appropriate for asymmetric distributions.  %Note that this interval is not the same as the percentile interval that is discussed by \citet{hall_theoretical_1988}

\begin{table}[]
	\centering
	\caption{Empirical coverage, average width, and interval score of bootstrap intervals for $\delta_2^0$ with $\nu=1$.}
	\label{1:ta4a}
	\resizebox{\textwidth}{!}{%
		\begin{tabular}{|ccccccccccc|}
			\hline
			&  & \multicolumn{3}{c}{95\%} & \multicolumn{3}{c}{90\%} & \multicolumn{3}{c|}{80\%} \\ \hline
			T & N & \begin{tabular}[c]{@{}c@{}}Empirical\\  coverage\end{tabular} & \begin{tabular}[c]{@{}c@{}}Average\\ width\end{tabular} & \begin{tabular}[c]{@{}c@{}}Average\\  interval score\end{tabular} & \begin{tabular}[c]{@{}c@{}}Empirical\\  coverage\end{tabular} & \begin{tabular}[c]{@{}c@{}}Average\\ width\end{tabular} & \begin{tabular}[c]{@{}c@{}}Average\\  interval score\end{tabular} & \begin{tabular}[c]{@{}c@{}}Empirical\\  coverage\end{tabular} & \begin{tabular}[c]{@{}c@{}}Average\\ width\end{tabular} & \begin{tabular}[c]{@{}c@{}}Average\\  interval score\end{tabular} \\ \hline
			\multicolumn{11}{|c|}{Nonparametric bootstrap intervals using quantiles} \\ \hline
			\multirow{5}{*}{200}  & 50   & $0.820$ & $2.264$ & $6.900$ & $0.753$ & $1.876$ & $5.059$ & $0.634$ & $1.442$ & $4.259$ \\
			& 100  & $0.795$ & $2.225$ & $8.224$ & $0.748$ & $1.838$ & $5.988$ & $0.649$ & $1.415$ & $4.609$ \\
			& 200  & $0.807$ & $2.176$ & $7.376$ & $0.764$ & $1.801$ & $5.334$ & $0.660$ & $1.387$ & $4.147$ \\
			& 500  & $0.809$ & $2.185$ & $7.212$ & $0.761$ & $1.810$ & $5.172$ & $0.646$ & $1.393$ & $4.127$ \\
			& 1000 & $0.816$ & $2.185$ & $7.343$ & $0.761$ & $1.809$ & $5.224$ & $0.655$ & $1.391$ & $4.185$ \\
			&&&&&&&&&& \\
			\multirow{5}{*}{500}  & 50   & $0.894$ & $2.614$ & $5.478$ & $0.846$ & $2.184$ & $4.263$ & $0.731$ & $1.691$ & $3.682$ \\
			& 100  & $0.897$ & $2.550$ & $5.205$ & $0.844$ & $2.130$ & $4.062$ & $0.746$ & $1.652$ & $3.420$ \\
			& 200  & $0.892$ & $2.576$ & $5.342$ & $0.853$ & $2.148$ & $4.154$ & $0.768$ & $1.667$ & $3.439$ \\
			& 500  & $0.898$ & $2.599$ & $5.202$ & $0.860$ & $2.167$ & $4.049$ & $0.764$ & $1.678$ & $3.402$ \\
			& 1000 & $0.894$ & $2.564$ & $5.193$ & $0.862$ & $2.139$ & $4.052$ & $0.753$ & $1.656$ & $3.360$ \\
			&&&&&&&&&& \\
			\multirow{5}{*}{1000} & 50   & $0.919$ & $2.720$ & $4.608$ & $0.879$ & $2.280$ & $3.899$ & $0.795$ & $1.772$ & $3.324$ \\
			& 100  & $0.926$ & $2.697$ & $4.502$ & $0.876$ & $2.259$ & $3.766$ & $0.768$ & $1.753$ & $3.292$ \\
			& 200  & $0.915$ & $2.672$ & $4.554$ & $0.869$ & $2.237$ & $3.775$ & $0.789$ & $1.739$ & $3.192$ \\
			& 500  & $0.928$ & $2.668$ & $4.549$ & $0.884$ & $2.237$ & $3.683$ & $0.794$ & $1.737$ & $3.156$ \\
			& 1000 & $0.919$ & $2.682$ & $4.672$ & $0.868$ & $2.248$ & $3.869$ & $0.762$ & $1.749$ & $3.362$ \\ \hline
			\multicolumn{11}{|c|}{Parametric bootstrap intervals based on normality} \\ \hline
			\multirow{5}{*}{200}  & 50   & $0.857$ & $2.314$ & $6.008$ & $0.781$ & $1.942$ & $5.015$ & $0.658$ & $1.513$ & $4.216$ \\
			& 100  & $0.835$ & $2.271$ & $7.695$ & $0.783$ & $1.906$ & $5.851$ & $0.660$ & $1.485$ & $4.534$ \\
			& 200  & $0.833$ & $2.224$ & $6.523$ & $0.789$ & $1.867$ & $5.131$ & $0.677$ & $1.454$ & $4.102$ \\
			& 500  & $0.841$ & $2.235$ & $6.175$ & $0.778$ & $1.875$ & $4.979$ & $0.663$ & $1.461$ & $4.092$ \\
			& 1000 & $0.837$ & $2.233$ & $6.293$ & $0.768$ & $1.874$ & $5.109$ & $0.680$ & $1.460$ & $4.125$ \\
			&&&&&&&&&& \\
			\multirow{5}{*}{500}  & 50   & $0.905$ & $2.644$ & $4.733$ & $0.845$ & $2.219$ & $4.177$ & $0.743$ & $1.729$ & $3.656$ \\
			& 100  & $0.917$ & $2.579$ & $4.528$ & $0.868$ & $2.164$ & $3.892$ & $0.761$ & $1.686$ & $3.382$ \\
			& 200  & $0.914$ & $2.606$ & $4.706$ & $0.868$ & $2.187$ & $4.027$ & $0.780$ & $1.704$ & $3.408$ \\
			& 500  & $0.920$ & $2.628$ & $4.551$ & $0.868$ & $2.206$ & $3.896$ & $0.769$ & $1.719$ & $3.390$ \\
			& 1000 & $0.920$ & $2.591$ & $4.597$ & $0.867$ & $2.175$ & $3.867$ & $0.765$ & $1.694$ & $3.342$ \\
			&&&&&&&&&& \\
			\multirow{5}{*}{1000} & 50   & $0.936$ & $2.742$ & $4.331$ & $0.891$ & $2.301$ & $3.802$ & $0.800$ & $1.793$ & $3.305$ \\
			& 100  & $0.938$ & $2.717$ & $4.205$ & $0.893$ & $2.280$ & $3.679$ & $0.771$ & $1.777$ & $3.266$ \\
			& 200  & $0.927$ & $2.693$ & $4.023$ & $0.887$ & $2.260$ & $3.624$ & $0.796$ & $1.761$ & $3.177$ \\
			& 500  & $0.941$ & $2.691$ & $4.160$ & $0.892$ & $2.259$ & $3.581$ & $0.801$ & $1.760$ & $3.135$ \\
			& 1000 & $0.936$ & $2.706$ & $4.329$ & $0.881$ & $2.271$ & $3.769$ & $0.763$ & $1.769$ & $3.363$ \\ \hline
		\end{tabular}%
	}
\end{table}

As shown in Tables~\ref{1:ta3c} and~\ref{1:ta4c}, unreversed nonparametric bootstrap intervals using quantiles outperform the other two competitors for $\delta_1$ with almost all combinations of $N$ and $T$ and for $\delta_2$ with small $T$. Meanwhile, the failure of nonparametric bootstrap intervals using quantiles and parametric bootstrap intervals based on normality verifies the skewness in the distribution of $\widehat{\delta}_i(k)$. Although some bias-corrected intervals may also be constructed, for example, by double bootstrap, to improve the empirical coverage probabilities further, those methods for reducing the error of bootstrap intervals generally have significant requirements on computations and are beyond the scope of this work.
			
\begin{table}[!htbp]
	\centering
	\caption{Empirical coverage, average width and interval score of unreversed nonparametric bootstrap intervals using quantiles for $\delta_1^0$ with $\nu=1$.}
	\label{1:ta3c}
	\resizebox{\textwidth}{!}{%
				\begin{tabular}{|ccccccccccc|}
			\hline
			\multicolumn{11}{|c|}{Unreversed nonparametric bootstrap intervals using quantiles} \\ \hline
			&  & \multicolumn{3}{c}{95\%} & \multicolumn{3}{c}{90\%} & \multicolumn{3}{c|}{80\%} \\ \hline
			T & N & \begin{tabular}[c]{@{}c@{}}Empirical\\  coverage\end{tabular} & \begin{tabular}[c]{@{}c@{}}Average\\ width\end{tabular} & \begin{tabular}[c]{@{}c@{}}Average\\  interval score\end{tabular} & \begin{tabular}[c]{@{}c@{}}Empirical\\  coverage\end{tabular} & \begin{tabular}[c]{@{}c@{}}Average\\ width\end{tabular} & \begin{tabular}[c]{@{}c@{}}Average\\  interval score\end{tabular} & \begin{tabular}[c]{@{}c@{}}Empirical\\  coverage\end{tabular} & \begin{tabular}[c]{@{}c@{}}Average\\ width\end{tabular} & \begin{tabular}[c]{@{}c@{}}Average\\  interval score\end{tabular} \\ \hline
			\multirow{5}{*}{200}  & 50   & $0.956$ & $11.881$ & $15.015$ & $0.913$ & $9.775$ & $13.388$ & $0.815$ & $7.470$ & $11.862$ \\
			& 100  & $0.959$ & $11.999$ & $14.554$ & $0.909$ & $9.895$ & $12.973$ & $0.813$ & $7.587$ & $11.561$ \\
			& 200  & $0.959$ & $11.732$ & $14.986$ & $0.909$ & $9.676$ & $13.065$ & $0.823$ & $7.390$ & $11.500$ \\
			& 500  & $0.960$ & $11.805$ & $14.519$ & $0.912$ & $9.730$ & $12.925$ & $0.844$ & $7.444$ & $11.144$ \\
			& 1000 & $0.954$ & $12.077$ & $15.885$ & $0.914$ & $9.967$ & $13.801$ & $0.820$ & $7.623$ & $11.833$ \\
			&&&&&&&&&& \\
			\multirow{5}{*}{500}  & 50   & $0.947$ & $11.377$ & $15.161$ & $0.901$ & $9.481$ & $13.763$ & $0.793$ & $7.347$ & $12.370$ \\
			& 100  & $0.951$ & $11.326$ & $14.969$ & $0.900$ & $9.441$ & $13.533$ & $0.818$ & $7.317$ & $11.918$ \\
			& 200  & $0.947$ & $11.444$ & $15.947$ & $0.904$ & $9.541$ & $14.044$ & $0.810$ & $7.391$ & $12.305$ \\
			& 500  & $0.941$ & $11.426$ & $16.354$ & $0.901$ & $9.521$ & $14.474$ & $0.792$ & $7.366$ & $12.665$ \\
			& 1000 & $0.946$ & $11.357$ & $15.300$ & $0.896$ & $9.478$ & $13.906$ & $0.775$ & $7.339$ & $12.620$ \\
			&&&&&&&&&& \\
			\multirow{5}{*}{1000} & 50   & $0.955$ & $11.440$ & $14.654$ & $0.910$ & $9.582$ & $13.399$ & $0.818$ & $7.446$ & $12.072$ \\
			& 100  & $0.958$ & $11.322$ & $14.361$ & $0.914$ & $9.490$ & $13.157$ & $0.810$ & $7.372$ & $11.781$ \\
			& 200  & $0.944$ & $11.263$ & $15.158$ & $0.906$ & $9.422$ & $13.922$ & $0.811$ & $7.324$ & $12.264$ \\
			& 500  & $0.951$ & $11.281$ & $14.943$ & $0.901$ & $9.457$ & $13.633$ & $0.811$ & $7.347$ & $12.262$ \\
			& 1000 & $0.957$ & $11.221$ & $14.608$ & $0.905$ & $9.395$ & $13.326$ & $0.817$ & $7.299$ & $12.017$ \\ \hline
		\end{tabular}%
	}
\end{table}

\begin{table}[!htbp]
	\centering
	\caption{Empirical coverage, average width, and interval score of unreversed nonparametric bootstrap intervals using quantiles for $\delta_2^0$ with $\nu=1$.}
	\label{1:ta4c}
	\resizebox{\textwidth}{!}{%
		\begin{tabular}{|ccccccccccc|}
			\hline
			\multicolumn{11}{|c|}{Unreversed nonparametric bootstrap intervals using quantiles}                                                                              \\ \hline
			&  & \multicolumn{3}{c}{95\%} & \multicolumn{3}{c}{90\%} & \multicolumn{3}{c|}{80\%} \\ \hline
			T & N & \begin{tabular}[c]{@{}c@{}}Empirical\\  coverage\end{tabular} & \begin{tabular}[c]{@{}c@{}}Average\\ width\end{tabular} & \begin{tabular}[c]{@{}c@{}}Average\\  interval score\end{tabular} & \begin{tabular}[c]{@{}c@{}}Empirical\\  coverage\end{tabular} & \begin{tabular}[c]{@{}c@{}}Average\\ width\end{tabular} & \begin{tabular}[c]{@{}c@{}}Average\\  interval score\end{tabular} & \begin{tabular}[c]{@{}c@{}}Empirical\\  coverage\end{tabular} & \begin{tabular}[c]{@{}c@{}}Average\\ width\end{tabular} & \begin{tabular}[c]{@{}c@{}}Average\\  interval score\end{tabular} \\ \hline
			\multirow{5}{*}{200}  & 50   & $0.861$ & $2.264$ & $5.674$ & $0.786$ & $1.876$ & $4.698$ & $0.675$ & $1.442$ & $3.843$ \\
			& 100  & $0.846$ & $2.225$ & $6.225$ & $0.769$ & $1.838$ & $5.036$ & $0.649$ & $1.415$ & $4.077$ \\
			& 200  & $0.848$ & $2.176$ & $6.303$ & $0.783$ & $1.801$ & $4.997$ & $0.667$ & $1.387$ & $3.959$ \\
			& 500  & $0.851$ & $2.185$ & $6.057$ & $0.776$ & $1.810$ & $4.935$ & $0.651$ & $1.393$ & $4.008$ \\
			& 1000 & $0.848$ & $2.185$ & $6.018$ & $0.780$ & $1.809$ & $4.899$ & $0.652$ & $1.391$ & $3.965$ \\
			&&&&&&&&&& \\
			\multirow{5}{*}{500}  & 50   & $0.908$ & $2.614$ & $4.757$ & $0.861$ & $2.184$ & $4.092$ & $0.757$ & $1.691$ & $3.467$ \\
			& 100  & $0.919$ & $2.550$ & $4.468$ & $0.854$ & $2.130$ & $3.880$ & $0.751$ & $1.652$ & $3.382$ \\
			& 200  & $0.917$ & $2.576$ & $4.606$ & $0.874$ & $2.148$ & $3.928$ & $0.778$ & $1.667$ & $3.300$ \\
			& 500  & $0.923$ & $2.599$ & $4.329$ & $0.870$ & $2.167$ & $3.799$ & $0.777$ & $1.678$ & $3.250$ \\
			& 1000 & $0.915$ & $2.564$ & $4.479$ & $0.867$ & $2.139$ & $3.891$ & $0.758$ & $1.656$ & $3.320$ \\
			&&&&&&&&&& \\
			\multirow{5}{*}{1000} & 50   & $0.938$ & $2.720$ & $4.155$ & $0.887$ & $2.280$ & $3.698$ & $0.798$ & $1.772$ & $3.267$ \\
			& 100  & $0.938$ & $2.697$ & $4.219$ & $0.879$ & $2.259$ & $3.706$ & $0.777$ & $1.753$ & $3.272$ \\
			& 200  & $0.929$ & $2.672$ & $4.115$ & $0.872$ & $2.237$ & $3.709$ & $0.780$ & $1.739$ & $3.240$ \\
			& 500  & $0.934$ & $2.668$ & $4.295$ & $0.886$ & $2.237$ & $3.715$ & $0.802$ & $1.737$ & $3.185$ \\
			& 1000 & $0.934$ & $2.682$ & $4.280$ & $0.873$ & $2.248$ & $3.787$ & $0.767$ & $1.749$ & $3.342$ \\ \hline
		\end{tabular}%
	}
\end{table}

\newpage
\color{darkblue}
\subsection{Comparison with moving block bootstrap}

To demonstrate the necessity of dimension reduction for high-dimensional inference and to evaluate the ``curse of dimensionality'' on ordinary resampling methods, we compare the proposed AR-sieve bootstrap with the standard moving block bootstrap (MBB). We consider block sizes of $l=5$ and $l=10$ to assess sensitivity to the block length.

We utilize the same strong factor DGP ($\nu=1$) as described in the simulation studies of the main text. We evaluate the performance of MBB on the statistics that have been studied in the main text, the (standardized) mean statistic ($\theta_y$) and the first two eigenvalues ($\delta_1, \delta_2$) of the symmetrized lag-1 autocovariance matrix.

As shown in Table~\ref{1:tambbmean}, the performance of the MBB for the mean statistic is close to that of the proposed AR-sieve bootstrap method in terms of empirical coverage. The coverage remains stable and close to the nominal level even as the dimension $N$ increases. This result is expected, as the cross-sectional averaging operation inherent in the mean statistic effectively mitigates the impact of high-dimensional noise (via the Law of Large Numbers), allowing standard methods to perform adequately for first-order moments.

A striking contrast is observed for the eigenvalues. Tables~\ref{1:tambb51}~to~\ref{1:tambb102} reveal that the performance of the MBB deteriorates rapidly as the dimension $N$ increases, regardless of the sample size $T$. As $N$ grows to 500 or 1000, the empirical coverage for the eigenvalues collapses. This failure occurs because the MBB resamples the full $N$-dimensional residual vector, thereby accumulating the noise from all $N$ idiosyncratic components. Unlike the proposed method, which filters out this noise via factor modeling, the MBB allows the accumulated noise to distort the covariance structure, leading to biased bootstrap estimates of the eigenvalues. The results are consistent across block sizes ($l=5$ and $l=10$), indicating that this is a fundamental limitation of ordinary resampling methods in high-dimensional settings.

This comparison highlights that while ordinary bootstrap methods like MBB may suffice for simple linear statistics (such as means), they suffer severely from the ``curse of dimensionality'' when inferring complex dependence structures (such as eigenvalues). Our proposed AR-sieve bootstrap, by explicitly reducing dimensionality, effectively filters out idiosyncratic noise and provides consistent estimation for high-dimensional second-order statistics, leading to robust empirical coverage in finite samples.

\begin{table}[!htbp]
	\centering
	\caption{Empirical coverage, average width, and interval score of bootstrap intervals constructed using moving block bootstrap (block size = 5) for $\theta_y$ with $\nu=1$.}
	\label{1:tambbmean}
	\resizebox{\textwidth}{!}{%
		\begin{tabular}{|ccccccccccc|}
			\hline
             &  & \multicolumn{3}{c}{95\%} & \multicolumn{3}{c}{90\%} & \multicolumn{3}{c|}{80\%} \\ \hline
            T & N & \begin{tabular}[c]{@{}c@{}}Empirical\\ coverage\end{tabular} & \begin{tabular}[c]{@{}c@{}}Average\\ width\end{tabular} & \begin{tabular}[c]{@{}c@{}}Average\\ interval score\end{tabular} & \begin{tabular}[c]{@{}c@{}}Empirical\\ coverage\end{tabular} & \begin{tabular}[c]{@{}c@{}}Average\\ width\end{tabular} & \begin{tabular}[c]{@{}c@{}}Average\\ interval score\end{tabular} & \begin{tabular}[c]{@{}c@{}}Empirical\\ coverage\end{tabular} & \begin{tabular}[c]{@{}c@{}}Average\\ width\end{tabular} & \begin{tabular}[c]{@{}c@{}}Average\\ interval score\end{tabular} \\ \hline
            \multicolumn{11}{|c|}{Nonparametric bootstrap intervals using quantiles} \\ \hline
            \multirow{5}{*}{200} & 50 & 0.908 & 8.685 & 15.350 & 0.861 & 7.317 & 13.200 & 0.760 & 5.726 & 11.258 \\
             & 100 & 0.922 & 8.981 & 15.134 & 0.869 & 7.564 & 13.405 & 0.780 & 5.901 & 11.600 \\
             & 200 & 0.926 & 9.477 & 15.023 & 0.882 & 7.971 & 13.535 & 0.782 & 6.227 & 11.864 \\
             & 500 & 0.942 & 9.253 & 15.009 & 0.889 & 7.788 & 13.115 & 0.798 & 6.069 & 11.369 \\
             & 1000 & 0.938 & 9.429 & 15.593 & 0.891 & 7.924 & 13.460 & 0.779 & 6.186 & 11.797 \\
             &  &  &  &  &  &  &  &  &  &  \\
            \multirow{5}{*}{500} & 50 & 0.921 & 8.777 & 15.049 & 0.860 & 7.398 & 13.251 & 0.767 & 5.781 & 11.465 \\
             & 100 & 0.910 & 9.201 & 16.247 & 0.846 & 7.750 & 14.583 & 0.749 & 6.045 & 12.554 \\
             & 200 & 0.931 & 9.462 & 14.431 & 0.876 & 7.952 & 13.007 & 0.770 & 6.206 & 11.427 \\
             & 500 & 0.943 & 9.468 & 13.434 & 0.896 & 7.963 & 12.342 & 0.807 & 6.211 & 10.946 \\
             & 1000 & 0.945 & 9.199 & 13.942 & 0.884 & 7.736 & 12.669 & 0.795 & 6.028 & 11.347 \\
             &  &  &  &  &  &  &  &  &  &  \\
            \multirow{5}{*}{1000} & 50 & 0.899 & 8.674 & 16.941 & 0.838 & 7.315 & 14.473 & 0.749 & 5.718 & 12.166 \\
             & 100 & 0.935 & 9.060 & 14.196 & 0.884 & 7.612 & 12.705 & 0.774 & 5.942 & 11.256 \\
             & 200 & 0.934 & 9.238 & 14.840 & 0.887 & 7.771 & 13.182 & 0.793 & 6.060 & 11.413 \\
             & 500 & 0.938 & 9.398 & 15.186 & 0.884 & 7.906 & 13.260 & 0.775 & 6.162 & 11.644 \\
             & 1000 & 0.939 & 9.466 & 15.228 & 0.896 & 7.963 & 13.306 & 0.802 & 6.214 & 11.359 \\ \hline
            \multicolumn{11}{|c|}{Parametric bootstrap intervals based on normality} \\ \hline
            \multirow{5}{*}{200} & 50 & 0.909 & 8.746 & 15.022 & 0.860 & 7.340 & 13.151 & 0.760 & 5.719 & 11.227 \\
             & 100 & 0.922 & 9.038 & 14.934 & 0.872 & 7.585 & 13.386 & 0.780 & 5.910 & 11.589 \\
             & 200 & 0.930 & 9.533 & 15.004 & 0.885 & 8.001 & 13.482 & 0.782 & 6.233 & 11.824 \\
             & 500 & 0.948 & 9.298 & 15.027 & 0.891 & 7.803 & 13.007 & 0.795 & 6.079 & 11.366 \\
             & 1000 & 0.938 & 9.472 & 15.425 & 0.892 & 7.949 & 13.419 & 0.777 & 6.193 & 11.768 \\
             &  &  &  &  &  &  &  &  &  &  \\
            \multirow{5}{*}{500} & 50 & 0.921 & 8.835 & 14.781 & 0.862 & 7.415 & 13.115 & 0.764 & 5.777 & 11.469 \\
             & 100 & 0.912 & 9.253 & 16.216 & 0.847 & 7.765 & 14.528 & 0.748 & 6.050 & 12.549 \\
             & 200 & 0.932 & 9.505 & 14.345 & 0.883 & 7.977 & 12.933 & 0.771 & 6.215 & 11.422 \\
             & 500 & 0.942 & 9.515 & 13.446 & 0.897 & 7.985 & 12.290 & 0.808 & 6.221 & 10.927 \\
             & 1000 & 0.948 & 9.241 & 13.947 & 0.882 & 7.755 & 12.667 & 0.795 & 6.042 & 11.315 \\
             &  &  &  &  &  &  &  &  &  &  \\
            \multirow{5}{*}{1000} & 50 & 0.901 & 8.741 & 16.752 & 0.839 & 7.336 & 14.372 & 0.750 & 5.715 & 12.121 \\
             & 100 & 0.936 & 9.094 & 14.134 & 0.882 & 7.632 & 12.717 & 0.774 & 5.947 & 11.222 \\
             & 200 & 0.937 & 9.281 & 14.908 & 0.888 & 7.789 & 13.175 & 0.792 & 6.069 & 11.399 \\
             & 500 & 0.944 & 9.439 & 15.058 & 0.889 & 7.922 & 13.232 & 0.771 & 6.172 & 11.584 \\
             & 1000 & 0.941 & 9.511 & 15.081 & 0.897 & 7.982 & 13.242 & 0.798 & 6.219 & 11.327 \\ \hline
		\end{tabular}%
	}
\end{table}

\begin{table}[!htbp]
	\centering
	\caption{Empirical coverage, average width, and interval score of bootstrap intervals constructed using moving block bootstrap (block size = 5) for $\delta_1^0$ with $\nu=1$.}
	\label{1:tambb51}
	\resizebox{\textwidth}{!}{%
		\begin{tabular}{|ccccccccccc|}
			\hline
             &  & \multicolumn{3}{c}{95\%} & \multicolumn{3}{c}{90\%} & \multicolumn{3}{c|}{80\%} \\ \hline
            T & N & \begin{tabular}[c]{@{}c@{}}Empirical\\ coverage\end{tabular} & \begin{tabular}[c]{@{}c@{}}Average\\ width\end{tabular} & \begin{tabular}[c]{@{}c@{}}Average\\ interval score\end{tabular} & \begin{tabular}[c]{@{}c@{}}Empirical\\ coverage\end{tabular} & \begin{tabular}[c]{@{}c@{}}Average\\ width\end{tabular} & \begin{tabular}[c]{@{}c@{}}Average\\ interval score\end{tabular} & \begin{tabular}[c]{@{}c@{}}Empirical\\ coverage\end{tabular} & \begin{tabular}[c]{@{}c@{}}Average\\ width\end{tabular} & \begin{tabular}[c]{@{}c@{}}Average\\ interval score\end{tabular} \\ \hline
            \multicolumn{11}{|c|}{Nonparametric bootstrap intervals using quantiles} \\ \hline
            \multirow{5}{*}{200} & 50 & 0.803 & 10.026 & 34.194 & 0.762 & 8.370 & 23.338 & 0.692 & 6.484 & 16.871 \\
             & 100 & 0.744 & 7.127 & 36.156 & 0.685 & 5.958 & 24.154 & 0.582 & 4.623 & 17.582 \\
             & 200 & 0.633 & 5.176 & 47.719 & 0.559 & 4.337 & 31.860 & 0.450 & 3.373 & 21.903 \\
             & 500 & 0.449 & 3.386 & 75.632 & 0.392 & 2.837 & 45.990 & 0.312 & 2.210 & 28.189 \\
             & 1000 & 0.323 & 2.369 & 104.476 & 0.267 & 1.990 & 58.739 & 0.207 & 1.551 & 33.488 \\
             &  &  &  &  &  &  &  &  &  &  \\
            \multirow{5}{*}{500} & 50 & 0.901 & 14.808 & 27.284 & 0.876 & 12.373 & 20.386 & 0.807 & 9.591 & 15.669 \\
             & 100 & 0.864 & 10.773 & 25.346 & 0.819 & 8.998 & 18.959 & 0.740 & 6.987 & 14.739 \\
             & 200 & 0.819 & 7.856 & 29.113 & 0.755 & 6.583 & 20.937 & 0.636 & 5.118 & 16.090 \\
             & 500 & 0.617 & 5.094 & 52.659 & 0.522 & 4.279 & 35.082 & 0.409 & 3.334 & 23.988 \\
             & 1000 & 0.472 & 3.572 & 75.805 & 0.412 & 3.000 & 45.575 & 0.325 & 2.339 & 28.008 \\
             &  &  &  &  &  &  &  &  &  &  \\
            \multirow{5}{*}{1000} & 50 & 0.953 & 20.588 & 25.175 & 0.942 & 17.185 & 20.665 & 0.902 & 13.297 & 16.276 \\
             & 100 & 0.939 & 14.986 & 20.957 & 0.906 & 12.528 & 17.026 & 0.858 & 9.717 & 13.700 \\
             & 200 & 0.908 & 11.009 & 21.171 & 0.869 & 9.222 & 16.756 & 0.773 & 7.168 & 13.635 \\
             & 500 & 0.783 & 7.075 & 33.420 & 0.709 & 5.942 & 24.229 & 0.593 & 4.629 & 18.209 \\
             & 1000 & 0.598 & 4.928 & 54.864 & 0.528 & 4.143 & 36.060 & 0.422 & 3.232 & 24.029 \\ \hline
            \multicolumn{11}{|c|}{Parametric bootstrap intervals based on normality} \\ \hline
            \multirow{5}{*}{200} & 50 & 0.827 & 10.232 & 29.479 & 0.781 & 8.587 & 22.042 & 0.706 & 6.690 & 16.654 \\
             & 100 & 0.764 & 7.215 & 32.501 & 0.690 & 6.055 & 23.390 & 0.588 & 4.718 & 17.396 \\
             & 200 & 0.643 & 5.218 & 48.022 & 0.559 & 4.379 & 32.135 & 0.451 & 3.412 & 21.958 \\
             & 500 & 0.445 & 3.400 & 77.490 & 0.396 & 2.854 & 46.504 & 0.304 & 2.223 & 28.260 \\
             & 1000 & 0.326 & 2.381 & 105.473 & 0.265 & 1.998 & 59.071 & 0.207 & 1.557 & 33.535 \\
             &  &  &  &  &  &  &  &  &  &  \\
            \multirow{5}{*}{500} & 50 & 0.936 & 15.135 & 23.033 & 0.897 & 12.702 & 18.888 & 0.824 & 9.896 & 15.308 \\
             & 100 & 0.896 & 10.906 & 21.579 & 0.848 & 9.152 & 17.354 & 0.763 & 7.131 & 14.367 \\
             & 200 & 0.829 & 7.924 & 26.968 & 0.761 & 6.650 & 20.457 & 0.639 & 5.181 & 15.958 \\
             & 500 & 0.609 & 5.127 & 53.068 & 0.519 & 4.303 & 35.405 & 0.413 & 3.353 & 24.038 \\
             & 1000 & 0.475 & 3.591 & 76.217 & 0.404 & 3.013 & 45.836 & 0.324 & 2.348 & 28.047 \\
             &  &  &  &  &  &  &  &  &  &  \\
            \multirow{5}{*}{1000} & 50 & 0.973 & 21.013 & 22.837 & 0.960 & 17.635 & 19.737 & 0.902 & 13.297 & 16.276 \\
             & 100 & 0.966 & 15.155 & 18.355 & 0.938 & 12.719 & 15.745 & 0.858 & 9.717 & 13.700 \\
             & 200 & 0.920 & 11.106 & 18.789 & 0.875 & 9.321 & 15.930 & 0.773 & 7.168 & 13.635 \\
             & 500 & 0.787 & 7.120 & 32.583 & 0.713 & 5.975 & 24.118 & 0.593 & 4.629 & 18.209 \\
             & 1000 & 0.597 & 4.962 & 54.831 & 0.526 & 4.164 & 36.024 & 0.422 & 3.232 & 24.029 \\ \hline
		\end{tabular}%
	}
\end{table}

\begin{table}[!htbp]
	\centering
	\caption{Empirical coverage, average width, and interval score of bootstrap intervals constructed using moving block bootstrap (block size = 5) for $\delta_2^0$ with $\nu=1$.}
	\label{1:tambb52}
	\resizebox{\textwidth}{!}{%
		\begin{tabular}{|ccccccccccc|}
			\hline
             &  & \multicolumn{3}{c}{95\%} & \multicolumn{3}{c}{90\%} & \multicolumn{3}{c|}{80\%} \\ \hline
            T & N & \begin{tabular}[c]{@{}c@{}}Empirical\\ coverage\end{tabular} & \begin{tabular}[c]{@{}c@{}}Average\\ width\end{tabular} & \begin{tabular}[c]{@{}c@{}}Average\\ interval score\end{tabular} & \begin{tabular}[c]{@{}c@{}}Empirical\\ coverage\end{tabular} & \begin{tabular}[c]{@{}c@{}}Average\\ width\end{tabular} & \begin{tabular}[c]{@{}c@{}}Average\\ interval score\end{tabular} & \begin{tabular}[c]{@{}c@{}}Empirical\\ coverage\end{tabular} & \begin{tabular}[c]{@{}c@{}}Average\\ width\end{tabular} & \begin{tabular}[c]{@{}c@{}}Average\\ interval score\end{tabular} \\ \hline
            \multicolumn{11}{|c|}{Nonparametric bootstrap intervals using quantiles} \\ \hline
            \multirow{5}{*}{200} & 50 & 0.736 & 1.923 & 9.489 & 0.689 & 1.619 & 6.383 & 0.692 & 6.484 & 16.871 \\
             & 100 & 0.666 & 1.474 & 11.781 & 0.598 & 1.236 & 7.478 & 0.582 & 4.623 & 17.582 \\
             & 200 & 0.579 & 1.095 & 14.221 & 0.507 & 0.920 & 8.767 & 0.450 & 3.373 & 21.903 \\
             & 500 & 0.406 & 0.729 & 20.221 & 0.356 & 0.613 & 11.741 & 0.312 & 2.210 & 28.189 \\
             & 1000 & 0.329 & 0.526 & 23.589 & 0.265 & 0.442 & 13.151 & 0.207 & 1.551 & 33.488 \\
             &  &  &  &  &  &  &  &  &  &  \\
            \multirow{5}{*}{500} & 50 & 0.905 & 3.198 & 5.564 & 0.861 & 2.689 & 4.443 & 0.785 & 2.102 & 3.677 \\
             & 100 & 0.842 & 2.448 & 6.982 & 0.793 & 2.056 & 5.102 & 0.718 & 1.603 & 3.916 \\
             & 200 & 0.774 & 1.821 & 8.181 & 0.714 & 1.528 & 5.750 & 0.599 & 1.190 & 4.297 \\
             & 500 & 0.573 & 1.187 & 14.577 & 0.493 & 0.996 & 9.370 & 0.399 & 0.776 & 6.127 \\
             & 1000 & 0.469 & 0.842 & 19.910 & 0.403 & 0.708 & 11.705 & 0.323 & 0.551 & 7.037 \\
             &  &  &  &  &  &  &  &  &  &  \\
            \multirow{5}{*}{1000} & 50 & 0.980 & 4.598 & 5.069 & 0.962 & 3.868 & 4.307 & 0.913 & 3.029 & 3.577 \\
             & 100 & 0.936 & 3.509 & 4.960 & 0.908 & 2.946 & 4.111 & 0.848 & 2.294 & 3.396 \\
             & 200 & 0.917 & 2.637 & 4.829 & 0.877 & 2.212 & 3.941 & 0.791 & 1.723 & 3.276 \\
             & 500 & 0.776 & 1.707 & 8.317 & 0.686 & 1.432 & 6.073 & 0.559 & 1.116 & 4.581 \\
             & 1000 & 0.589 & 1.204 & 14.106 & 0.523 & 1.010 & 9.067 & 0.420 & 0.788 & 5.987 \\ \hline
            \multicolumn{11}{|c|}{Parametric bootstrap intervals based on normality} \\ \hline
            \multirow{5}{*}{200} & 50 & 0.769 & 1.960 & 8.471 & 0.781 & 8.587 & 22.042 & 0.706 & 6.690 & 16.654 \\
             & 100 & 0.674 & 1.492 & 10.769 & 0.690 & 6.055 & 23.390 & 0.588 & 4.718 & 17.396 \\
             & 200 & 0.586 & 1.105 & 13.797 & 0.559 & 4.379 & 32.135 & 0.451 & 3.412 & 21.958 \\
             & 500 & 0.416 & 0.734 & 20.209 & 0.396 & 2.854 & 46.504 & 0.304 & 2.223 & 28.260 \\
             & 1000 & 0.322 & 0.529 & 23.589 & 0.265 & 1.998 & 59.071 & 0.207 & 1.557 & 33.535 \\
             &  &  &  &  &  &  &  &  &  &  \\
            \multirow{5}{*}{500} & 50 & 0.922 & 3.263 & 4.885 & 0.882 & 2.738 & 4.231 & 0.789 & 2.133 & 3.641 \\
             & 100 & 0.867 & 2.482 & 6.011 & 0.811 & 2.083 & 4.776 & 0.718 & 1.623 & 3.872 \\
             & 200 & 0.799 & 1.839 & 7.583 & 0.720 & 1.543 & 5.593 & 0.604 & 1.202 & 4.276 \\
             & 500 & 0.570 & 1.194 & 14.429 & 0.502 & 1.002 & 9.332 & 0.398 & 0.781 & 6.121 \\
             & 1000 & 0.473 & 0.847 & 19.846 & 0.403 & 0.710 & 11.697 & 0.323 & 0.554 & 7.028 \\
             &  &  &  &  &  &  &  &  &  &  \\
            \multirow{5}{*}{1000} & 50 & 0.986 & 4.697 & 4.969 & 0.971 & 3.942 & 4.327 & 0.918 & 3.071 & 3.614 \\
             & 100 & 0.960 & 3.557 & 4.448 & 0.925 & 2.985 & 3.905 & 0.854 & 2.326 & 3.370 \\
             & 200 & 0.925 & 2.661 & 4.572 & 0.892 & 2.233 & 3.882 & 0.786 & 1.740 & 3.272 \\
             & 500 & 0.777 & 1.716 & 8.176 & 0.693 & 1.440 & 6.044 & 0.560 & 1.122 & 4.567 \\
             & 1000 & 0.590 & 1.210 & 13.914 & 0.520 & 1.015 & 9.038 & 0.419 & 0.791 & 5.974 \\ \hline
		\end{tabular}%
	}
\end{table}

\begin{table}[!htbp]
	\centering
	\caption{Empirical coverage, average width, and interval score of bootstrap intervals constructed using moving block bootstrap (block size = 10) for $\delta_1^0$ with $\nu=1$.}
	\label{1:tambb101}
	\resizebox{\textwidth}{!}{%
		\begin{tabular}{|ccccccccccc|}
			\hline
             &  & \multicolumn{3}{c}{95\%} & \multicolumn{3}{c}{90\%} & \multicolumn{3}{c|}{80\%} \\ \hline
            T & N & \begin{tabular}[c]{@{}c@{}}Empirical\\ coverage\end{tabular} & \begin{tabular}[c]{@{}c@{}}Average\\ width\end{tabular} & \begin{tabular}[c]{@{}c@{}}Average\\ interval score\end{tabular} & \begin{tabular}[c]{@{}c@{}}Empirical\\ coverage\end{tabular} & \begin{tabular}[c]{@{}c@{}}Average\\ width\end{tabular} & \begin{tabular}[c]{@{}c@{}}Average\\ interval score\end{tabular} & \begin{tabular}[c]{@{}c@{}}Empirical\\ coverage\end{tabular} & \begin{tabular}[c]{@{}c@{}}Average\\ width\end{tabular} & \begin{tabular}[c]{@{}c@{}}Average\\ interval score\end{tabular} \\ \hline
            \multicolumn{11}{|c|}{Nonparametric bootstrap intervals using quantiles} \\ \hline
            \multirow{5}{*}{200} & 50 & 0.758 & 8.991 & 40.054 & 0.716 & 7.573 & 27.102 & 0.635 & 5.916 & 19.319 \\
             & 100 & 0.708 & 6.813 & 39.818 & 0.652 & 5.716 & 26.607 & 0.564 & 4.448 & 18.739 \\
             & 200 & 0.619 & 5.056 & 50.016 & 0.546 & 4.247 & 32.785 & 0.443 & 3.302 & 22.262 \\
             & 500 & 0.444 & 3.352 & 76.877 & 0.384 & 2.815 & 46.431 & 0.312 & 2.194 & 28.342 \\
             & 1000 & 0.323 & 2.357 & 104.899 & 0.265 & 1.979 & 58.936 & 0.207 & 1.541 & 33.563 \\
             &  &  &  &  &  &  &  &  &  &  \\
            \multirow{5}{*}{500} & 50 & 0.871 & 13.313 & 29.885 & 0.829 & 11.211 & 22.310 & 0.764 & 8.912 & 16.852 \\
             & 100 & 0.845 & 10.289 & 27.427 & 0.801 & 8.630 & 20.250 & 0.715 & 6.827 & 15.519 \\
             & 200 & 0.809 & 7.656 & 30.666 & 0.740 & 6.417 & 21.963 & 0.636 & 5.049 & 16.497 \\
             & 500 & 0.615 & 5.061 & 53.256 & 0.517 & 4.252 & 35.505 & 0.409 & 3.332 & 24.159 \\
             & 1000 & 0.472 & 3.555 & 76.310 & 0.407 & 2.989 & 45.858 & 0.326 & 2.338 & 28.121 \\
             &  &  &  &  &  &  &  &  &  &  \\
            \multirow{5}{*}{1000} & 50 & 0.937 & 18.445 & 25.485 & 0.914 & 15.550 & 20.922 & 0.839 & 12.140 & 17.244 \\
             & 100 & 0.920 & 14.309 & 21.712 & 0.896 & 12.003 & 17.545 & 0.828 & 9.335 & 14.119 \\
             & 200 & 0.888 & 10.734 & 22.324 & 0.849 & 9.004 & 17.501 & 0.758 & 7.000 & 14.042 \\
             & 500 & 0.780 & 7.028 & 33.991 & 0.703 & 5.897 & 24.582 & 0.587 & 4.594 & 18.427 \\
             & 1000 & 0.604 & 4.908 & 55.123 & 0.526 & 4.130 & 36.089 & 0.420 & 3.220 & 24.034 \\ \hline
            \multicolumn{11}{|c|}{Parametric bootstrap intervals based on normality} \\ \hline
            \multirow{5}{*}{200} & 50 & 0.777 & 9.201 & 36.973 & 0.724 & 7.722 & 26.432 & 0.648 & 6.016 & 19.176 \\
             & 100 & 0.730 & 6.906 & 37.009 & 0.662 & 5.796 & 25.932 & 0.563 & 4.516 & 18.658 \\
             & 200 & 0.636 & 5.105 & 50.233 & 0.549 & 4.284 & 33.119 & 0.443 & 3.338 & 22.324 \\
             & 500 & 0.442 & 3.377 & 78.460 & 0.385 & 2.834 & 46.917 & 0.310 & 2.208 & 28.419 \\
             & 1000 & 0.320 & 2.369 & 105.794 & 0.264 & 1.988 & 59.194 & 0.207 & 1.549 & 33.581 \\
             &  &  &  &  &  &  &  &  &  &  \\
            \multirow{5}{*}{500} & 50 & 0.897 & 13.630 & 26.512 & 0.850 & 11.439 & 21.042 & 0.785 & 8.756 & 14.971 \\
             & 100 & 0.879 & 10.440 & 24.187 & 0.820 & 8.762 & 19.219 & 0.744 & 6.720 & 14.314 \\
             & 200 & 0.818 & 7.721 & 29.026 & 0.745 & 6.480 & 21.454 & 0.608 & 4.991 & 16.827 \\
             & 500 & 0.611 & 5.097 & 53.892 & 0.517 & 4.277 & 35.753 & 0.409 & 3.314 & 24.946 \\
             & 1000 & 0.470 & 3.576 & 76.612 & 0.402 & 3.001 & 46.021 & 0.320 & 2.329 & 28.394 \\
             &  &  &  &  &  &  &  &  &  &  \\
            \multirow{5}{*}{1000} & 50 & 0.956 & 18.913 & 23.411 & 0.926 & 15.873 & 20.344 & 0.852 & 12.367 & 17.040 \\
             & 100 & 0.953 & 14.504 & 18.737 & 0.914 & 12.172 & 16.279 & 0.841 & 9.484 & 13.784 \\
             & 200 & 0.910 & 10.827 & 19.839 & 0.864 & 9.086 & 16.632 & 0.764 & 7.079 & 13.859 \\
             & 500 & 0.782 & 7.069 & 33.052 & 0.703 & 5.932 & 24.396 & 0.589 & 4.622 & 18.375 \\
             & 1000 & 0.606 & 4.943 & 54.850 & 0.524 & 4.149 & 36.031 & 0.421 & 3.232 & 23.984 \\ \hline
		\end{tabular}%
	}
\end{table}

\begin{table}[!htbp]
	\centering
	\caption{Empirical coverage, average width, and interval score of bootstrap intervals constructed using moving block bootstrap (block size = 10) for $\delta_2^0$ with $\nu=1$.}
	\label{1:tambb102}
	\resizebox{\textwidth}{!}{%
		\begin{tabular}{|ccccccccccc|}
			\hline
             &  & \multicolumn{3}{c}{95\%} & \multicolumn{3}{c}{90\%} & \multicolumn{3}{c|}{80\%} \\ \hline
            T & N & \begin{tabular}[c]{@{}c@{}}Empirical\\ coverage\end{tabular} & \begin{tabular}[c]{@{}c@{}}Average\\ width\end{tabular} & \begin{tabular}[c]{@{}c@{}}Average\\ interval score\end{tabular} & \begin{tabular}[c]{@{}c@{}}Empirical\\ coverage\end{tabular} & \begin{tabular}[c]{@{}c@{}}Average\\ width\end{tabular} & \begin{tabular}[c]{@{}c@{}}Average\\ interval score\end{tabular} & \begin{tabular}[c]{@{}c@{}}Empirical\\ coverage\end{tabular} & \begin{tabular}[c]{@{}c@{}}Average\\ width\end{tabular} & \begin{tabular}[c]{@{}c@{}}Average\\ interval score\end{tabular} \\ \hline
            \multicolumn{11}{|c|}{Nonparametric bootstrap intervals using quantiles} \\ \hline
            \multirow{5}{*}{200} & 50 & 0.692 & 1.762 & 10.860 & 0.640 & 1.498 & 7.200 & 0.555 & 1.182 & 5.020 \\
             & 100 & 0.636 & 1.403 & 12.591 & 0.562 & 1.182 & 7.957 & 0.466 & 0.924 & 5.289 \\
             & 200 & 0.571 & 1.077 & 14.764 & 0.506 & 0.906 & 8.984 & 0.423 & 0.707 & 5.689 \\
             & 500 & 0.401 & 0.722 & 20.423 & 0.350 & 0.607 & 11.825 & 0.278 & 0.473 & 6.980 \\
             & 1000 & 0.322 & 0.523 & 23.682 & 0.262 & 0.440 & 13.201 & 0.204 & 0.343 & 7.469 \\
             &  &  &  &  &  &  &  &  &  &  \\
            \multirow{5}{*}{500} & 50 & 0.869 & 2.938 & 6.534 & 0.817 & 2.500 & 5.146 & 0.724 & 1.972 & 4.162 \\
             & 100 & 0.824 & 2.338 & 7.654 & 0.772 & 1.968 & 5.542 & 0.676 & 1.539 & 4.210 \\
             & 200 & 0.780 & 1.792 & 8.423 & 0.694 & 1.507 & 5.989 & 0.595 & 1.175 & 4.443 \\
             & 500 & 0.565 & 1.178 & 14.793 & 0.489 & 0.989 & 9.437 & 0.394 & 0.770 & 6.166 \\
             & 1000 & 0.471 & 0.839 & 20.046 & 0.400 & 0.705 & 11.748 & 0.327 & 0.550 & 7.048 \\
             &  &  &  &  &  &  &  &  &  &  \\
            \multirow{5}{*}{1000} & 50 & 0.963 & 4.222 & 5.207 & 0.921 & 3.586 & 4.592 & 0.847 & 2.828 & 4.053 \\
             & 100 & 0.926 & 3.345 & 5.202 & 0.881 & 2.819 & 4.383 & 0.818 & 2.204 & 3.676 \\
             & 200 & 0.911 & 2.589 & 5.030 & 0.860 & 2.178 & 4.144 & 0.777 & 1.698 & 3.426 \\
             & 500 & 0.770 & 1.691 & 8.512 & 0.686 & 1.421 & 6.178 & 0.548 & 1.106 & 4.636 \\
             & 1000 & 0.594 & 1.198 & 14.180 & 0.521 & 1.007 & 9.098 & 0.416 & 0.785 & 5.999 \\ \hline
            \multicolumn{11}{|c|}{Parametric bootstrap intervals based on normality} \\ \hline
            \multirow{5}{*}{200} & 50 & 0.723 & 1.806 & 10.097 & 0.653 & 1.516 & 7.051 & 0.552 & 1.181 & 5.008 \\
             & 100 & 0.653 & 1.423 & 11.733 & 0.572 & 1.194 & 7.702 & 0.460 & 0.930 & 5.249 \\
             & 200 & 0.579 & 1.088 & 14.307 & 0.509 & 0.913 & 8.888 & 0.421 & 0.711 & 5.675 \\
             & 500 & 0.406 & 0.727 & 20.385 & 0.352 & 0.610 & 11.830 & 0.284 & 0.475 & 6.978 \\
             & 1000 & 0.320 & 0.526 & 23.684 & 0.265 & 0.441 & 13.207 & 0.203 & 0.344 & 7.470 \\
             &  &  &  &  &  &  &  &  &  &  \\
            \multirow{5}{*}{500} & 50 & 0.882 & 3.014 & 6.068 & 0.828 & 2.529 & 5.055 & 0.722 & 1.970 & 4.170 \\
             & 100 & 0.846 & 2.370 & 6.801 & 0.778 & 1.989 & 5.295 & 0.683 & 1.550 & 4.180 \\
             & 200 & 0.784 & 1.811 & 8.011 & 0.712 & 1.520 & 5.867 & 0.595 & 1.184 & 4.407 \\
             & 500 & 0.567 & 1.184 & 14.634 & 0.495 & 0.993 & 9.409 & 0.398 & 0.774 & 6.155 \\
             & 1000 & 0.471 & 0.843 & 19.932 & 0.402 & 0.708 & 11.735 & 0.326 & 0.551 & 7.042 \\
             &  &  &  &  &  &  &  &  &  &  \\
            \multirow{5}{*}{1000} & 50 & 0.965 & 4.326 & 5.278 & 0.925 & 3.630 & 4.734 & 0.855 & 2.828 & 4.102 \\
             & 100 & 0.941 & 3.395 & 4.779 & 0.895 & 2.849 & 4.228 & 0.827 & 2.220 & 3.641 \\
             & 200 & 0.917 & 2.617 & 4.881 & 0.871 & 2.196 & 4.113 & 0.784 & 1.711 & 3.430 \\
             & 500 & 0.772 & 1.702 & 8.370 & 0.681 & 1.428 & 6.140 & 0.556 & 1.113 & 4.622 \\
             & 1000 & 0.594 & 1.205 & 13.998 & 0.521 & 1.012 & 9.073 & 0.413 & 0.788 & 5.989 \\ \hline
		\end{tabular}%
	}
\end{table}

\color{black}
\newpage
\section{Applications on sparsely observed functional time series}\label{appendix0}

The second contribution of this work is that we compare the proposed novel AR-sieve bootstrap for high-dimensional time series with the AR-sieve bootstrap method for functional time series \citep{paparoditis_sieve_2018} in terms of their applications on sparse and unsmoothed functional observations. We suggest that the sparse and unsmoothed observations need to be treated as high-dimensional time series, and the AR-sieve bootstrap proposed in this work needs to be applied. 
%In conclusion, when the observations are not dense enough, the novel AR-sieve bootstrap method for high-dimensional time series is generally a better approach. 
In the literature of functional time series studies, a very fundamental assumption is that the actual observations come from a smoothed functional curve, and statistical inferences for functional data usually require the observations to be dense. In a classic functional set-up, dense and discrete points are observed on a sample of $T$ curves. Denoted by $N_t$ the number of observations for the curve $t$, the discussions on the density of observations in functional data literature are generally through assumptions made on $N_t$. Typically, when $N_t$ is much larger than the sample size $T$, the data can be considered dense functional data where each curve can be well smoothed before analysis.  %Also, \citet{wang_asymptotics_2017} studies the order of $N_t$ and suggests considering functional data as dense if the root-$n$ rate of an estimator can be obtained. 
However, in the case where $N_t$ is small compared with the sample size $T$ for all $t$, the discrete observations should be considered sparse along the population functional curve. %Some literature, such as \citet{yao_functional_2005} and \citet{hall_properties_2006} suggest using the information from all observations to achieve improved inferences. % see \citet{yao2005functional} and \citet{hall2006properties}. 
The fundamental problem of sparse functional data is that the local patterns of the population functional curve are generally not captured by those sparse observations. %, and the measurement errors can also mislead the smoothing results.

To illustrate the potential problems of pre-smoothing sparse observations for functional time series analysis, we consider a toy example. % which incorporates measurement errors. 
For a square-integrable functional process $\{\mathcal{X}(u),\ u \in \mathcal{I}\}$, let $y_{i,t}$ be the $i$\textsuperscript{th} observation of $\{\mathcal{X}_t(\cdot)\}$, observed at random time $t$ with the measurement errors defined as $\epsilon_{i,t}$ for $t=1,2,...,T$ and $i=1,2,...,N$. Consider now a model of functional observations
%for the $t$\textsuperscript{th} random curve,
\begin{align}\label{1:e0}
	y_{i,t} = \mathcal{X}_t(u_{i}) + \epsilon_{i,t}, \ \ \ u_{i} \in \mathcal{I},
\end{align}
where $\epsilon_{i,t}$ is independent and identically distributed (i.i.d.) with $\mathbb{E}(\epsilon_{i,t}) = 0$, $\mathbb{V}(\epsilon_{i,t}) = \sigma^2$ and $\mathcal{I}$ is a functional support. In this model, the observations of $\{\mathcal{X}_t(\cdot)\}$ are assumed to be equally spaced, and the number of measurements $N$ assesses the density and design of the actual observations. %A linear smoother estimates the function value $\widehat{Y}_{ti}$ on the $t$\textsuperscript{th} curve by a combination of observations $Y_{ti}$ and the weight $w_{i,u_{ti}}$ of the $i$\textsuperscript{th} observation on the $j$\textsuperscript{th} observation. 
In functional data analysis, $\mathcal{X}_t(u_i)$ can be estimated or recovered by some smoothing methods, such as a linear smoother, as follows.
\begin{align*}
\widehat{\mathcal{X}_t}(u_{i}) = \sum_{j=1}^{N} w_{i}(u_{j}) y_{t,i},
\end{align*}
where $w_{i}(u_{j})$ is the weight of the $j$\textsuperscript{th} point on the $i$\textsuperscript{th} point with $\sum_{j=1}^{N} w_{i}(u_{j}) = 1$ for $t=1,2,...,T$ and $i=1,2,\dots,N$. The accuracy of the smoothing curve is highly related to the density of observations and measurement errors. If observations along the curve are equally spaced, the change of density can affect the quality of smoothness and its recovery power to the population curve. For a relatively sparse curve, smoothing can fail to work under certain situations; for example, when there are local patterns that observations are too sparse to capture. To visually depict this phenomenon, we provide a toy example by simulation in the following part. We consider a contaminated functional time series model generated from three Fourier bases with different frequencies reflecting local patterns. The details of the simulation setting can be found in Section~\ref{app.f.1}. The curves in Figure~\ref{1:f1} are plotted based on $401$ grid points defined on functional support $[0,1]$, whereas the actual number of observations $N$ along each curve is chosen as $51$, $21$, and $5$ to address different observation densities.
\begin{figure}[!htbp]
\centering
\includegraphics[width=0.9\textwidth]{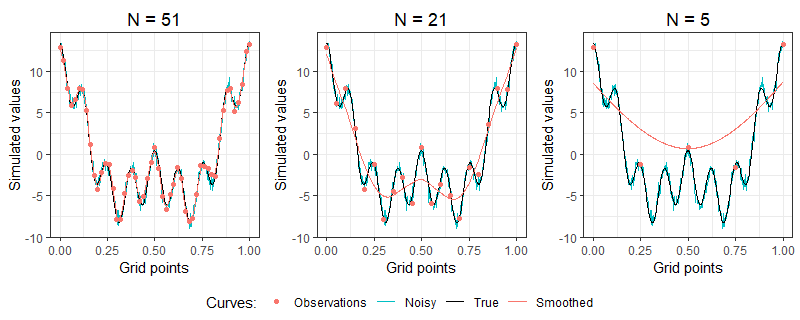}
\caption{Example of smoothing error of sparse functional time series observations.}\label{1:f1}
\end{figure}
As shown in Figure~\ref{1:f1}, when the observations (red points) become sparse (but still equally spaced), the (red) smoothing curve can lead to an obviously misleading result with local patterns not accurately captured by the smoothing curve. The errors associated with pre-smoothing on those sparse observations are generally large. In this situation, the assumption of dense functional data suffers from insufficient observations along each curve. As a result, we cannot adopt the pre-smoothing results based on a functional set-up but instead treat the data as a multivariate time series with growing dimensions. In other words, when $N$ grows with sample size $T$ but at a relatively slower rate, the real data may adapt to a high-dimensional set-up rather than a functional set-up, which makes statistical inferences and applications rather different. This phenomenon is associated with an area where functional data analysis and high-dimensional data analysis may overlap, yet follow different assumptions and produce quite different asymptotic results. 

In contrast to functional data analysis, where the increase of observations along a curve can practically improve pre-smoothing and recover the functional curve, the growth of dimensions is associated with the increase of complexity for high-dimensional data analysis. This key difference makes it vital to choose between functional time series and high-dimensional time series methods. In the following part, we consider the situation where $N$ is growing but not fast enough. The curve smoothed from the sparse observations is inaccurate, especially for local patterns of a functional curve. We apply the proposed AR-sieve bootstrap method for studying the inferences of this type of high-dimensional time series.

\subsection{Smoothing on sparse discrete functional time series}\label{app.f.1}

To study the impact of smoothing on the sparse functional time series observations, we can compare bootstrap samples' empirical distributions under various densities of observations. To start, we first assume the data originated from functional curves, which are temporally dependent. Recall the model~\eqref{1:e0} that
\begin{align*}
y_{t,i} = \mathcal{X}_t(u_{i}) + \epsilon_{t,i}, \ \ \ u_{i} \in \mathcal{I},
\end{align*}
where $\epsilon_{t,i}$ is i.i.d.\ with $\mathbb{E}(\epsilon_{t,i}) = 0$ and $\mathbb{V}(\epsilon_{t,i}) = \sigma^2$, for $t=1,2,...,T$ and $i=1,2,...,N$. In this model, the number of measurements $N$ reflects the density of the actual observations. To study the impact of density, we assume the observations are equally spaced and generated from a three-factor model 
\begin{align*}
\py_{t} = \pQ \pf_{t} + \pu_{t},
\end{align*}
where $u_{t,i}$, the element in $\{\pu_t\}$, is an independent random noise $\mathcal{N}\ (0,1)$, $\pQ$ is a $N \times 3$ matrix with each column a Fourier basis and $\cos(2\pi i/N)$, $\cos(4\pi i/N)$, $0.5\cos(16\pi i/N)$ as the $i$\textsuperscript{th} element, respectively. The factors $\{\pf_t\}$ follow a VAR(1) model with a coefficient matrix
\begin{align*}
\left[\begin{matrix}
0.5&0.1&0.1\\
0.1&0.5&0.1\\
0.1&0.1&0.5\\ 
\end{matrix} \right]
\end{align*}  
and errors independently simulated from $\mathcal{N}\ (0,1)$. The Fourier basis is selected to produce a smooth population curve, with the third basis reflecting local patterns. Hence, we can generate discrete observations from a functional curve with local patterns. In Section~\ref{1:intro}, we have presented graphs of $\{\py_t\}$ at a particular time $t$ with three different densities of observations to illustrate the potential issue of smoothing. This section takes it one step further and considers a wider choice of densities so that the actual dimensions of observations along each curve are $N=101,51,21,17,11$ and $5$.

For the same choice of time $t$ as in Section~\ref{1:intro}, we have generated $6$ plots under various densities in Figure~\ref{1:f2} to compare the smoothing results with the true population curve and the noisy curve with small measurement errors. The smoothing results are obtained using $B$-splines with the number of basis functions set to $N$, the actual number of observations in each case, and the roughness penalties selected based on generalized cross-validation (GCV). As depicted in Figure \ref{1:f2}, when the actual number of observations $N$ is relatively small, for example, $N<21$, some local patterns of the population curve are generally not captured. In addition, the smoothing curve sometimes also averaged out the actual observations to achieve relatively flat results, for example, when $N=21,17$ and $5$ as in Figure \ref{1:f2}. As a result, the observations after smoothing are generally less spread than the original observations, which produces very different bootstrap samples and inferences' results. To see that, we generate $B=499$ AR-sieve bootstrap samples and computed two summary statistics to compare the bootstrap distribution based on original observations with smoothed observations. We use AR-sieve bootstrap to obtain estimates of a so-called (standardized) mean statistic, computed as $\overline{\overline{y^\ast}}=\frac{\sqrt{T}}{\sqrt{N}} \pone^\top \hQ \overline{\bof}$ according to Theorem \ref{1:t1}, and $\delta_1^\ast$, the estimate of (standardized) largest eigenvalue of symmetrized lag-$1$ sample autocovariance matrix as defined in Proposition \ref{1:t3}, to compare bootstrap samples from original observations with bootstrap samples from pre-smoothed observations. 
\begin{figure}[!htbp]
\centering
\includegraphics[width=0.9\textwidth]{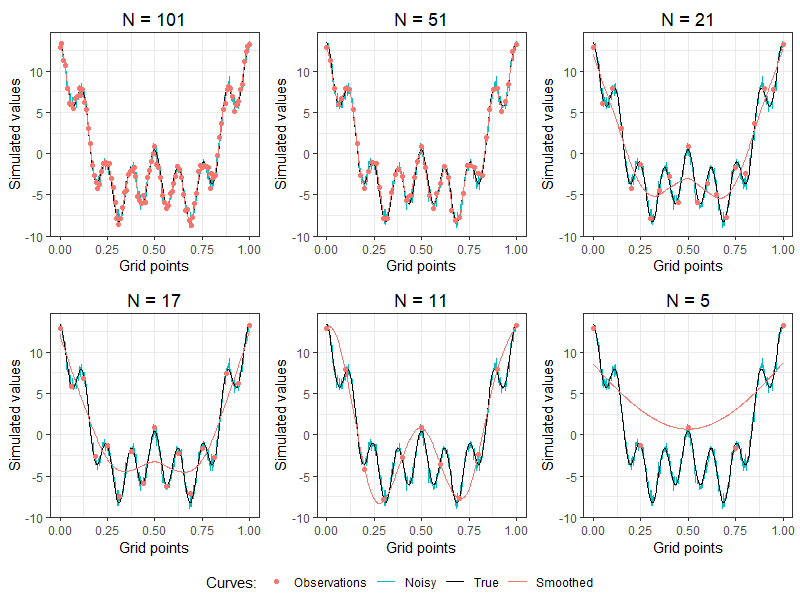}
\caption{Example of smoothing errors on sparse functional observations.}\label{1:f2}
\end{figure}

Figures~\ref{1:f5} and~\ref{1:f6} compare the histograms and boxplots of $\delta_1^\ast$, the AR-sieve bootstrap estimates of the largest eigenvalue of the symmetrized lag-$1$ autocovariance matrix, while Figures~\ref{1:f7} and~\ref{1:f8} compare the histograms and boxplots of $\overline{\overline{y^\ast}}$, the AR-sieve bootstrap estimates of the (standardized) mean statistic. As seen in Figure \ref{1:f2}, when $N=21,17$ and $5$, the pre-smoothed observations are averaged out compared with the original observations. As a result, the bootstrap estimates of the two statistics perform differently before and after smoothing, when $N=21,17$ and $5$. Figures~\ref{1:f5} and~\ref{1:f7} use boxplots to present the difference in empirical distributions of $\overline{\overline{y^\ast}}$ and $\delta_1^\ast$ for $N=21,17$ and $5$, whereas Figures~\ref{1:f6} and~\ref{1:f8} illustrate the impact of smoothing by comparing the histograms of $\overline{\overline{y^\ast}}$ and $\delta_{1}^{\ast}$.
\begin{figure}[!htb]
\centering
\includegraphics[width=0.76\textwidth]{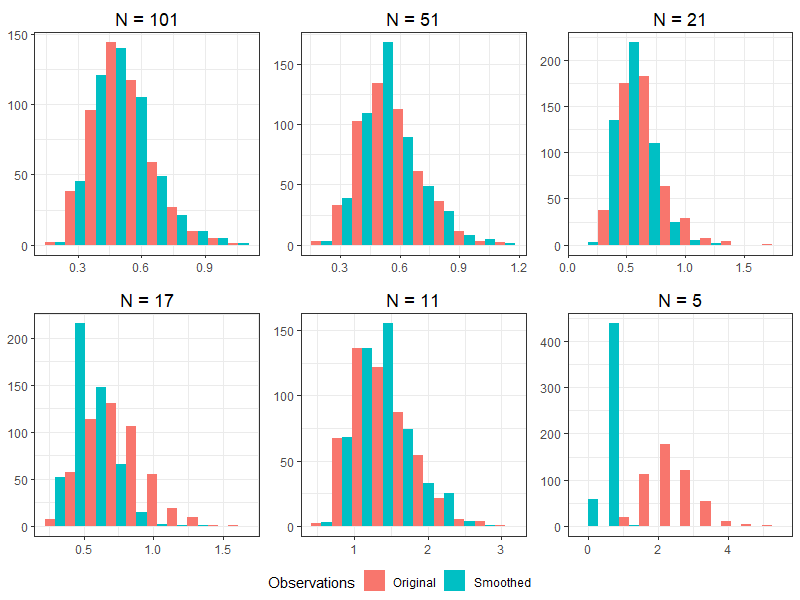}
\caption{Histograms of $\delta_1^\ast$, the AR-sieve bootstrap estimates of the largest eigenvalue of symmetrized lag-$1$ sample autocovariance matrix.}\label{1:f5}
\end{figure}

\begin{figure}[!htb]
\centering
\includegraphics[width=0.76\textwidth]{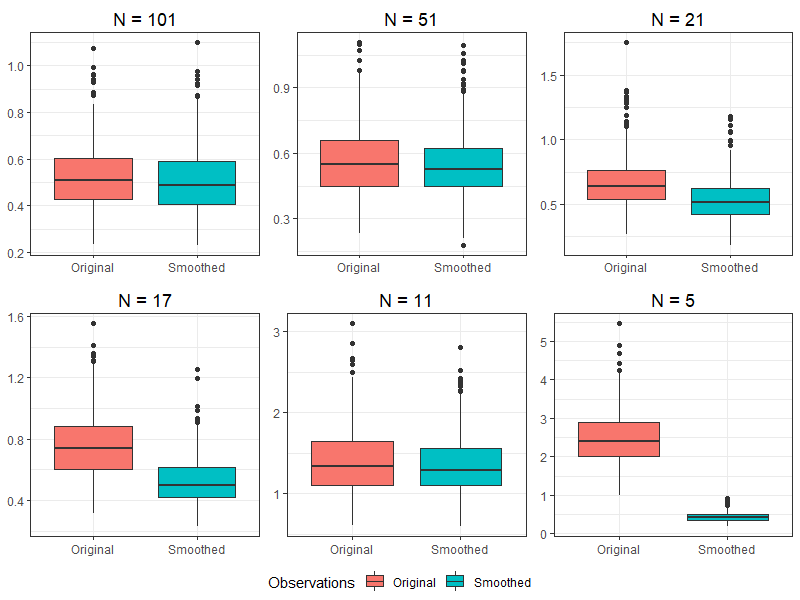}
\caption{Boxplots of $\delta_1^\ast$, the AR-sieve bootstrap estimates of the largest eigenvalue of symmetrized lag-$1$ sample autocovariance matrix.}\label{1:f6}
\end{figure}

\begin{figure}[!htbp]
\centering
\includegraphics[width=0.8\textwidth]{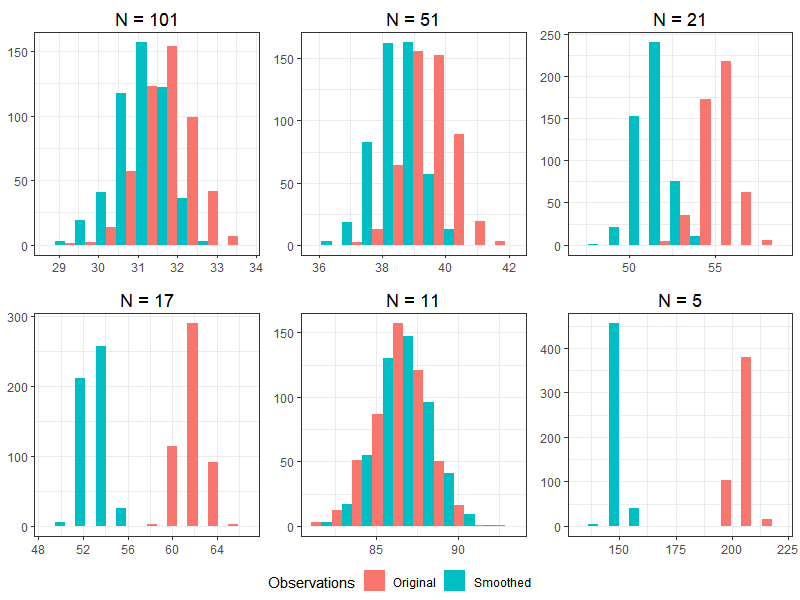}
\caption{Histograms of $\overline{\overline{y^\ast}}$, the AR-sieve bootstrap estimates of the mean statistic.}\label{1:f7}
\end{figure}

\begin{figure}[!htbp]
\centering
\includegraphics[width=0.8\textwidth]{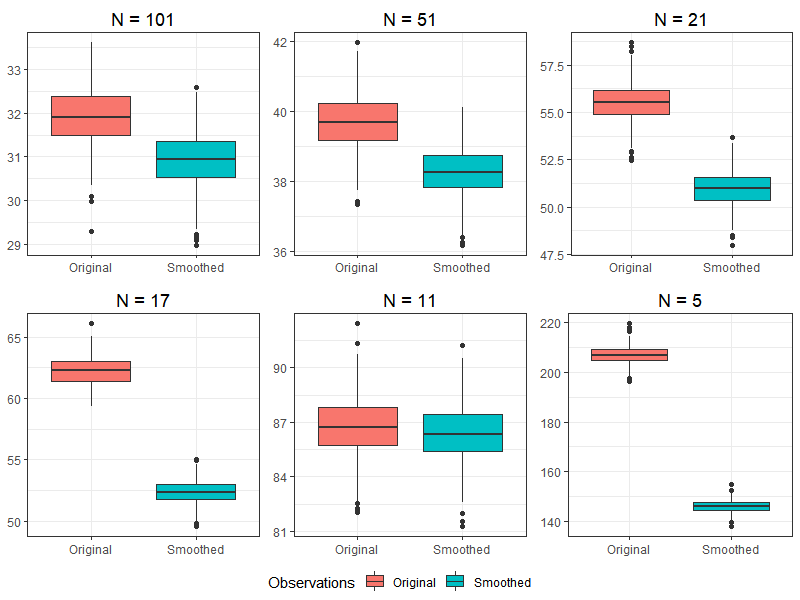}
\caption{Boxplots of $\overline{\overline{y^\ast}}$, the AR-sieve bootstrap estimates of the mean statistic.}\label{1:f8}
\end{figure}

\begin{figure}[!htbp]
\centering
\includegraphics[width=0.8\textwidth]{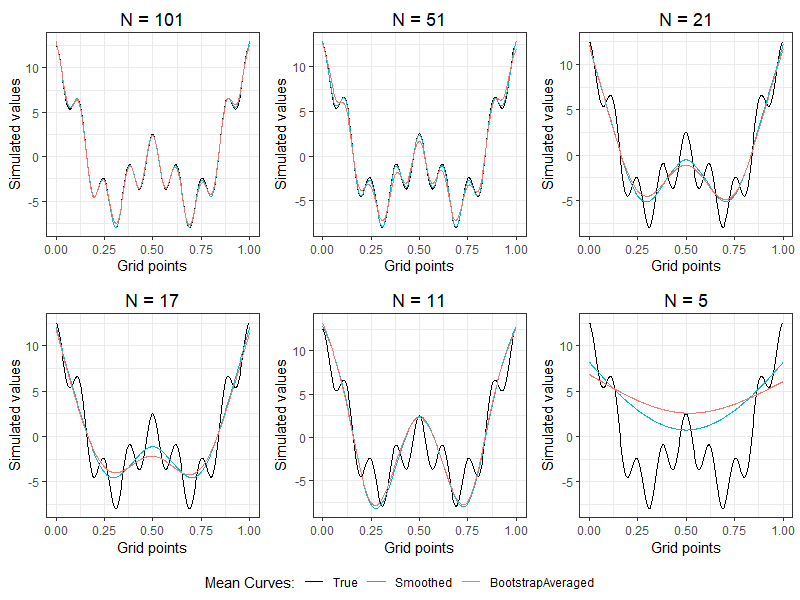}
\caption{Example of errors of the AR-sieve bootstrap mean curve for sparse functional observations.}\label{1:f9}
\end{figure} 

The last example we presented in Figure~\ref{1:f9} illustrates the results of AR-sieve bootstrap estimates (bootstrap average) of the functional mean curve when we pre-smooth the observations under various densities of data. As shown in Figure~\ref{1:f9}, when the actual observations are relatively dense, for example, $N \ge 51$, AR-sieve bootstrap estimates of the mean functional curve are close to the pre-smoothed curve and the population curve. However, when observations are sparse, for example, $N \le 21$, AR-sieve bootstrap estimates of the mean functional curve do not correctly capture the local patterns of the population curve, which is due to unacceptable smoothing results. This result is also typical evidence of the impact of pre-smoothing on the AR-sieve bootstrap for functional time series. Hence, when the actual functional time series observations are sparse, pre-smoothing can significantly impact statistical inferences, including those based on the bootstrap. In fact, for many real-world time series data, the rule on considering a data set as a dense functional time series is generally not clear and often varies across researchers and problems. Practically speaking, the impact of observations' density is only about whether to pre-smooth the functional time series before performing bootstrap or other statistical analysis. 

%Nonetheless, the theoretical assumptions behind functional time series and high-dimensional time series vary, leading to very different theoretical results on statistical inferences, including AR-sieve bootstrap. On the other hand, this difference in data structure assumptions demonstrates the importance of developing statistical methods on sparse functional time series observations. It verifies our contributions to the building blocks of AR-sieve bootstrap for high-dimensional time series. 

\newpage
\commHS{
\subsection{Simulation studies}\label{app.f.2}

To further investigate the impact of pre-smoothing observed in Appendix~\ref{app.f.1}, we conducted a simulation study comparing the proposed AR-sieve bootstrap (treating data as high-dimensional vectors) with the functional AR-sieve bootstrap (using B-spline pre-smoothing).

The simulation utilizes the same DGP as described in Appendix~\ref{app.f.1}, which features factors with high-frequency local patterns generated by Fourier bases. However, to better isolate the bias introduced by smoothing from the variance caused by large noise, we set the standard deviation of the error term $u_t$ to $\sigma = 0.1$. We evaluated the performance of bootstrap confidence intervals (CIs) across a range of observation densities, specifically $N \in \{11, 17, 21, 51, 101, 201\}$.

\begin{table}[!htbp]
	\centering
	\caption{Empirical coverage, average width, and interval score of bootstrap intervals constructed using the functional AR-sieve bootstrap (pre-smoothed) for $\theta_{y}$.}
	\label{tab:mean_functional}
	\resizebox{\textwidth}{!}{%
		\begin{tabular}{|ccccccccccc|}
			\hline
             &  & \multicolumn{3}{c}{95\%} & \multicolumn{3}{c}{90\%} & \multicolumn{3}{c|}{80\%} \\ \hline
            T & N & \begin{tabular}[c]{@{}c@{}}Empirical\\ coverage\end{tabular} & \begin{tabular}[c]{@{}c@{}}Average\\ width\end{tabular} & \begin{tabular}[c]{@{}c@{}}Average\\ interval score\end{tabular} & \begin{tabular}[c]{@{}c@{}}Empirical\\ coverage\end{tabular} & \begin{tabular}[c]{@{}c@{}}Average\\ width\end{tabular} & \begin{tabular}[c]{@{}c@{}}Average\\ interval score\end{tabular} & \begin{tabular}[c]{@{}c@{}}Empirical\\ coverage\end{tabular} & \begin{tabular}[c]{@{}c@{}}Average\\ width\end{tabular} & \begin{tabular}[c]{@{}c@{}}Average\\ interval score\end{tabular} \\ \hline
            \multicolumn{11}{|c|}{Nonparametric bootstrap intervals using quantiles} \\ \hline
            \multirow{6}{*}{100} & 11 & 0.878 & 5.244 & 12.380 & 0.798 & 4.417 & 10.216 & 0.684 & 3.450 & 8.599 \\
             & 17 & 0.858 & 4.059 & 9.794 & 0.796 & 3.418 & 8.266 & 0.680 & 2.666 & 6.860 \\
             & 21 & 0.872 & 3.618 & 8.911 & 0.788 & 3.048 & 7.380 & 0.676 & 2.381 & 6.112 \\
             & 51 & 0.982 & 3.529 & 3.824 & 0.954 & 2.972 & 3.376 & 0.884 & 2.323 & 2.934 \\
             & 101 & 0.990 & 3.021 & 3.155 & 0.982 & 2.543 & 2.661 & 0.948 & 1.986 & 2.154 \\
             & 201 & 0.984 & 2.683 & 2.773 & 0.984 & 2.260 & 2.336 & 0.976 & 1.761 & 1.824 \\ \hline
            \multicolumn{11}{|c|}{Parametric bootstrap intervals based on normality} \\ \hline
            \multirow{6}{*}{100} & 11 & 0.876 & 5.297 & 12.170 & 0.804 & 4.446 & 10.199 & 0.680 & 3.464 & 8.598 \\
             & 17 & 0.874 & 4.096 & 9.671 & 0.800 & 3.438 & 8.199 & 0.680 & 2.678 & 6.806 \\
             & 21 & 0.880 & 3.655 & 8.700 & 0.794 & 3.067 & 7.260 & 0.676 & 2.390 & 6.077 \\
             & 51 & 0.984 & 3.559 & 3.811 & 0.956 & 2.987 & 3.382 & 0.880 & 2.327 & 2.926 \\
             & 101 & 0.988 & 3.046 & 3.151 & 0.984 & 2.556 & 2.667 & 0.954 & 1.992 & 2.153 \\
             & 201 & 0.984 & 2.705 & 2.785 & 0.984 & 2.270 & 2.343 & 0.978 & 1.769 & 1.829 \\ \hline
		\end{tabular}%
	}
\end{table}

\begin{table}[!htbp]
	\centering
	\caption{Empirical coverage, average width, and interval score of bootstrap intervals constructed using the proposed high-dimensional AR-sieve bootstrap (no smoothing) for $\theta_{y}$.}
	\label{tab:mean_highdim}
	\resizebox{\textwidth}{!}{%
		\begin{tabular}{|ccccccccccc|}
			\hline
             &  & \multicolumn{3}{c}{95\%} & \multicolumn{3}{c}{90\%} & \multicolumn{3}{c|}{80\%} \\ \hline
            T & N & \begin{tabular}[c]{@{}c@{}}Empirical\\ coverage\end{tabular} & \begin{tabular}[c]{@{}c@{}}Average\\ width\end{tabular} & \begin{tabular}[c]{@{}c@{}}Average\\ interval score\end{tabular} & \begin{tabular}[c]{@{}c@{}}Empirical\\ coverage\end{tabular} & \begin{tabular}[c]{@{}c@{}}Average\\ width\end{tabular} & \begin{tabular}[c]{@{}c@{}}Average\\ interval score\end{tabular} & \begin{tabular}[c]{@{}c@{}}Empirical\\ coverage\end{tabular} & \begin{tabular}[c]{@{}c@{}}Average\\ width\end{tabular} & \begin{tabular}[c]{@{}c@{}}Average\\ interval score\end{tabular} \\ \hline
            \multicolumn{11}{|c|}{Nonparametric bootstrap intervals using quantiles} \\ \hline
            \multirow{6}{*}{100} & 11 & 0.860 & 13.474 & 36.086 & 0.792 & 11.347 & 29.281 & 0.666 & 8.850 & 23.717 \\
             & 17 & 0.878 & 5.247 & 12.339 & 0.794 & 4.420 & 10.204 & 0.684 & 3.453 & 8.590 \\
             & 21 & 0.964 & 5.464 & 6.846 & 0.920 & 4.587 & 6.132 & 0.812 & 3.584 & 5.550 \\
             & 51 & 0.948 & 4.617 & 6.514 & 0.900 & 3.881 & 5.724 & 0.800 & 3.038 & 5.123 \\
             & 101 & 0.978 & 3.504 & 3.923 & 0.952 & 2.952 & 3.421 & 0.874 & 2.307 & 2.958 \\
             & 201 & 0.914 & 1.844 & 2.991 & 0.858 & 1.554 & 2.592 & 0.794 & 1.214 & 2.162 \\ \hline
            \multicolumn{11}{|c|}{Parametric bootstrap intervals based on normality} \\ \hline
            \multirow{6}{*}{100} & 11 & 0.868 & 13.598 & 35.734 & 0.796 & 11.412 & 28.930 & 0.678 & 8.891 & 23.421 \\
             & 17 & 0.878 & 5.301 & 12.143 & 0.804 & 4.449 & 10.185 & 0.680 & 3.466 & 8.589 \\
             & 21 & 0.966 & 5.509 & 6.852 & 0.922 & 4.623 & 6.102 & 0.816 & 3.602 & 5.531 \\
             & 51 & 0.948 & 4.659 & 6.452 & 0.906 & 3.910 & 5.684 & 0.794 & 3.046 & 5.102 \\
             & 101 & 0.980 & 3.535 & 3.918 & 0.952 & 2.967 & 3.422 & 0.876 & 2.312 & 2.950 \\
             & 201 & 0.914 & 1.862 & 2.935 & 0.858 & 1.563 & 2.574 & 0.786 & 1.218 & 2.152 \\ \hline
		\end{tabular}%
	}
\end{table}

We calculated the empirical coverage, average width, and interval score for the (standardized) mean statistic. The results (Table~\ref{tab:mean_functional}~and~\ref{tab:mean_highdim}) reveal a distinct trade-off governed by observation density. In the sparse regime ($N \le 21$), the bootstrap confidence intervals are undercovered, where pre-smoothing significantly distorts the signal when observations are sparse. For $N=21$, the functional approach yields coverage probabilities notably lower than the nominal level and lower than those of the proposed high-dimensional method. This under-coverage confirms that pre-smoothing averages out critical local patterns, introducing bias that the bootstrap cannot recover. In the dense regime ($N \ge 51$), as $N$ increases, the smoothing bias diminishes. However, we observe that the functional approach tends to produce higher coverage probabilities (over-coverage) compared to the proposed method. This suggests that even with denser grids, the imposed smoothness structure may overestimate variability when the underlying signal retains persistent local irregularities. In contrast, the proposed high-dimensional method adapts well to the increasing dimension, maintaining stable coverage.

\begin{table}[!htbp]
	\centering
	\caption{Empirical coverage, average width, and interval score of bootstrap intervals constructed using the functional AR-sieve bootstrap (pre-smoothed) for $\delta_{1}^0$.}
	\label{tab:eigen_functional}
	\resizebox{\textwidth}{!}{%
		\begin{tabular}{|ccccccccccc|}
			\hline
             &  & \multicolumn{3}{c}{95\%} & \multicolumn{3}{c}{90\%} & \multicolumn{3}{c|}{80\%} \\ \hline
            T & N & \begin{tabular}[c]{@{}c@{}}Empirical\\ coverage\end{tabular} & \begin{tabular}[c]{@{}c@{}}Average\\ width\end{tabular} & \begin{tabular}[c]{@{}c@{}}Average\\ interval score\end{tabular} & \begin{tabular}[c]{@{}c@{}}Empirical\\ coverage\end{tabular} & \begin{tabular}[c]{@{}c@{}}Average\\ width\end{tabular} & \begin{tabular}[c]{@{}c@{}}Average\\ interval score\end{tabular} & \begin{tabular}[c]{@{}c@{}}Empirical\\ coverage\end{tabular} & \begin{tabular}[c]{@{}c@{}}Average\\ width\end{tabular} & \begin{tabular}[c]{@{}c@{}}Average\\ interval score\end{tabular} \\ \hline
            \multicolumn{11}{|c|}{Nonparametric bootstrap intervals using quantiles} \\ \hline
            \multirow{6}{*}{100} & 11 & 0.098 & 10.829 & 671.030 & 0.090 & 8.762 & 347.923 & 0.080 & 6.555 & 182.500 \\
             & 17 & 0.936 & 35.115 & 48.712 & 0.930 & 28.274 & 35.710 & 0.790 & 21.259 & 34.235 \\
             & 21 & 0.756 & 12.622 & 43.207 & 0.746 & 10.200 & 26.894 & 0.692 & 7.645 & 17.800 \\
             & 51 & 0.888 & 13.919 & 25.456 & 0.878 & 11.250 & 17.773 & 0.796 & 8.435 & 13.650 \\
             & 101 & 0.936 & 14.163 & 19.935 & 0.928 & 11.452 & 14.996 & 0.816 & 8.598 & 18.237 \\
             & 201 & 0.938 & 13.116 & 18.308 & 0.932 & 10.626 & 13.675 & 0.826 & 8.004 & 17.050 \\ \hline
            \multicolumn{11}{|c|}{Parametric bootstrap intervals based on normality} \\ \hline
            \multirow{6}{*}{100} & 11 & 0.156 & 11.302 & 558.569 & 0.122 & 9.485 & 307.439 & 0.098 & 7.390 & 171.521 \\
             & 17 & 0.954 & 36.552 & 44.494 & 0.936 & 30.675 & 36.208 & 0.808 & 23.900 & 34.889 \\
             & 21 & 0.846 & 13.184 & 32.453 & 0.814 & 11.065 & 23.067 & 0.732 & 8.621 & 16.833 \\
             & 51 & 0.928 & 14.530 & 22.105 & 0.916 & 12.194 & 17.011 & 0.834 & 9.500 & 13.558 \\
             & 101 & 0.954 & 14.725 & 18.638 & 0.944 & 12.358 & 15.019 & 0.834 & 9.628 & 13.323 \\
             & 201 & 0.952 & 13.706 & 16.977 & 0.942 & 11.503 & 13.823 & 0.842 & 8.962 & 12.196 \\ \hline
		\end{tabular}%
	}
\end{table}

\begin{table}[!htbp]
	\centering
	\caption{Empirical coverage, average width, and interval score of bootstrap intervals constructed using the proposed high-dimensional AR-sieve bootstrap (no smoothing) for $\delta_{1}^0$.}
	\label{tab:eigen_highdim}
	\resizebox{\textwidth}{!}{%
		\begin{tabular}{|ccccccccccc|}
			\hline
             &  & \multicolumn{3}{c}{95\%} & \multicolumn{3}{c}{90\%} & \multicolumn{3}{c|}{80\%} \\ \hline
            T & N & \begin{tabular}[c]{@{}c@{}}Empirical\\ coverage\end{tabular} & \begin{tabular}[c]{@{}c@{}}Average\\ width\end{tabular} & \begin{tabular}[c]{@{}c@{}}Average\\ interval score\end{tabular} & \begin{tabular}[c]{@{}c@{}}Empirical\\ coverage\end{tabular} & \begin{tabular}[c]{@{}c@{}}Average\\ width\end{tabular} & \begin{tabular}[c]{@{}c@{}}Average\\ interval score\end{tabular} & \begin{tabular}[c]{@{}c@{}}Empirical\\ coverage\end{tabular} & \begin{tabular}[c]{@{}c@{}}Average\\ width\end{tabular} & \begin{tabular}[c]{@{}c@{}}Average\\ interval score\end{tabular} \\ \hline
            \multicolumn{11}{|c|}{Nonparametric bootstrap intervals using quantiles} \\ \hline
            \multirow{6}{*}{100} & 11 & 0.936 & 62.553 & 89.998 & 0.926 & 50.730 & 66.696 & 0.808 & 38.180 & 60.948 \\
             & 17 & 0.938 & 35.285 & 48.726 & 0.932 & 28.411 & 35.760 & 0.790 & 21.363 & 34.438 \\
             & 21 & 0.936 & 20.980 & 28.825 & 0.930 & 16.986 & 21.775 & 0.818 & 12.751 & 18.934 \\
             & 51 & 0.930 & 17.389 & 24.650 & 0.924 & 14.051 & 18.372 & 0.818 & 10.533 & 16.004 \\
             & 101 & 0.936 & 14.208 & 19.996 & 0.930 & 11.487 & 15.023 & 0.830 & 8.622 & 12.707 \\
             & 201 & 0.934 & 13.184 & 18.250 & 0.932 & 10.650 & 13.730 & 0.838 & 8.040 & 11.594 \\ \hline
            \multicolumn{11}{|c|}{Parametric bootstrap intervals based on normality} \\ \hline
            \multirow{6}{*}{100} & 11 & 0.954 & 65.438 & 84.349 & 0.932 & 54.918 & 68.195 & 0.828 & 42.788 & 63.939 \\
             & 17 & 0.954 & 36.729 & 44.549 & 0.936 & 30.824 & 36.305 & 0.804 & 24.016 & 35.111 \\
             & 21 & 0.958 & 21.869 & 26.988 & 0.948 & 18.353 & 22.021 & 0.852 & 14.299 & 19.669 \\
             & 51 & 0.952 & 18.131 & 23.162 & 0.948 & 15.216 & 18.546 & 0.848 & 11.855 & 16.215 \\
             & 101 & 0.954 & 14.769 & 18.677 & 0.944 & 12.395 & 15.051 & 0.838 & 9.657 & 13.335 \\
             & 201 & 0.952 & 13.768 & 17.145 & 0.944 & 11.555 & 13.885 & 0.842 & 9.003 & 12.343 \\ \hline
		\end{tabular}%
	}
\end{table}

The analysis of the largest eigenvalue (Table~\ref{tab:eigen_functional}~and~\ref{tab:eigen_highdim}) of the symmetrized lag-1 autocovariance matrix yields similar insights. In the sparse regime ($N \le 21$), pre-smoothing alters the spectral density of the data, resulting in bootstrap CIs with poor coverage. As the grid becomes denser ($N \ge 51$), the pre-smoothing effect on the global temporal dependence structure weakens, and the performance of the functional approach converges toward that of the unsmoothed high-dimensional approach. This indicates that while second-order properties are more robust to smoothing than first-order means, they are still vulnerable in sparse settings.

To provide a comprehensive evaluation, we also examined the impact of pre-smoothing in scenarios where the underlying factor structure is relatively smooth, lacking the high-frequency local patterns used in Appendix~\ref{app.f.2}. We revisited the simulation setup described in Section~\ref{simu2} with strong factors. We applied the functional AR-sieve bootstrap (with pre-smoothing) to this dataset and compared the mean statistics, shown in Table~\ref{1:ta1pre}, with those obtained from the proposed high-dimensional approach (without smoothing), as shown in Table~\ref{1:ta1a} in the main text.

\begin{table}[!htb]
	\centering
	\caption{Empirical coverage, average width, and interval score of bootstrap intervals constructed using the functional AR-sieve bootstrap (pre-smoothed) for $\theta_{y}$ ($\nu=1$) under the DGP considered in the main paper.}
	\label{1:ta1pre}
	\resizebox{\textwidth}{!}{%
		\begin{tabular}{|ccccccccccc|}
			\hline
             &  & \multicolumn{3}{c}{95\%} & \multicolumn{3}{c}{90\%} & \multicolumn{3}{c|}{80\%} \\ \hline
            T & N & \begin{tabular}[c]{@{}c@{}}Empirical\\ coverage\end{tabular} & \begin{tabular}[c]{@{}c@{}}Average\\ width\end{tabular} & \begin{tabular}[c]{@{}c@{}}Average\\ interval score\end{tabular} & \begin{tabular}[c]{@{}c@{}}Empirical\\ coverage\end{tabular} & \begin{tabular}[c]{@{}c@{}}Average\\ width\end{tabular} & \begin{tabular}[c]{@{}c@{}}Average\\ interval score\end{tabular} & \begin{tabular}[c]{@{}c@{}}Empirical\\ coverage\end{tabular} & \begin{tabular}[c]{@{}c@{}}Average\\ width\end{tabular} & \begin{tabular}[c]{@{}c@{}}Average\\ interval score\end{tabular} \\ \hline
            \multicolumn{11}{|c|}{Nonparametric bootstrap intervals using quantiles} \\ \hline
            \multirow{5}{*}{200} & 50 & $0.938$ & $9.076$ & $13.869$ & $0.878$ & $7.633$ & $12.285$ & $0.778$ & $5.942$ & $10.892$ \\
             & 100 & $0.947$ & $9.137$ & $13.842$ & $0.901$ & $7.683$ & $12.294$ & $0.801$ & $5.994$ & $10.704$ \\
             & 200 & $0.931$ & $9.381$ & $15.385$ & $0.875$ & $7.878$ & $13.645$ & $0.770$ & $6.143$ & $11.761$ \\
             & 500 & $0.928$ & $9.421$ & $16.573$ & $0.876$ & $7.926$ & $14.318$ & $0.784$ & $6.182$ & $12.276$ \\
             & 1000 & $0.948$ & $9.368$ & $13.834$ & $0.899$ & $7.874$ & $12.349$ & $0.821$ & $6.136$ & $10.741$ \\
             &  &  &  &  &  &  &  &  &  &  \\
            \multirow{5}{*}{500} & 50 & $0.946$ & $9.313$ & $13.339$ & $0.905$ & $7.823$ & $11.788$ & $0.787$ & $6.102$ & $10.544$ \\
             & 100 & $0.937$ & $9.366$ & $14.124$ & $0.882$ & $7.880$ & $12.592$ & $0.779$ & $6.149$ & $11.290$ \\
             & 200 & $0.937$ & $9.339$ & $14.175$ & $0.887$ & $7.845$ & $12.727$ & $0.790$ & $6.121$ & $11.291$ \\
             & 500 & $0.941$ & $9.195$ & $14.497$ & $0.882$ & $7.734$ & $12.803$ & $0.778$ & $6.029$ & $11.321$ \\
             & 1000 & $0.938$ & $9.496$ & $14.501$ & $0.893$ & $7.993$ & $12.756$ & $0.796$ & $6.235$ & $11.267$ \\
             &  &  &  &  &  &  &  &  &  &  \\
            \multirow{5}{*}{1000} & 50 & $0.945$ & $9.415$ & $13.606$ & $0.886$ & $7.914$ & $12.363$ & $0.790$ & $6.173$ & $11.207$ \\
             & 100 & $0.938$ & $9.393$ & $14.646$ & $0.883$ & $7.889$ & $13.050$ & $0.769$ & $6.154$ & $11.642$ \\
             & 200 & $0.935$ & $9.395$ & $14.975$ & $0.874$ & $7.891$ & $13.355$ & $0.772$ & $6.148$ & $11.825$ \\
             & 500 & $0.941$ & $9.556$ & $14.563$ & $0.888$ & $8.036$ & $12.896$ & $0.784$ & $6.268$ & $11.539$ \\
             & 1000 & $0.953$ & $9.161$ & $12.377$ & $0.913$ & $7.702$ & $11.076$ & $0.817$ & $6.006$ & $9.828$ \\ \hline
            \multicolumn{11}{|c|}{Parametric bootstrap intervals based on normality} \\ \hline
            \multirow{5}{*}{200} & 50 & $0.938$ & $9.109$ & $13.873$ & $0.879$ & $7.644$ & $12.293$ & $0.775$ & $5.956$ & $10.893$ \\
             & 100 & $0.950$ & $9.185$ & $13.710$ & $0.898$ & $7.708$ & $12.229$ & $0.804$ & $6.006$ & $10.694$ \\
             & 200 & $0.934$ & $9.414$ & $15.347$ & $0.874$ & $7.901$ & $13.589$ & $0.775$ & $6.156$ & $11.747$ \\
             & 500 & $0.928$ & $9.475$ & $16.539$ & $0.881$ & $7.952$ & $14.318$ & $0.783$ & $6.196$ & $12.242$ \\
             & 1000 & $0.949$ & $9.407$ & $13.762$ & $0.902$ & $7.895$ & $12.311$ & $0.820$ & $6.151$ & $10.753$ \\
             &  &  &  &  &  &  &  &  &  &  \\
            \multirow{5}{*}{500} & 50 & $0.948$ & $9.347$ & $13.173$ & $0.908$ & $7.844$ & $11.775$ & $0.789$ & $6.112$ & $10.550$ \\
             & 100 & $0.939$ & $9.410$ & $13.924$ & $0.892$ & $7.897$ & $12.475$ & $0.781$ & $6.153$ & $11.262$ \\
             & 200 & $0.944$ & $9.377$ & $14.128$ & $0.885$ & $7.870$ & $12.653$ & $0.787$ & $6.131$ & $11.274$ \\
             & 500 & $0.941$ & $9.234$ & $14.398$ & $0.882$ & $7.750$ & $12.777$ & $0.779$ & $6.038$ & $11.310$ \\
             & 1000 & $0.942$ & $9.553$ & $14.383$ & $0.896$ & $8.017$ & $12.662$ & $0.795$ & $6.246$ & $11.244$ \\
             &  &  &  &  &  &  &  &  &  &  \\
            \multirow{5}{*}{1000} & 50 & $0.946$ & $9.456$ & $13.543$ & $0.889$ & $7.936$ & $12.263$ & $0.794$ & $6.183$ & $11.163$ \\
             & 100 & $0.939$ & $9.435$ & $14.682$ & $0.884$ & $7.918$ & $12.944$ & $0.776$ & $6.169$ & $11.607$ \\
             & 200 & $0.934$ & $9.424$ & $14.820$ & $0.874$ & $7.909$ & $13.276$ & $0.774$ & $6.162$ & $11.825$ \\
             & 500 & $0.944$ & $9.600$ & $14.337$ & $0.892$ & $8.057$ & $12.855$ & $0.786$ & $6.277$ & $11.517$ \\
             & 1000 & $0.956$ & $9.201$ & $12.315$ & $0.911$ & $7.722$ & $11.058$ & $0.814$ & $6.016$ & $9.789$ \\ \hline
		\end{tabular}%
	}
\end{table}

Unlike the sparse case with local patterns, the simulation results under this smooth setting show negligible differences between the two approaches. The empirical coverage probabilities for both the pre-smoothed and unsmoothed methods are very close to the nominal levels across various dimensions ($N$). This suggests that when the underlying functional curve is smooth and well-behaved, pre-smoothing does not introduce significant bias, nor does it substantially improve the inference compared to the high-dimensional approach. This additional comparison highlights that the proposed high-dimensional AR-sieve bootstrap is a robust methodology. It matches the performance of functional methods in standard smooth settings while offering superior accuracy in sparse settings with complex local patterns.

The simulation results, combined with the visual illustrations in Appendix~\ref{app.f.1}, demonstrate that treating sparse functional observations as high-dimensional time series is advantageous when local patterns are present. In practice, the distinction between dense functional data and sparse high-dimensional data is often ambiguous. Nonetheless, the theoretical assumptions behind functional time series and high-dimensional time series vary, leading to different implications for statistical inference. Our results highlight that when data sparsity prevents accurate pre-smoothing, the proposed high-dimensional AR-sieve bootstrap offers a robust alternative. This validates our contribution to developing building blocks for AR-sieve bootstrap in high-dimensional settings, bridging the gap where functional methods may falter.
}
% \end{appendices}

%\newpage
%\bibliographystyle{agsm}
%%\bibliographystyle{abbrvnatmod}%Choose a bibliograhpic style
%\bibliography{reference1}

\end{appendices}
\end{document}